\documentclass{article}
\PassOptionsToPackage{numbers,compress}{natbib}
\usepackage[preprint]{skf}
\usepackage[utf8]{inputenc}
\usepackage[T1]{fontenc}
\usepackage[colorlinks=true,linkcolor=blue,citecolor=blue,urlcolor=blue,hypertexnames=false]{hyperref}
\usepackage{url}
\usepackage{booktabs}
\usepackage{amsfonts,amsmath,amssymb,amsthm}
\usepackage{nicefrac}
\usepackage{microtype}
\usepackage{xcolor}
\usepackage{graphicx}
\usepackage{enumitem}
\usepackage{algorithm}
\usepackage{algorithmic}
\usepackage{chngcntr}
\usepackage{placeins}

\newcommand{\R}{\mathbb{R}}
\newcommand{\N}{\mathbb{N}}
\newcommand{\E}{\mathbb{E}}
\newcommand{\Tr}{\operatorname{Tr}}
\newtheorem{proposition}{Proposition}
\newtheorem{theorem}{Theorem}
\theoremstyle{definition}
\newtheorem{remark}{Remark}

\title{The Score Kalman Filter}

\author{%
  Kaito Iwasaki \\
  Department of Mathematics \\
  University of Michigan \\
  \texttt{kaitoi@umich.edu} \\
  \And
  Anthony Bloch \\
  Department of Mathematics \\
  University of Michigan \\
  \texttt{abloch@umich.edu} \\
  \And
  Taeyoung Lee \\
  Department of Mechanical and \\
  Aerospace Engineering \\
  George Washington University \\
  \texttt{tylee@gwu.edu} \\
  \And
  Maani Ghaffari \\
  Department of Naval Architecture \\
  \& Marine Engineering and \\
  Department of Robotics \\
  University of Michigan \\
  \texttt{maanigj@umich.edu} \\
}

\makeatletter
\renewcommand{\@noticestring}{}
\renewcommand{\@notice}{}
\makeatother

\begin{document}
\maketitle

\begin{abstract}
A central obstacle in nonlinear Bayesian filtering is representing the belief distribution. Moment-based filters address this by propagating polynomial moments and reconstructing a density from them. Recent work completes the predict-update loop via the maximum-entropy (MaxEnt) principle, but each step requires the partition function and its gradient, both $n$-dimensional integrals whose cost scales exponentially, restricting the demonstrated MaxEnt moment filtering to $n \le 4$. We avoid the partition function entirely by combining score matching with Stein's identity. In our setting, score matching reduces the density fit to a single linear solve whose coefficients are assembled directly from the propagated moments. The same parameters then drive Stein's identity to close the moment hierarchy during prediction and to recover posterior moments after each Bayesian update, keeping the full predict-update loop free of partition function evaluation. The resulting \textbf{Score Kalman Filter} (SKF) reduces to the classical information-form Kalman filter as a special case and performs every step through linear algebra. On nonlinear coupled-oscillator networks, the SKF runs through $n=20$ and reports lower RMSE than the EKF, UKF, EnKF, and particle-filter baselines on the tested synthetic benchmarks.
\end{abstract}

\section{Introduction}

Bayesian filtering, the recursive estimation of the state of a dynamical system
from noisy observations, is fundamental to robotics, control, and
signal processing.
The Kalman filter~\cite{kalman1960} solves this problem optimally for
linear-Gaussian systems, but many real-world dynamical systems are nonlinear and produce
non-Gaussian beliefs: skewed, curved distributions from nonlinear dynamics, and multimodal distributions from systems with multiple stable equilibria.
Gaussian filters (the extended Kalman filter
(EKF~\cite{schmidt1966,jazwinski1970}), unscented Kalman filter
(UKF~\cite{julier1997,wan2000}), ensemble Kalman filter
(EnKF~\cite{evensen1994,burgers1998}), and their Lie-group variants
\cite{barrau2017,brossard2020}) are not designed to represent these
density shapes directly.
Particle filters can, but their convergence rate in the number of
particles makes high-accuracy filtering expensive.

A promising middle ground is \emph{moment-based filtering}, which
propagates the moments of the belief density via the system's generator
(Dynkin's formula) and reconstructs the density from those moments.
Teng et al.~\cite{teng2025} recently showed that polynomial MaxEnt
distributions $p(x;\lambda) \propto \exp(-\lambda \cdot \phi(x))$
give a flexible, multimodal representation, demonstrated on SE(2)
localization. Here $\phi(x)$ collects monomials of total degree at most
$r$ and $\lambda \in \R^M$ with $M = \binom{n+r}{n}$ is the
corresponding parameter vector.
The main difficulty in this approach is therefore its computational cost. Fitting $\lambda$ to predicted
moments and recovering posterior moments after a Bayesian update require both
the \emph{partition function}
$Z(\lambda) = \int \exp(-\lambda\cdot\phi(x))\,dx$ and its gradients,
$n$-dimensional quadratures whose cost scales as $O(G^n)$ and becomes
intractable as $n$ grows.

\textbf{Our key observation.}
Score matching~\cite{hyvarinen2005} is a general method for estimating
parameters of non-normalized statistical models by minimizing the
Fisher divergence between the model score and the data score, without
evaluating the normalization constant.
We apply this idea to the polynomial exponential family and show that
the estimator reduces to a \emph{single linear system} on the non-constant
coefficients, with entries assembled directly from the propagated moments
(Proposition~\ref{prop:sm}).
The resulting fit does not require partition function evaluation,
iterative optimization, or numerical integration.
The density-fit cost is $O(M^3)$, requiring only linear algebra.

Two additional ingredients are needed to turn score matching into a
complete filtering algorithm.
First, for nonlinear systems whose moment hierarchy does not close~\cite{kuehn2016}, a \emph{moment closure} is required.
Classical closures either assume Gaussianity~\cite{whittle1957}
or require the partition function~\cite{levermore1996,singh2011}.
We instead exploit the \emph{Stein identity}~\cite{kume2026},
which provides algebraic relations between moments of different orders,
parameterized by the score matching parameters $\lambda$
(Proposition~\ref{prop:stein}).
Second, after a Bayesian measurement update, the posterior moments must
be recovered from the updated parameters $\lambda^+$.
For polynomial measurement models, the update is conjugate:
$\lambda^+ = \lambda^- + \lambda_{\mathrm{lik}}$,
where $\lambda_{\mathrm{lik}}$ is the score parameter induced by the
likelihood and $-$ and $+$ denote pre- and post-measurement parameters
(analogous to the information-form Kalman update). The same
Stein identity yields a linear system for the posterior moments without
evaluating $Z$.
The resulting algorithm, the \textbf{Score Kalman Filter} (SKF),
performs every step of the predict-update loop via linear algebra.
In Section~\ref{sec:stein}, we study extensively how this method extends to higher-dimensional
systems.

\textbf{Contributions.}
(1)~\textbf{Score matching reconstruction (Section~\ref{sec:sm}).}
We show that score matching for the polynomial exponential family
reduces the density fit to a single linear system on the propagated
moments, removing the need for partition-function evaluation, iterative
optimization, or numerical quadrature.
(2)~\textbf{Stein closure and posterior recovery
(Sections~\ref{sec:stein}--\ref{sec:recovery}).}
We use Stein's identity to obtain algebraic moment relations
parameterized by the score matching coefficients $\lambda$, closing the
moment hierarchy during prediction and recovering posterior moments
after each Bayesian update, both as linear solves.
(3)~\textbf{Score Kalman Filter and experiments
(Section~\ref{sec:experiments}).}
We introduce the Score Kalman Filter (SKF) and evaluate it on coupled
oscillator networks. On this structured benchmark, the SKF scales through $n{=}20$ with mean RMSE
below the reported EKF/UKF/EnKF/PF baselines at $n{=}4$--$20$
via the active Stein closure,
extending beyond the
$n \le 4$ regime demonstrated for partition-function-based MaxEnt
moment filtering~\cite{teng2025}.
(4)~\textbf{Recovery of the information-form Kalman filter
(Appendix~\ref{app:gaussian}).} We verify that at $r{=}2$, the SKF
reduces to the information-form Kalman filter exactly.

\section{Background}
\label{sec:background}

\subsection{Stochastic systems and moment propagation}

Consider a continuous-time stochastic system with state $x \in \R^n$
evolving via the It\^o stochastic differential equation (SDE):
\begin{equation}\label{eq:sde}
    dx = X(x)\,dt + h(x)\,dW,
\end{equation}
where $X : \R^n \to \R^n$ is the drift vector field,
$h : \R^n \to \R^{n \times n_w}$ is the diffusion matrix,
and $W$ is a standard $\R^{n_w}$-valued Wiener process.
The density $p(t,x)$ of the state satisfies the Fokker--Planck equation
$\partial_t p = -\nabla \cdot (Xp) + \sum_{ij}\partial_i\partial_j(H_{ij}\, p)$,
where $H = \frac{1}{2}hh^\top$ is the diffusion tensor.
For constant $H$, the second term reduces to $\Tr(H\nabla^2 p)$.

\textbf{Infinitesimal generator.}
Rather than solving the Fokker--Planck PDE for $p$ directly,
we work with the \emph{stochastic Koopman semigroup}
$\{\mathcal{K}_t\}_{t \ge 0}$, a family of linear operators
acting on observable functions $f \colon \R^n \to \R$ by
$\mathcal{K}_t f(x) = \E[f(X_t) \mid X_0 = x]$.
The \emph{infinitesimal generator}
$\mathcal{A} \colon C^2(\R^n) \to C(\R^n)$ is the derivative
of this semigroup at $t = 0$. For $f \in C^2(\R^n)$,
\begin{equation}\label{eq:generator}
    \mathcal{A}f(x) := \lim_{t \searrow 0}
    \frac{\mathcal{K}_t f(x) - f(x)}{t}
    = \nabla f(x) \cdot X(x) + \Tr\!\bigl(H(x)\,\nabla^2 f(x)\bigr),
\end{equation}
where the second equality follows from It\^o's lemma.
The first term captures the deterministic drift and the second
captures the stochastic diffusion.
The generator gives rise to Dynkin's formula,
which describes how expectations of $f$ evolve over time,
\begin{equation}\label{eq:dynkin}
    \frac{d}{dt}\E[f(X_t)] = \E[\mathcal{A}f(X_t)].
\end{equation}

\textbf{Moment propagation.}
Applying Dynkin's formula to the monomial test functions
$\phi_\alpha(x) = x^\alpha$ (multi-index $\alpha \in \N^n$,
$|\alpha| = \sum_i \alpha_i$, and $e_i \in \N^n$ denotes the $i$-th standard basis multi-index) gives a system of ODEs for the moments
$m_\alpha(t) := \E[X_t^\alpha]$:
\begin{equation}\label{eq:moment-ode}
    \frac{d}{dt} m_\alpha(t) = \E\bigl[\mathcal{A}\phi_\alpha(X_t)\bigr].
\end{equation}
We now restrict to \emph{polynomial systems}, where each component
of the drift $X(x)$ is a polynomial in $x$ of total degree at most $d_X$,
and each entry of the diffusion matrix $h(x)$ is a polynomial of total
degree at most $d_h$.
Under this assumption, the generator~\eqref{eq:generator}
applied to a monomial $\phi_\alpha$ of degree $|\alpha|$ produces
two polynomial contributions.
The drift term $\nabla\phi_\alpha \cdot X$ has degree
$|\alpha| - 1 + d_X$,
and the diffusion term $\Tr(H\nabla^2\phi_\alpha)$
has degree $|\alpha| - 2 + 2d_h$, since
$H = \frac{1}{2}hh^\top$ has entries of degree $2d_h$.
Taking expectations, the right-hand side of~\eqref{eq:moment-ode}
involves moments up to degree $|\alpha| + \bar{d}$, where
\begin{equation}\label{eq:dbar}
    \bar{d} = \max(d_X - 1,\; 2d_h - 2)
\end{equation}
is the \emph{excess degree}, measuring how many additional moment degrees
are introduced by the drift nonlinearity and the state-dependence of the noise.
If $\bar{d} = 0$, as in systems with linear drift and additive or
linear multiplicative noise, the moment ODE closes at any truncation order.
If $\bar{d} \ge 1$, as in systems with quadratic drift
or quadratic multiplicative noise, propagating degree-$K$ moments
requires degree-$(K{+}\bar{d})$ moments, producing an infinite, hierarchically
coupled system.

This is the classical \emph{moment closure problem}~\cite{kuehn2016},
arising in SDEs~\cite{whittle1957,gillespie2009},
kinetic theory~\cite{grad1949,levermore1996},
and network dynamics~\cite{kirkwood1935,keeling1999,kiss2017}.
Classical closures either assume Gaussianity~\cite{whittle1957} or require
$Z$-evaluation~\cite{levermore1996,singh2011}.
Our Stein closure (Section~\ref{sec:stein})
exploits the exponential family structure to derive algebraic
moment relations from $\lambda$, without Gaussianity, $Z$, or independence.

\subsection{Maximum-entropy density reconstruction and computational cost}

A moment-based filter must also \emph{reconstruct} a density from
the propagated moments.
The maximum-entropy (MaxEnt) principle~\cite{mead1984}
solves this by maximizing Shannon entropy subject to moment constraints,
yielding the \emph{polynomial exponential family}:
\begin{equation}\label{eq:med}
    p^*(x;\lambda)
    = \frac{1}{Z(\lambda)}\exp \bigl(-\lambda \cdot \phi(x)\bigr),
    \quad
    Z(\lambda) = \int \exp(-\lambda \cdot \phi(x))\,dx,
\end{equation}
where $\phi(x) = (x^\alpha)_{|\alpha| \le r}$ and
$\lambda \in \R^M$, $M = \binom{n+r}{n}$.
Let $\Lambda=\{\lambda:Z(\lambda)<\infty\}$ denote the natural
parameter domain.
This family includes Gaussians ($r{=}2$) and captures
multimodality at $r \ge 4$~\cite{teng2025}.
The \textbf{computational cost} is dominated by evaluating $Z(\lambda)$
and its gradient at every optimization step,
an $n$-dimensional integral with cost $O(G^n)$.
This exponential scaling in $n$ is what the score matching formulation
in Section~\ref{sec:sm} replaces with a single linear solve.

\subsection{Score matching for unnormalized models}

Let $p_{\mathrm{data}}$ denote the true (unknown) density of the state
at a given time, whose moments $m_\alpha = \E_{p_{\mathrm{data}}}[x^\alpha]$
have been propagated by the moment ODE~\eqref{eq:moment-ode}.
We wish to fit the polynomial exponential-family model in~\eqref{eq:med}
to $p_{\mathrm{data}}$,
without the intractable computation of the normalization $Z(\lambda)$.

Hyv\"arinen~\cite{hyvarinen2005} observed that the \emph{score function}
$s(x;\lambda) = \nabla_x \log p(x;\lambda)\in \R^n$ does not depend on $Z$.
For the polynomial exponential-family model in~\eqref{eq:med},
$s(x;\lambda) = -\nabla_x(\lambda \cdot \phi) = -J_\phi(x)\lambda$,
where $[J_\phi]_{i\alpha} = \partial_i \phi_\alpha$, so $J_\phi(x) \in \R^{n \times M}$.
The Fisher divergence between $p_{\mathrm{data}}$ and the model,
$\frac{1}{2}\E_{p_{\mathrm{data}}}[\|s_\lambda - s_{\mathrm{data}}\|^2]$
where $s_{\mathrm{data}} = \nabla_x \log p_{\mathrm{data}}$,
can be rewritten via integration by parts as
\begin{equation}\label{eq:sm-objective}
    J_{\mathrm{SM}}(\lambda) = \E_{p_{\mathrm{data}}}\!\left[
    \sum_i \partial_i s_i(x;\lambda) + \frac{1}{2} s_i(x;\lambda)^2
    \right] + \text{const},
\end{equation}
which depends only on the model score and its divergence,
not on $Z$ or $s_{\mathrm{data}}$.
The expectation in~\eqref{eq:sm-objective} is over $p_{\mathrm{data}}$
and reduces to polynomial moments of $p_{\mathrm{data}}$ when
$s(x;\lambda)$ is polynomial, as we show in Section~\ref{sec:sm}.

\section{Score matching for the polynomial exponential family}
\label{sec:sm}

We now specialize the score matching objective~\eqref{eq:sm-objective}
to the polynomial exponential family~\eqref{eq:med}.
Substituting the polynomial score $s(x;\lambda) = -J_\phi(x)\lambda$,
the expectation over $p_{\mathrm{data}}$ reduces to polynomial moments,
and the objective becomes a quadratic function of $\lambda$,
$J_{\mathrm{SM}}(\lambda) = \frac{1}{2}\lambda^\top A\lambda - b^\top\lambda + \text{const}$,
where $A$ and $b$ are defined below.
Additive constants do not affect the score, so we fix
the zero-multi-index coefficient $\lambda_{\mathbf{0}}$ to $0$
and write $A$ and $b$ for the system after the
constant row and column are removed.

\begin{proposition}[Score matching as a linear system]\label{prop:sm}
For the polynomial MED~\eqref{eq:med}, $A$ and $b$ in the right-hand side of the score matching objective have the form
\begin{equation}\label{eq:Ab}
    A_{\alpha\beta} = \sum_{i=1}^n \alpha_i\,\beta_i\; m_{\alpha+\beta-2e_i},
    \qquad
    b_\alpha = \sum_{i=1}^n \alpha_i(\alpha_i{-}1)\; m_{\alpha-2e_i},
\end{equation}
for $\alpha, \beta \in \N^n$ with $|\alpha|, |\beta| \le r$.
The constant row and column are then omitted according to the convention
above. If $A$ is nonsingular, the score matching estimator is
\begin{equation}\label{eq:lambda-star}
    \lambda^* = A^{-1}b,
    \qquad
    \lambda^*_{\mathbf{0}}=0 .
\end{equation}
Every entry of $A$ and $b$ is a polynomial moment.
$A$ requires moments up to degree $2r{-}2$, and $b$ requires moments up to $r{-}2$.
\end{proposition}

\begin{proof}[Proof sketch]
$A = \E[J_\phi^\top J_\phi]$ and $b_\alpha = \E[\Delta\phi_\alpha]$
(Laplacian of the monomial $\phi_\alpha$).
Both reduce to moment lookups via $\partial_i(x^\alpha) = \alpha_i x^{\alpha-e_i}$ with the convention that $\alpha_i x^{\alpha - e_i} = 0$ when $\alpha_i = 0$.
The constant row and column vanish because $\nabla \phi_{\mathbf{0}} = 0$.
Under this convention, positive definiteness follows from the
polynomial-gradient argument in Appendix~\ref{app:proof-equiv}.
See Appendix~\ref{app:proofs}.
\end{proof}

\textbf{Key properties.}
(i) Solving \eqref{eq:lambda-star} costs $O(M^3)$, where $M = \binom{n+r}{n}$.
(ii) No partition function $Z(\lambda)$ is evaluated.
(iii) No iterative optimization (BFGS, Newton, etc.) is needed.
(iv) \emph{Graceful degradation} (Proposition~\ref{prop:graceful},
Appendix~\ref{app:perturbation}):
if the propagated moments have relative error
$\|\delta m\|/\|m\|$, the fitted parameters shift by
\begin{equation}\label{eq:graceful}
    \frac{\|\delta\lambda\|}{\|\lambda\|}
    = O\bigl(\kappa(A)\, \|\delta m\| / \|m\|\bigr),
\end{equation}
where $\kappa(A)$ denotes the condition number of $A$,
giving continuous, proportional degradation
rather than the realizability failure of MaxEnt
(see Appendix~\ref{app:perturbation} for details).

\subsection{Score matching recovers the correct parameters for consistent moments}

A natural concern is that score matching (which minimizes the Fisher divergence)
does not explicitly enforce moment constraints, and hence
the fitted density may not match the propagated moments.
The following theorem shows that this concern vanishes when the
moments are consistent with a density in the polynomial
exponential family.

\begin{theorem}[Correct recovery on the model class]\label{thm:sm-maxent}
Let $p_0=p(\cdot;\lambda_{\mathrm{true}})$ be a member of the degree-$r$ polynomial
exponential family~\eqref{eq:med}, with
$\lambda_{\mathrm{true}} \in \operatorname{int}(\Lambda)$ under the
zero-constant convention $\lambda_{\mathbf{0}}=0$.
Form $A$ and $b$ in~\eqref{eq:Ab} from the moments of $p_0$ up to degree
$2r{-}2$. Then the score matching solution recovers the non-constant
coefficients,
\begin{equation}
    \lambda^*_\alpha = (\lambda_{\mathrm{true}})_\alpha,
    \qquad 1 \le |\alpha| \le r.
\end{equation}
Consequently $p(\cdot;\lambda^*)=p_0$.
\end{theorem}

\begin{proof}[Proof sketch]
Write $D_F(p_0 \| p_\lambda) = \frac{1}{2}\E_{p_0}[\|s_0 - s_\lambda\|^2]$
for the Fisher divergence, where $s_0 = -J_\phi \lambda_{\mathrm{true}}$
and $s_\lambda = -J_\phi \lambda$ are the polynomial scores.
Then $D_F = \frac{1}{2}(\lambda{-}\lambda_{\mathrm{true}})^\top A_0 (\lambda{-}\lambda_{\mathrm{true}})$,
where $A_0 = \E_{p_0}[J_\phi^\top J_\phi]$.
Since $A_0$ equals $A$ (both are computed from the moments of $p_0$),
and $A$ is positive definite under the convention above,
$D_F = 0$ if and only if $\lambda_\alpha = (\lambda_{\mathrm{true}})_\alpha$
for all $|\alpha| \ge 1$.

By Hyv\"arinen's identity~\cite{hyvarinen2005},
$J_{\mathrm{SM}}(\lambda) = D_F(p_0 \| p_\lambda) + C$
where $C$ is independent of $\lambda$
(this uses the integration-by-parts formula and the
assumption that boundary terms vanish, which holds because
$\lambda_{\mathrm{true}} \in \operatorname{int}(\Lambda)$ ensures
$p_0$ decays sufficiently fast).
Since $D_F$ and $J_{\mathrm{SM}}$ share the same minimizer,
and $J_{\mathrm{SM}}(\lambda) = \frac{1}{2}\lambda^\top A\lambda - b^\top\lambda$
has a unique minimizer under the convention above,
we conclude $\lambda^* = \lambda_{\mathrm{true}}$ on $|\alpha| \ge 1$.
The full proof with regularity conditions is in
Appendix~\ref{app:proof-equiv}.
\end{proof}

\begin{remark}[Connection to MaxEnt]
By a classical result in information theory, the member of the
exponential family $\mathcal{P}_r$ with moments $m$ is the unique
density that maximizes the Shannon entropy among all densities
sharing those moments.
Combined with Theorem~\ref{thm:sm-maxent}, this implies that
\emph{when the moments are consistent with the family},
score matching and MaxEnt produce the same density,
even though the two methods optimize different objectives
(Fisher divergence vs.\ KL divergence).
When the true density lies outside $\mathcal{P}_r$,
the two methods generally give different parameters.
In our filtering pipeline, the moments are propagated by Dynkin's formula
(exact for polynomial systems with appropriate $\bar{d}$),
so the conditions of Theorem~\ref{thm:sm-maxent} are approximately
satisfied, explaining the close moment agreement observed empirically
(Section~\ref{sec:experiments}).
\end{remark}

\section{Stein closure for the moment hierarchy}
\label{sec:stein}

\subsection{The closure problem}

As discussed in Section~\ref{sec:background}, propagating the moments
of degree $|\alpha| \le K$ via~\eqref{eq:moment-ode} requires moments
of degree $K + \bar{d}$ on the right-hand side. The case $\bar{d} = 0$,
which covers linear drift with at most linearly multiplicative noise,
is the easy one because the hierarchy closes at any chosen truncation
$K$ and no auxiliary information is needed. The cases of practical
interest are $\bar{d} \ge 1$, where the moment ODE for a degree-$K$
state vector depends on moments at degree $K + 1$ or higher, and some
rule must supply those higher moments. Classical closures fill this
gap by imposing distributional assumptions on the unresolved moments,
typically Gaussianity~\cite{whittle1957} or independence of pairs and
triples~\cite{kirkwood1935,keeling1999}. We instead derive the higher
moments from the score parameters $\lambda$ that the score matching
fit has already produced, using the Stein identity below.

\subsection{Score-based closure via Stein's identity}

\begin{proposition}[Stein's identity for the polynomial MED]\label{prop:stein}
If $p(x;\lambda)$ has score $s = -J_\phi\lambda$ and
the displayed moments are finite and all boundary terms vanish, then for any
multi-index $\beta$ and coordinate $i$,
\begin{equation}\label{eq:stein}
    \sum_{|\alpha| \le r} \lambda_\alpha\, \alpha_i\; m_{\alpha+\beta-e_i}
    = \beta_i\; m_{\beta-e_i}.
\end{equation}
Terms with negative multi-indices are interpreted as zero.
\end{proposition}

\begin{proof}[Proof sketch]
Apply Stein's identity $\E[s_i f + \partial_i f] = 0$
with $f(x) = x^\beta$ and substitute the polynomial score.
The full derivation, including the boundary-term conditions,
is given in Appendix~\ref{app:proofs}.
\end{proof}

Equation~\eqref{eq:stein} is a \textbf{linear relation} between moments,
parameterized by $\lambda$.
To close the hierarchy, we target an unclosed moment $m_\gamma$
with $|\gamma| = K + 1$ by choosing $\alpha_0$ with $|\alpha_0| = r$
and setting $\beta = \gamma - \alpha_0 + e_i$, which gives an equation expressing
$m_\gamma$ in terms of moments of degree $\le |\gamma|$ and $\lambda$.
Specifically, since $|\beta| = K - r + 2$, we have $|\alpha + \beta - e_i| \le K + 1$ and the left-hand side involves moments up to degree $|\gamma| = K + 1$, while the right-hand side depends only on known moments as $|\beta - e_i| = K - r + 1 \le K$.
Collecting all such equations yields a linear system
$\Lambda_1(\lambda)\, m^{(K+1)} = c_1(\lambda, m^{(\le K)})$
for the degree-$(K{+}1)$ unknowns.
For $\bar{d} \ge 2$, the procedure is applied sequentially, with
layer $j$ resolving degree $K{+}j$ from degree $\le K{+}j{-}1$.

The coefficient matrix $\Lambda_1(\lambda)$ has polynomial entries in $\lambda$,
so its rank-deficient locus is an algebraic variety -- either all of $\R^M$
or a proper algebraic subset of Lebesgue measure zero.
For $n{=}2$ (any $r \ge 2$), the first-layer system is square and generically
nonsingular (Appendix~\ref{app:stein-wellposed}). At $r=2$, positive
precision is a simple sufficient condition.
For general $n$ and $r \ge 3$, the simple $n{=}2$ picture gives way
to an explicit counting problem.
Appendix~\ref{app:stein-wellposed} carries out this count.
At $r{=}3$, the augmented full-basis system is overdetermined through
$n{=}15$ and becomes underdetermined at $n{=}16$.
At $r{=}4$, the same crossover moves to the gap between $n{=}35$ and
$n{=}36$.
Our full-basis experiments all fall on the overdetermined side of this
count, and $\Lambda_1$ had full column rank at every time step in those tests
($n \le 10$, $r \le 8$).
For structured systems, we may instead close only the unclosed moments
that appear in the moment ODE. This is the active Stein closure used
for the coupled oscillator scaling experiment through $n{=}20$.
The resulting linear system is overdetermined but typically inconsistent,
and we take its least-squares solution.

\begin{remark}[Factorization caching]\label{rem:factor-cache}
The coefficient matrix $\Lambda_1(\lambda)$ depends only on the score
parameters $\lambda$, not on the moments $m$.
During prediction, $\lambda$ is fixed while $m(t)$ evolves through
$N_{\mathrm{sub}}$ ODE substeps.
Therefore $\Lambda_1$ need be factorized only \emph{once} per prediction
window.
Each substep rebuilds only the right-hand side
$c(\lambda,\, m^{(\le K)})$ and solves via back-substitution,
reducing the closure cost from
$O(N_{\mathrm{sub}}\, N_{\mathrm{eq}}\, N_{\mathrm{unk}}^2)$ to
$O(N_{\mathrm{eq}}\, N_{\mathrm{unk}}^2 + N_{\mathrm{sub}}\, N_{\mathrm{unk}}^2)$.
\end{remark}

\textbf{Self-consistency loop.} At each time step,
$\lambda(t)$ from the previous score matching fit provides the
coefficients in the Stein closure.
After propagation, the new moments yield the next score matching fit,
which is used for the following closure.

\begin{remark}[Moment budget versus identifiability]\label{rem:moment-budget}
Although a member of $\mathcal{P}_r$ is uniquely determined, after fixing
the constant parameter, by its moments up to degree $r$ via the
bijectivity of the moment map (Appendix~\ref{app:expfam}),
the score matching estimator in Proposition~\ref{prop:sm} requires moments up to
degree $2r{-}2$. The gap reflects a distinction between information-theoretic
identifiability and algorithmic recoverability, since the quadratic structure of
$J_{\mathrm{SM}}(\lambda) = \frac{1}{2}\lambda^\top A \lambda - b^\top \lambda$
makes the Gram matrix $A$ depend on moments of degree up to $2r{-}2$. Therefore, the
resulting Stein closure should be interpreted as a \emph{model-based}
reconstruction of higher-order moments (given lower-order moments plus the
polynomial-MED assumption), rather than a universal lower-to-higher moment
generation procedure.
\end{remark}

\subsection{Stein-augmented score matching}

The Stein identity at $|\beta| = \ell$ with $\ell < r$ reproduces the
score matching equations (up to scaling). The first non-redundant constraint
is at $\ell = r$, introducing $m_{2r-1}$.
This augmentation tightens the moment consistency of the fitted $\lambda$
beyond what score matching alone provides.
The connection between the generalized score matching and the generalized method of moments is formalized
by Kume and Walker~\cite{kume2026}.

\section{Measurement update and moment recovery}
\label{sec:update}

\subsection{Score-based Bayesian update}

Let $z \in \R^{n_z}$ be a measurement of the state, given by a function of the state and noise through the likelihood $p(z \mid x)$.
At measurement time, the posterior score adds the likelihood score,
\begin{equation}\label{eq:score-update}
    s_{\mathrm{post}}(x) = s_{\mathrm{prior}}(x) + \nabla_x \log p(z|x).
\end{equation}
Equation~\eqref{eq:score-update} is Bayes' rule in score form:
$\log p(x|z)=\log p(x)+\log p(z|x)-\log p(z)$, and taking
$\nabla_x$ removes the evidence term because $p(z)$ is independent of $x$.
If the likelihood $p(z|x)$ has a polynomial measurement model $g(x)$
of degree $d_g$ with Gaussian noise, the likelihood score
$\nabla_x \log p(z|x) = J_g^\top R^{-1}(z - g(x))$
is polynomial of degree $\le 2d_g - 1$.
When $2d_g - 1 \le r - 1$ (i.e., $d_g \le r/2$), the posterior
score fits within the basis, giving
\begin{equation}\label{eq:conjugate}
    \lambda^+ = \lambda^- + \lambda_{\mathrm{lik}}(z).
\end{equation}
Equation~\eqref{eq:conjugate} is the usual exponential-family
conjugate update written in natural parameters. The update changes
$\lambda$ by addition and does not require evaluating the evidence or
the partition function.

\subsection{Moment recovery from posterior parameters}
\label{sec:recovery}

After the update \eqref{eq:conjugate}, $\lambda^+$ is known but $m^+$ is not.
Although $m^+=\mu(\lambda^+)$ is identifiable on the non-constant minimal
parameterization (Appendix~\ref{app:expfam}), evaluating $\mu$ would again
require $Z$. We instead impose the componentwise Stein identity at
$\lambda=\lambda^+$ and $m_0=1$, using directional rows $(\beta,i)$ with
$|\beta|\ge 1$, $\beta_i\ge 1$,
\begin{equation}\label{eq:stein-recovery}
    \sum_{|\alpha| \le r} \lambda^+_\alpha\, \alpha_i\;
    m^+_{\alpha+\beta-e_i}
    = \beta_i\; m^+_{\beta-e_i}.
\end{equation}
Rows with $|\beta|=\ell$ couple moments up to degree $r+\ell-1$.
Within the tracked budget $K=2r-2$, the exact directional rows are
underdetermined (e.g., $R=12$, $U=27$ for $n=2,r=4$), so we include
higher-$|\beta|$ rows and truncate degree $>K$ terms. This gives a
least-squares projection onto the Stein consistency manifold, with omitted
near-Gaussian centered terms of order $O(\sigma^{K+2})$. Counts, optional
$\beta_i=0$ rows, and conditioning are detailed in Appendix~\ref{app:stein-wellposed}.

\textbf{Conditioning.}
Working in centered coordinates $z = x - \mu$ (where $\mu$ is the
current mean, available from the propagated moments) is essential.
The monomial basis in raw coordinates gives $\|\lambda\| \sim 10^5$,
while centering reduces this to $O(10)$--$O(100)$.
Orthogonal bases (Legendre) provide further improvement~\cite{liu2026mepoly}.
The truncation error decreases as the neglected
higher-order centered moments become smaller,
which is why centering is important for accuracy.

\textbf{Iterative consistency refinement.}
The conjugate update $\lambda^+ = \lambda^- + \lambda_{\mathrm{lik}}$
followed by truncated Stein recovery produces moments $m^+$ that may
not be perfectly consistent with $\lambda^+$.
We enforce consistency by iterating: (i) re-fit $\lambda$ via score
matching to the recovered $m^+$, (ii) re-apply the measurement
$\lambda \leftarrow \lambda_{\mathrm{refit}} + \lambda_{\mathrm{lik}}$,
(iii) re-recover moments.
A fixed point of this iteration is a $(\lambda, m)$ pair that is
mutually consistent under both the Stein identity and the
measurement update.
In experiments where refinement is enabled, a small number of
iterations suffices (the residual drops by 2--3 orders of magnitude).
The high-dimensional coupled-oscillator sweep uses one refinement step
for speed.
The construction connects score matching to Stein's method of moments
(SMoM)~\cite{kume2026,nagai2026}, with the canonical decomposition and
its convergence implications in Appendix~\ref{app:smom}.

\section{The Score Kalman Filter algorithm}
\label{sec:skf}

Algorithm~\ref{alg:skf} assembles the preceding pieces into one
predict-update loop. The filter carries both moments $m$, used for
Dynkin propagation, and score parameters $\lambda$, used for density
fitting and Stein relations. Prediction advances $m$ and refits
$\lambda$. The update adds the likelihood score, recovers posterior
moments by the same Stein system, and reports the standard posterior
mean $\hat{x}_k = \E[x_k \mid z_{1:k}]$ from the recovered first moments
for plots and RMSE.

\begin{algorithm}[t]
\caption{Score Kalman Filter (polynomial MED, degree $r$)}
\label{alg:skf}
\begin{algorithmic}[1]
\REQUIRE Initial moments $m(0)$, basis order $r$, dynamics $(X, h)$, measurement model $(g_{\mathrm{obs}}, R)$
\STATE $K \gets 2r - 2$ \hfill\COMMENT{moment budget}
\STATE $\lambda \gets A(m(0))^{-1}b(m(0))$ \hfill\COMMENT{initial score matching}
\FOR{each time step}
    \STATE \textbf{Predict.}
    \STATE \quad Propagate moments $m(t) \to m(t{+}\Delta t)$ via \eqref{eq:moment-ode},
    using Stein closure \eqref{eq:stein} with $\lambda(t)$ if $\bar{d} \ge 1$
    \STATE \quad $\lambda^- \gets A(m)^{-1}b(m)$ \hfill\COMMENT{score-matching solve}
    \STATE \textbf{Update} at measurement $z$.
    \STATE \quad $\lambda^+ \gets \lambda^- + \lambda_{\mathrm{lik}}(z)$ \hfill\COMMENT{conjugate score addition}
    \STATE \quad Recover $m^+$ from $\lambda^+$ via Stein system \eqref{eq:stein-recovery}
    \hfill\COMMENT{one least-squares solve}
    \STATE \quad $\lambda \gets A(m^+)^{-1}b(m^+)$ \hfill\COMMENT{re-fit $\lambda$}
\ENDFOR
\end{algorithmic}
\end{algorithm}

\textbf{Complexity.} The dense score fit costs $O(M^3)$ for
$M=\binom{n+r}{n}$. If $\bar d\ge 1$, Stein closure adds a solve with
$N_{\mathrm{eq}}$ rows and $N_{\mathrm{unk}}$ unclosed moments
(dense QR: $O(N_{\mathrm{eq}}N_{\mathrm{unk}}^2)$). The factorization or
sparse operator is reused over ODE substeps because $\Lambda_1$ depends
only on $\lambda$ (Remark~\ref{rem:factor-cache}). Structured generators
can restrict closure to active moments, as in the $n=12$--$20$ oscillator
runs. No partition function is evaluated, and $r=2$ recovers the
information-form Kalman filter (Appendix~\ref{app:gaussian}).

\section{Experiments}
\label{sec:experiments}

We benchmark the SKF across SE(2) rigid-body kinematics,
ecosystem dynamics, 3D incompressible fluid advection, and coupled
oscillator filtering up to $n{=}20$. Full moment panels, density
reconstructions, and filter time series for the remaining cases are in
Appendices~\ref{app:moment-accuracy},
\ref{app:additional-experiments}, and \ref{app:coupled-figures}.
Note that all runs use Python (NumPy/SciPy) on a single laptop CPU
(Intel i9-11900H, no GPU or compiled extensions), so timings are
conservative. Generator calculations, moment ODE derivations, and
centered-coordinate formulas are in
Appendices~\ref{app:generators}--\ref{app:matrices}.

\subsection{SE(2) prediction with Fourier \texorpdfstring{$\times$}{x} position moments (\texorpdfstring{$\bar{d} = 0$}{d-bar = 0})}

SE(2) rigid body kinematics in embedded coordinates $(c, s, p_x, p_y)$
with $\omega{=}1$\,rad/s, $v{=}2$\,m/s, heading noise $\sigma{=}0.3$.
Drift and diffusion are linear, so $\bar{d}=0$ and $m(T) = e^{LT}m(0)$
in the geometry-adapted Fourier $\times$ position basis
$\E[\cos(k\theta)\, p_x^a p_y^b]$, $\E[\sin(k\theta)\, p_x^a p_y^b]$.
All moments match $5{\times}10^5$-particle MC to within $0.5\%$
(Appendix~\ref{app:moment-accuracy},
Figure~\ref{fig:moment-comparison-se2}). Score matching on the
$(p_x, p_y)$ marginal captures the non-Gaussian curved shape from
$r=4$ ($M=15$, $0.4$\,ms, timing in Table~\ref{tab:comparison})
(Figure~\ref{fig:se2-density}).

\begin{figure}[t]
\centering
\includegraphics[width=\linewidth]{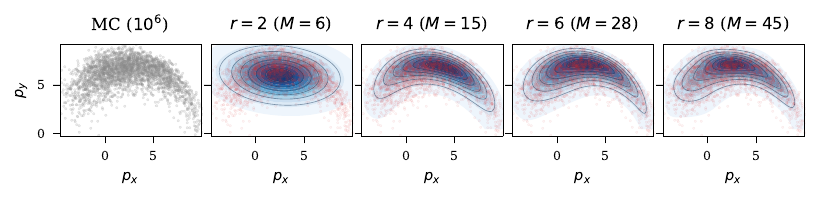}
\caption{SE(2) density reconstruction. The leftmost panel shows the
Monte Carlo reference for the $(p_x, p_y)$ marginal, and the four
panels to its right show the score matching reconstruction at basis
orders $r = 2, 4, 6, 8$. The Gaussian fit at $r = 2$ misses the
banana-shaped curvature entirely, while $r = 4$ already recovers
the bulk of the non-Gaussian structure, with $r = 6, 8$ refining
the tails.}
\label{fig:se2-density}
\end{figure}

\subsection{Stochastic Lotka-Volterra with Stein closure (\texorpdfstring{$\bar{d} = 1$}{d-bar = 1})}

Stochastic predator-prey dynamics
$dx_1 = (\alpha x_1 - \beta x_1 x_2)\,dt + \sigma_1\,dW_1$,
$dx_2 = (-\gamma x_2 + \delta x_1 x_2)\,dt + \sigma_2\,dW_2$
($\alpha{=}1$, $\beta{=}0.5$, $\gamma{=}0.8$, $\delta{=}0.3$,
$\sigma_1{=}0.3$, $\sigma_2{=}0.2$). The bilinear interaction gives
$d_X{=}2$, $\bar{d}{=}1$, the textbook moment closure
benchmark~\cite{whittle1957}. With $r{=}6$, $T{=}1.5$\,s, and centered
coordinates, Dynkin+Stein moments match $10^6$-particle MC to $0.065\%$
(variance) and $2$--$7\%$ (third-order), error growing monotonically
with degree (Figure~\ref{fig:lv-moments-compact}, full grid in
Appendix~\ref{app:moment-accuracy}).
The Dynkin+Stein run takes $0.7$\,s on the CPU used for Table~\ref{tab:comparison}.
The MC reference used for validation takes $69$\,s.
Density reconstruction in Appendix~\ref{app:additional-experiments}
(Figure~\ref{fig:lv-appendix}).

\begin{figure}[t]
\centering
\includegraphics[width=0.85\linewidth]{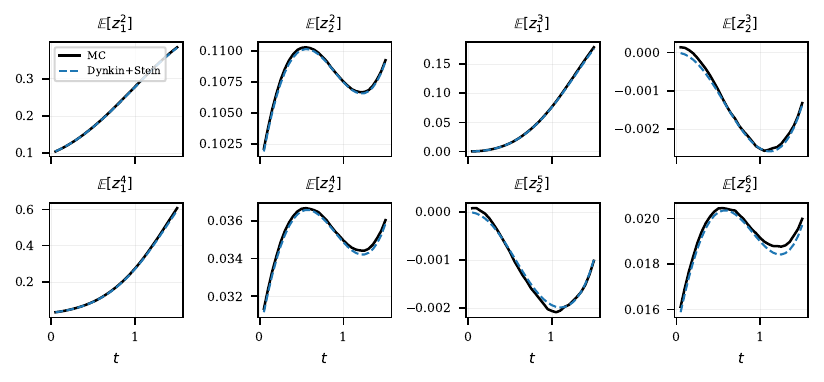}
\caption{Representative LV centered moments ($n{=}2$, $r{=}6$, $\bar{d}{=}1$):
Dynkin + Stein (blue dashed) vs.\ MC (black), degrees 2--6
(up to the score matching truncation order $r$).
Error grows monotonically with degree.}
\label{fig:lv-moments-compact}
\end{figure}

\subsection{Filtering on coupled oscillator networks}
\label{sec:exp-highdim}

We demonstrate the full predict-update SKF loop on $N$ coupled Duffing
oscillators with nearest-neighbor coupling
($\dot{q}_i = p_i$,
$\dot{p}_i = -\gamma p_i - \alpha q_i - \beta q_i^2 + \kappa(q_{i+1}{-}2q_i{+}q_{i-1}) + \sigma\,dW_i$,
$\gamma{=}0.3$, $\alpha{=}1$, $\beta{=}0.6$, $\kappa{=}0.3$, $\sigma{=}0.4$,
$n = 2N$, $d_X = 2$, $\bar{d} = 1$, centered coordinates,
Appendix~\ref{app:centering}). The single-oscillator $n{=}2$ Duffing
filter is in Appendix~\ref{app:duffing-filter}.
At $N{=}3,4,5$ ($n{=}6,8,10$, $r{=}3$) observing the odd-indexed
positions, the SKF has lower mean RMSE than EKF/UKF/EnKF/PF across
10 independent seeds in our implementation (Figure~\ref{fig:highdim},
Table~\ref{tab:comparison}).

\begin{figure}[t]
\centering
\includegraphics[width=\linewidth,height=0.4\textheight,keepaspectratio]{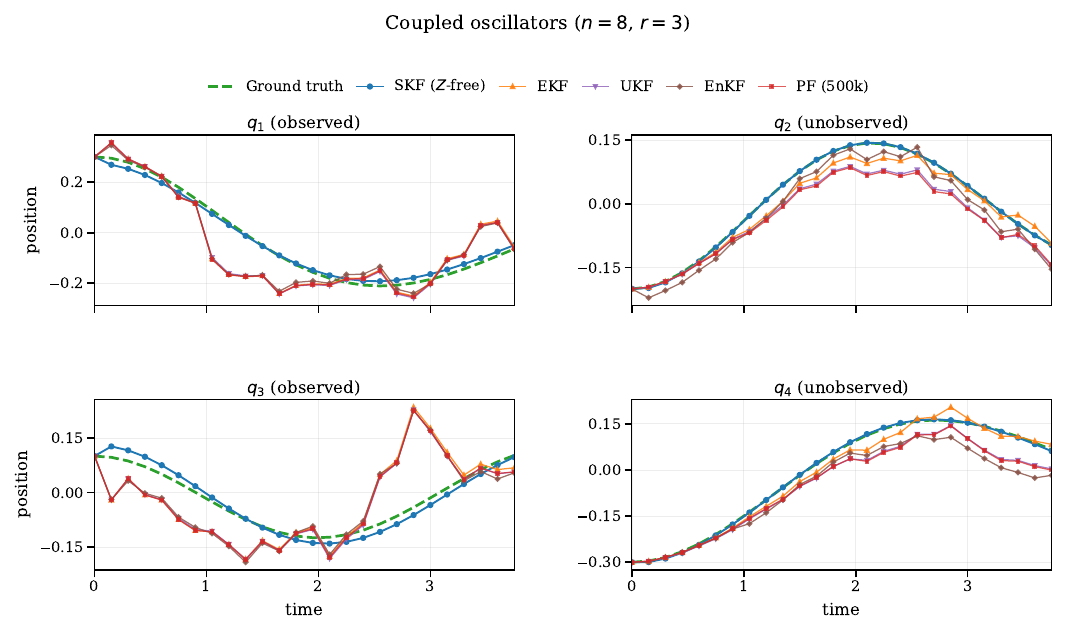}
\caption{Coupled oscillators ($n{=}8$, $r{=}3$). The SKF (blue) tracks
all four positions through 25 steps, including unobserved oscillators
($q_2$, $q_4$) inferred from coupling alone. EKF (orange), UKF
(purple), EnKF (brown), and PF (red) have larger
RMSE. Multi-seed statistics are reported in Table~\ref{tab:comparison}.}
\label{fig:highdim}
\end{figure}

We further scale to $n=12,\ldots,20$ with $r=3$ using an active Stein
closure that keeps only the degree-five moments that the moment ODE
actually requests, rather than the full augmented Stein system that
the structural count of
Appendix~\ref{app:stein-wellposed} predicts to become underdetermined
at $n=16$ (Appendix~\ref{app:coupled-figures}). Even at these
dimensions, the SKF stays more accurate than every baseline in
Table~\ref{tab:comparison}. The fact that EKF, UKF, and EnKF all
converge to a similar error band suggests that the Gaussian belief
approximation, not the choice of update law, is what is limiting them
on this benchmark. Factorization caching also keeps the SKF runtime
below the $5{\times}10^5$-particle bootstrap filter at $n=10$, where
the particle count is the only setting that gave the PF stable
tracking. All numbers in Table~\ref{tab:comparison} are averaged over
10 random seeds, with the SKF reported as mean$\,\pm\,$std.

\section{Related work}

\textbf{Gaussian filtering.}
The KF~\cite{kalman1960}, EKF~\cite{schmidt1966,jazwinski1970},
UKF~\cite{julier1997,wan2000}, EnKF~\cite{evensen1994,burgers1998},
and geometric variants
(InEKF~\cite{barrau2017,hartley2020}, EqF~\cite{vangoor2023},
UKF-M~\cite{brossard2020}) form a mature lineage, but their standard
forms do not provide an explicit multimodal density representation.

\textbf{Moment-based filtering.}
Moment-based filtering for polynomial systems has been developed
primarily for robotics state estimation, including closed-form
moment recursions for discrete-time trigonometric-polynomial
systems~\cite{jasour2021moment}, the moment-based Kalman
filter~\cite{shimizu2023moment}, the generalized moment Kalman
filter~\cite{teng2024gmkf}, and the MEM-KF~\cite{teng2025} which
couples moment propagation with MaxEnt density reconstruction.
These filters operate on discrete-time prediction/update maps. The SKF
extends this to continuous-time polynomial SDEs
(Appendix~\ref{app:closure-comparison}).

\textbf{Score matching and Stein's method of moments.}
Score matching~\cite{hyvarinen2005} fits unnormalized models without
$Z$, with extensions to generative modeling~\cite{song2021scorebased}
and $S^2$ filtering~\cite{bukal2017}. Stein's method of moments
(SMoM)~\cite{ebner2025,kume2026,nagai2026} derives $Z$-free
estimators that include generalized score matching as a special case
and motivate our iterative consistency refinement
(Section~\ref{sec:recovery}).

\textbf{Particle filters and data-driven methods.}
Particle filters~\cite{zanetti2024} represent non-Gaussian beliefs by
Monte Carlo samples but suffer weight degeneracy under partial
observation, with accuracy controlled by the particle count.
KalmanNet~\cite{revach2022} replaces the Kalman gain by a learned
neural network at the cost of supervised training data.

\section{Conclusion}

The SKF replaces partition-function evaluation by linear solves for density
fitting, moment closure, and posterior recovery. It agrees with MaxEnt on
the polynomial exponential family (Theorem~\ref{thm:sm-maxent}) and recovers
the information-form Kalman filter at $r=2$, giving a partition-function-free
route to high-order moment filtering. The same algebra also opens a research
program around automatic basis selection, adaptive truncation, active Stein
closures for sparse generators, and diagnostics that decide when higher-order
moments are worth propagating. We discuss these directions in
Appendix~\ref{app:limitations}.

\section*{Acknowledgements}
Supported in part by NSF grant 2103026, and AFOSR grants
FA9550-32-1-0215 and FA9550-23-1-0400 (MURI). M. Ghaffari was supported by AFOSR YIP FA9550-25-1-0224.

\bibliographystyle{unsrtnat}
\bibliography{references}

\newpage
\appendix
\renewcommand{\thefigure}{A\arabic{figure}}
\setcounter{figure}{0}
\renewcommand{\thetable}{A\arabic{table}}
\setcounter{table}{0}
\renewcommand{\theequation}{A.\arabic{equation}}
\setcounter{equation}{0}

\section{Background on exponential families and maximum entropy}
\label{app:expfam}

This section collects standard results on exponential families~\cite{wainwright2008,jaynes1957}
that are used in our proofs.
We state them for the polynomial exponential family
$\mathcal{P}_r = \{p(\cdot;\lambda) = \exp(-\lambda\cdot\phi)/Z(\lambda) : \lambda \in \Lambda\}$ with natural parameter domain $\Lambda = \{\lambda \in \R^M : Z(\lambda) < \infty\}$,
but with the usual identifiability convention that the constant statistic
is removed, or equivalently that $\lambda_{\mathbf{0}}=0$ is fixed.
Without this convention, shifting the constant parameter only rescales
$Z(\lambda)$ and leaves the density unchanged, so strict convexity and
injectivity cannot hold on the full redundant parameter space.
The results below hold for this minimal representation, and more generally
for any exponential family with affinely independent sufficient statistics.

\textbf{Log-partition function.}
The log-partition function $\Psi(\lambda) := \log Z(\lambda)
= \log \int \exp(-\lambda\cdot\phi(x))\,dx$ is convex on $\Lambda$
and strictly convex on $\operatorname{int}(\Lambda)$ after the constant
direction has been removed.
Its gradient and Hessian are
\begin{equation}
    \nabla_\lambda \Psi = -\E_{p_\lambda}[\phi(x)],
    \qquad
    \nabla^2_\lambda \Psi = \operatorname{Cov}_{p_\lambda}[\phi(x)],
\end{equation}
where $\operatorname{Cov}_{p_\lambda}[\phi]$ is the covariance matrix
of the sufficient statistics under $p(\cdot;\lambda)$.
Since the covariance matrix is positive semidefinite (and positive definite
when the sufficient statistics in the minimal representation are not
a.s.\ affinely dependent), $\Psi$ is strictly convex on that representation.

\textbf{Moment map.}
The \emph{moment space} $\mathcal{M} \subset \R^{M-1}$ is the set of all non-constant moment vectors $m = \E_q[\phi(x)]$ realizable by some probability distribution $q$ on $\R^n$.
Define the moment map $\mu \colon \operatorname{int}(\Lambda) \to \mathcal{M}$
by $\mu(\lambda) = \E_{p_\lambda}[\phi(x)]$.
Since $\mu = -\nabla\Psi$ and $\Psi$ is strictly convex on the minimal parameterization,
$\mu$ is injective: distinct parameters give distinct moments.
Moreover, $\mu$ is a diffeomorphism from $\operatorname{int}(\Lambda)$ onto $\operatorname{int}(\mathcal{M})$ in this minimal representation: bijectivity follows from strict convexity of $\Psi$, smoothness of $\mu$ from analyticity of $\Psi$, and smoothness of $\mu^{-1}$ from the inverse function theorem since $\nabla^2\Psi \succ 0$ on $\operatorname{int}(\Lambda)$ implies $\nabla\mu = -\nabla^2\Psi$ is invertible.
This means that for any moment vector $m \in \operatorname{int}(\mathcal{M})$,
there is a \emph{unique} $\lambda \in \operatorname{int}(\Lambda)$
such that $\E_{p_\lambda}[\phi] = m$.

\textbf{Maximum entropy characterization.}
Among all densities $q$ satisfying $\E_q[\phi_\alpha(x)] = m_\alpha$
for $1\le|\alpha| \le r$ (and $\int q = 1$, $q \ge 0$),
the one that maximizes the Shannon entropy
$\mathcal{H}(q) = -\int q \log q\, dx$ is $p(\cdot;\lambda(m))$,
the unique exponential family member with those moments.
This follows from the Gibbs inequality:
for any $q$ with $\E_q[\phi] = m$,
\begin{equation}
    \mathcal{H}(q) = -\E_q[\log q]
    \le -\E_q[\log p_\lambda]
    = \E_q[\lambda\cdot\phi] + \log Z(\lambda)
    = \lambda\cdot m + \Psi(\lambda),
\end{equation}
with equality if and only if $q = p_\lambda$ a.e.
The second inequality uses $\mathrm{KL}(q \| p_\lambda) \ge 0$,
which expands to $-\mathcal{H}(q) - \E_q[\log p_\lambda] \ge 0$.

\textbf{Implication for this paper.}
Score matching recovers $\lambda$ from moments without evaluating $\Psi$
or solving the entropy optimization.
Theorem~\ref{thm:sm-maxent} shows that the recovered parameter agrees
with the MaxEnt/moment-matching parameter when the moments are consistent
with the family.
The bijectivity of $\mu$ guarantees that this parameter is unique.

\section{Proofs}
\label{app:proofs}

\subsection{Proof of Proposition~\ref{prop:sm} (score matching linear system)}

For the polynomial MED with energy $E_\lambda(x) = \lambda \cdot \phi(x)$,
the model score is $s_\lambda(x) = -\nabla_x E_\lambda = -J_\phi(x)\lambda$,
where $[J_\phi]_{i\alpha} = \partial_i \phi_\alpha = \alpha_i x^{\alpha - e_i}$.
The Hyv\"arinen score matching objective~\eqref{eq:sm-objective} becomes
\begin{align}
    J_{\mathrm{SM}}(\lambda)
    &= \E\biggl[\sum_i \partial_i s_i(x;\lambda)
    + \frac{1}{2} s_i(x;\lambda)^2\biggr] \nonumber\\
    &= \E\biggl[-\sum_i \partial_i(J_\phi(x) \lambda)_i
    + \frac{1}{2}\|J_\phi(x)\lambda\|^2\biggr] \nonumber\\
    &= -\E[\Delta\phi(x)]^\top\lambda
    + \frac{1}{2}\lambda^\top\E[J_\phi(x)^\top J_\phi(x)]\lambda \nonumber\\
    &= \frac{1}{2}\lambda^\top A\lambda - b^\top\lambda,
\end{align}
where $A = \E[J_\phi^\top J_\phi]$ and $b = \E[\Delta\phi]$.
We now compute the entries explicitly.

\textbf{Matrix $A$.}
The $(i,\alpha)$-entry of $J_\phi$ is
$[J_\phi]_{i\alpha} = \partial_i(x^\alpha) = \alpha_i x^{\alpha - e_i}$.
Therefore
\begin{align}
    [J_\phi^\top J_\phi]_{\alpha\beta}
    &= \sum_{i=1}^n [J_\phi]_{i\alpha}\,[J_\phi]_{i\beta} \nonumber\\
    &= \sum_{i=1}^n \alpha_i x^{\alpha - e_i} \cdot \beta_i x^{\beta - e_i} \nonumber\\
    &= \sum_{i=1}^n \alpha_i\,\beta_i\; x^{\alpha + \beta - 2e_i}.
\end{align}
Taking the expectation under $p_{\mathrm{data}}$
and writing $m_\gamma = \E_{p_{\mathrm{data}}}[x^\gamma]$ gives
\begin{equation}
    A_{\alpha\beta} = \E_{p_{\mathrm{data}}}[J_\phi^\top J_\phi]_{\alpha\beta}
    = \sum_{i=1}^n \alpha_i\,\beta_i\; m_{\alpha+\beta-2e_i}.
\end{equation}
The multi-index $\alpha + \beta - 2e_i$ has degree
$|\alpha| + |\beta| - 2 \le 2r - 2$,
so every entry of $A$ is a known propagated moment.

\textbf{Vector $b$.}
The Laplacian of $\phi_\alpha$ is
$\Delta\phi_\alpha = \sum_i \partial_i^2(x^\alpha)
= \sum_i \alpha_i(\alpha_i - 1)\,x^{\alpha - 2e_i}$.
Taking the expectation gives
$b_\alpha = \sum_i \alpha_i(\alpha_i - 1)\; m_{\alpha - 2e_i}$,
which requires moments up to degree $r - 2$.

The minimizer satisfies $\nabla_\lambda J_{\mathrm{SM}} = A\lambda - b = 0$,
solved on the non-constant subspace.
The row and column corresponding to $|\alpha| = 0$ (the constant monomial)
are identically zero because $\partial_i(x^0) = 0$,
so the constant component $\lambda_{\mathbf{0}}$ (corresponding to the
zero multi-index) is unidentifiable and is dropped.

After this deletion, $A$ is the Gram matrix of the gradients of the
remaining monomials. Indeed, for any coefficient vector $v$ and
$q_v(x)=\sum_{1\le|\alpha|\le r}v_\alpha x^\alpha$,
\begin{equation}
    v^\top A v
    = \E_{p_{\mathrm{data}}}\!\left[\|J_\phi(x)v\|^2\right]
    = \E_{p_{\mathrm{data}}}\!\left[\|\nabla q_v(x)\|^2\right].
\end{equation}
Thus $A$ is positive semidefinite. It is positive definite precisely when
no nonzero polynomial in the chosen non-constant basis has zero gradient
almost surely under $p_{\mathrm{data}}$. For the smooth densities used in
this paper this holds: if the last display is zero, then $\nabla q_v$
vanishes on an open set, hence everywhere by polynomiality. Therefore
$q_v$ is constant, and since $q_v$ has no constant term, $v=0$.
\qed

\begin{remark}[A pathological example for positive definiteness]
The positive definiteness statement is not meant to cover laws that live
on a lower-dimensional set. For example, if a two-dimensional state is supported on
the line $x_2=0$, then the nonzero quadratic $q(x)=x_2^2$ has
$\nabla q=0$ on the support. The associated coefficient vector is
therefore in the nullspace of $A$.
\end{remark}

\subsection{Perturbation bound (graceful degradation)}
\label{app:perturbation}

\begin{proposition}[Graceful degradation]\label{prop:graceful}
Let $A(m)$ and $b(m)$ be the score matching matrix and vector
assembled from moments $m$ via~\eqref{eq:Ab},
and let $\lambda^* = A^{-1}b$ on the non-constant subspace, where
$A(m)$ is nonsingular and $\lambda^*, b(m)$ are nonzero.
Suppose the propagated moments are perturbed:
$\tilde{m} = m + \delta m$, with $\delta m$ small enough that
$A(\tilde m)$ is also nonsingular.
Then the perturbed parameters $\tilde{\lambda} = A(\tilde{m})^{-1} b(\tilde{m})$
satisfy
\begin{equation}
    \frac{\|\tilde\lambda - \lambda^*\|}{\|\lambda^*\|}
    \le \kappa(A)\,
    \biggl(\frac{\|\delta A\|}{\|A\|} + \frac{\|\delta b\|}{\|b\|}\biggr)
    + O(\|\delta m\|^2),
\end{equation}
where $\delta A = A(\tilde{m}) - A(m)$,
$\delta b = b(\tilde{m}) - b(m)$,
and $\kappa(A) = \|A\|\,\|A^{-1}\|$.
\end{proposition}

\begin{proof}
Since each entry of $A$ and $b$ is a single moment
(equations~\eqref{eq:A-entry}--\eqref{eq:b-entry}),
both $A(m)$ and $b(m)$ are \emph{linear} in $m$.
Therefore $\delta A = A(\delta m)$ and $\delta b = b(\delta m)$
(using linearity), and
$\|\delta A\| \le C_A \|\delta m\|$,
$\|\delta b\| \le C_b \|\delta m\|$
for constants $C_A$, $C_b$ depending only on the basis.

The perturbed system is $(A + \delta A)(\lambda^* + \delta\lambda) = b + \delta b$.
Expanding and dropping the $O(\|\delta m\|^2)$ term $\delta A \cdot \delta\lambda$:
\begin{equation}
    \delta\lambda = A^{-1}(\delta b - \delta A\, \lambda^*).
\end{equation}
Taking norms:
$\|\delta\lambda\| \le \|A^{-1}\|(\|\delta b\| + \|\delta A\|\,\|\lambda^*\|)$.
Dividing by $\|\lambda^*\|$ and using $\|A^{-1}\| = \kappa(A)/\|A\|$:
\begin{equation}
    \frac{\|\delta\lambda\|}{\|\lambda^*\|}
    \le \kappa(A)\,\frac{\|\delta A\|}{\|A\|}
    + \kappa(A)\,\frac{\|\delta b\|}{\|A\|\,\|\lambda^*\|}
    \le \kappa(A)\,
    \biggl(\frac{\|\delta A\|}{\|A\|} + \frac{\|\delta b\|}{\|b\|}\biggr),
\end{equation}
where the last step uses $\|b\| = \|A\lambda^*\| \le \|A\|\,\|\lambda^*\|$.
\end{proof}

\begin{remark}[Contrast with MaxEnt realizability]\label{rmk:realizability}
The MaxEnt dual is $\min_\lambda \{\lambda \cdot m + \log Z(\lambda)\}$,
a convex program on $\Lambda = \{\lambda : Z(\lambda) < \infty\}$.
On the non-constant minimal parameterization, $\log Z$ is strictly convex on $\operatorname{int}(\Lambda)$, so
the moment map $\lambda \mapsto m(\lambda) = -\nabla_\lambda \log Z(\lambda)$
is a diffeomorphism from $\operatorname{int}(\Lambda)$ onto an open convex set
$\mathcal{M} \subset \R^M$, the \emph{mean parameter space}
(Wainwright and Jordan~\cite{wainwright2008}, Theorem~3.3).
When the target moments $m$ lie outside $\mathcal{M}$
(which can happen when moments are propagated numerically and accumulate
error), the MaxEnt dual has no finite minimizer:
$\|\lambda_k\| \to \infty$ along the optimization iterates.
In practice, this manifests as the optimizer failing to converge,
returning a degenerate density (mass concentrating on domain vertices),
or producing numerical overflow.
In this case the optimization problem has no finite minimizer.

The score matching linear system $A\lambda = b$, by contrast,
has a solution whenever $A$ is nonsingular,
regardless of whether $m$ lies in $\mathcal{M}$.
This is because the score matching objective
$J_{\mathrm{SM}} = \frac{1}{2}\lambda^\top A \lambda - b^\top \lambda$
is a quadratic in $\lambda$ with no domain constraint.
The fitted $\lambda^*$ may not correspond to a normalizable density
when $m \notin \mathcal{M}$, but the score function
$s(x) = -J_\phi(x)\lambda^*$ remains well-defined and can still
be used for Langevin sampling or downstream computations.
\end{remark}

\subsection{Proof of Theorem~\ref{thm:sm-maxent} (score matching and MaxEnt equivalence)}
\label{app:proof-equiv}

We require the following setup and regularity conditions.

\textbf{Setup.}
Let $\mathcal{P}_r = \{p(\cdot;\lambda) =
\exp(-\lambda\cdot\phi(x))/Z(\lambda) : \lambda \in \Lambda\}$
be the degree-$r$ polynomial exponential family, where
$\phi(x) = (x^\alpha)_{|\alpha| \le r}$ is the monomial sufficient statistic
and $\Lambda = \{\lambda \in \R^M : Z(\lambda) < \infty\}$
is the natural parameter space.
Let $\lambda_{\mathrm{true}} \in \operatorname{int}(\Lambda)$ and denote
$p_0 = p(\cdot;\lambda_{\mathrm{true}})$.
Suppose the propagated moments $m_\alpha = \E_{p_0}[x^\alpha]$
for $|\alpha| \le 2r-2$ are the exact moments of $p_0$.

\textbf{Regularity conditions.}
\begin{enumerate}[label=(R\arabic*)]
    \item \label{R1} $\lambda_{\mathrm{true}} \in \operatorname{int}(\Lambda)$
    with $\lambda_{\mathbf{0}}=0$,
    and the support of $p_0$ has nonempty interior. For the polynomial
    family on $\R^n$, $p_0(x)>0$ on all of $\R^n$.
    \item \label{R2} All moments $\E_{p_0}[|x^\alpha|]$ for $|\alpha| \le 2r-2$
    are finite (guaranteed by $\lambda_{\mathrm{true}} \in \operatorname{int}(\Lambda)$
    for exponential families).
    \item \label{R3} The boundary terms in the integration-by-parts
    step of score matching vanish, i.e.,
    $\lim_{\|x\|\to\infty} p_0(x) \, s_\lambda(x) = 0$
    for all $\lambda \in \Lambda$.
    This holds when the leading-degree term of $E_\lambda$ grows
    sufficiently fast (e.g., even degree with positive leading coefficient).
\end{enumerate}

\begin{proof}
We proceed in three steps.

\textbf{Step 1. The Fisher divergence vanishes at $\lambda_{\mathrm{true}}$.}

The Fisher divergence between $p_0$ and $p(\cdot;\lambda)$ is
\begin{equation}\label{eq:app-fisher}
    D_F(p_0 \| p_\lambda)
    = \frac{1}{2}\E_{p_0}\bigl[\|s_0(x) - s_\lambda(x)\|^2\bigr],
\end{equation}
where $s_0 = \nabla_x \log p_0 = -J_\phi \lambda_{\mathrm{true}}$
and $s_\lambda = -J_\phi \lambda$.
Substituting,
\begin{equation}
    D_F(p_0 \| p_\lambda)
    = \frac{1}{2}\E_{p_0}\bigl[\|J_\phi(\lambda - \lambda_{\mathrm{true}})\|^2\bigr]
    = \frac{1}{2}(\lambda-\lambda_{\mathrm{true}})^\top A_0 (\lambda-\lambda_{\mathrm{true}}),
\end{equation}
where $A_0 = \E_{p_0}[J_\phi^\top J_\phi]$ is the Gram matrix
evaluated under $p_0$.
At $\lambda = \lambda_{\mathrm{true}}$, we have $D_F = 0$.

\textbf{Step 2. $A_0$ is positive definite (on the non-constant subspace),
so $\lambda_{\mathrm{true}}$ is the unique minimizer.}

We must show that $A_0 v = 0$ implies $v = 0$
(restricted to $|\alpha| \ge 1$).
$A_0 v = 0$ means $\E_{p_0}[\|J_\phi v\|^2] = 0$, which implies
$J_\phi(x) v = 0$ for $p_0$-a.e.\ $x$.
By condition~\ref{R1}, this polynomial vector field vanishes on a set
with nonempty interior, and hence vanishes identically. Writing
$q_v(x)=\sum_{1\le|\alpha|\le r}v_\alpha x^\alpha$, this says
$\nabla q_v \equiv 0$, so $q_v$ is constant. Since $q_v$ has no constant
term, $q_v\equiv0$, and therefore $v=0$.

This shows $A_0$ is positive definite on $\{v : v_{\mathbf{0}} = 0\}$,
so $D_F(p_0 \| p_\lambda) = 0$ if and only if $\lambda_\alpha = (\lambda_{\mathrm{true}})_\alpha$
for all $|\alpha| \ge 1$.

\textbf{Step 3. The score matching objective has the same minimizer.}

The Hyv\"arinen score matching objective, evaluated under $p_0$, is
\begin{equation}
    J_{\mathrm{SM}}(\lambda)
    = \E_{p_0}\biggl[\sum_i \partial_i (s_\lambda)_i
    + \frac{1}{2}(s_\lambda)_i^2\biggr]
    = D_F(p_0 \| p_\lambda)
    + \E_{p_0}\biggl[\sum_i \partial_i (s_0)_i
    + \frac{1}{2}(s_0)_i^2\biggr],
\end{equation}
where the second term is independent of $\lambda$
(this is the standard identity from Hyv\"arinen~\cite{hyvarinen2005},
valid under~\ref{R3}).
Therefore $J_{\mathrm{SM}}$ and $D_F$ differ by a constant,
and they share the same minimizer.

From Steps 1--2, the unique minimizer of $D_F$
(restricted to $\lambda_{\mathbf{0}} = 0$) is $\lambda_{\mathrm{true}}$.
Therefore $\lambda^*_\alpha = (\lambda_{\mathrm{true}})_\alpha$ for all $|\alpha| \ge 1$.

\textbf{Moment matching follows.}
Since $\lambda^*$ and $\lambda_{\mathrm{true}}$ agree on all non-constant components,
$p(\cdot;\lambda^*) = p(\cdot;\lambda_{\mathrm{true}}) = p_0$,
and the fitted density has exactly the moments $m$.
The uniqueness of $\lambda_{\mathrm{true}}$ within $\mathcal{P}_r$ follows from
the bijectivity of the moment map (Appendix~\ref{app:expfam}).
\end{proof}

\begin{remark}[When the conditions fail]
If the propagated moments $m$ do \emph{not} correspond to any
$p_0 \in \mathcal{P}_r$ (the generic case when the true filtering
density is not in the polynomial exponential family), then
$\lambda^*$ minimizes $D_F$ over $\mathcal{P}_r$ but does not
match $m$ exactly.
The moment mismatch $\|m_{\lambda^*} - m\|$
(where $m_{\lambda^*}$ are the moments of $p(\cdot;\lambda^*)$)
is then controlled by the approximation quality of $\mathcal{P}_r$
for the true density.
Bounding this mismatch in terms of $r$ and the smoothness of the
true density is an interesting direction for future work.
\end{remark}

\subsection{Proof of Proposition~\ref{prop:stein} (Stein identity)}

Stein's identity~\cite{stein1972} states that for a smooth density $p$ with score
$s_i = \partial_i \log p$ and any smooth test function $f$
with $\E[|s_i f|] < \infty$,
\begin{equation}
    \E[s_i(x) f(x) + \partial_i f(x)] = 0,
\end{equation}
provided $\lim_{\|x\|\to\infty} p(x)f(x) = 0$
(the boundary term from integration by parts vanishes).
This follows from
$\int (s_i f + \partial_i f) p\,dx
= \int (\partial_i p \cdot f / p + \partial_i f) p\,dx
= \int \partial_i(p f)\,dx = [pf]_{\text{boundary}} = 0$.

Substituting $s_i = -\sum_\alpha \lambda_\alpha \alpha_i x^{\alpha-e_i}$
(the polynomial MED score) and $f(x) = x^\beta$,
\begin{equation}
    \E\Bigl[-\sum_\alpha \lambda_\alpha \alpha_i x^{\alpha+\beta-e_i}
    + \beta_i x^{\beta-e_i}\Bigr] = 0.
\end{equation}
Rearranging gives~\eqref{eq:stein}.
\qed

\subsection{Well-posedness of the Stein closure system}
\label{app:stein-wellposed}

The one-layer Stein closure assembles a linear system
$\Lambda_1(\lambda)\,m^{(K+1)} = c(\lambda, m^{(\le K)})$
for the degree-$(K{+}1) = (2r{-}1)$ unknowns.
We count the rows (equations) and columns (unknowns) for
general $n$ and $r$.

\textbf{Unknowns.}
The number of monomials of degree exactly $2r{-}1$ in $n$ variables is
$U(n,r) = \binom{2r{+}n{-}2}{n{-}1}$.

\textbf{Equations (first layer, $|\beta| = r$ only).}
The implemented directional subsystem uses each $(\beta, i)$ pair with
$|\beta| = r$ and $\beta_i \ge 1$.
For a given $i$, the number of such $\beta$ is
$\binom{r{+}n{-}2}{n{-}1}$
(substitute $\beta_i' = \beta_i - 1$).
Summing over $i$ gives
$R_1(n,r) = n\binom{r{+}n{-}2}{n{-}1}$.
For $n = 2$, $R_1 = 2r = U$, so the system is square.
For $n \ge 3$, $R_1 < U$ and the first layer alone is underdetermined.
The component-wise Stein identity is also valid when $\beta_i=0$. Those
rows have zero right-hand side and may be appended as additional
least-squares constraints. The counts in this appendix are therefore
conservative counts for the directional subsystem used in the reported
implementation.

\begin{proposition}[Well-posedness of the directional Stein closure for $n{=}2$]
\label{prop:stein-closure}
For $n{=}2$ and any $r \ge 2$, the directional Stein closure subsystem is square
($2r$ equations for $2r$ unknowns). Its determinant is a nonzero
polynomial in the degree-$r$ score coefficients. Hence the closure is
well-posed for all $\lambda$ outside a proper algebraic variety of
Lebesgue measure zero.
For $r=2$, nonsingularity holds whenever the precision matrix
$\Omega \succ 0$.
\end{proposition}

\begin{proof}
Each entry of $\Lambda_1$ is a monomial in $\lambda$
(specifically $\lambda_\alpha \cdot \alpha_i$),
so every $k \times k$ minor is a polynomial in $\lambda$.
The rank-deficient locus is the zero set of all maximal minors,
which is an algebraic variety.
It remains to show that this variety is proper.
Set the only nonzero degree-$r$ coefficients to
$\lambda_{(r,0)}$ and $\lambda_{(0,r)}$.
Then the equations from the $x_1$-derivative solve the degree-$(2r{-}1)$
moments with $x_1$-exponent at least $r$, while the equations from the
$x_2$-derivative solve the remaining moments with $x_2$-exponent at least
$r$. After ordering the unknowns this matrix is diagonal, with nonzero
diagonal entries $r\lambda_{(r,0)}$ and $r\lambda_{(0,r)}$.
Thus $\det \Lambda_1$ is not the zero polynomial, and its zero set has
Lebesgue measure zero.
At $r{=}2$, the four Stein equations at $|\beta|{=}2$ for the
four degree-3 moments reduce to the Gaussian moment--cumulant
relation $\kappa_3 = 0$, uniquely solvable when $\Omega \succ 0$.
\end{proof}

\textbf{Augmented directional equations ($|\beta| \in [r, K]$).}
Stein equations at $|\beta| = r{+}j$ ($j = 1, \ldots, r{-}2$)
also contribute to the degree-$(2r{-}1)$ unknowns
via $|\alpha| = r{-}j$ terms.
Terms with $|\alpha|<r{-}j$ involve already tracked moments, while
terms with $|\alpha|>r{-}j$ involve degrees above $K{+}1$ and are
omitted in the truncated closure solve.
The augmented row count is
\[
    R(n,r) = \sum_{j=0}^{r-2} n\binom{r{+}j{+}n{-}2}{n{-}1}.
\]
Asymptotically, $R/U \sim (2r{-}1)(2r{-}2)/n$,
so the crossover occurs at a dimension of order
$(2r{-}1)(2r{-}2)$.
For the regimes used in the experiments, the exact counts are concrete:
at $r{=}3$, $R/U>1$ through $n{=}15$ and $R/U<1$ starting at $n{=}16$
($R/U = 1.37$ at $n{=}10$, $1.09$ at $n{=}14$,
and $0.45$ at $n{=}40$).
At $r{=}4$, the crossover is between $n{=}35$ and $n{=}36$.

\textbf{Extended directional equations ($|\beta| \in [r, K{+}1]$).}
The Stein identity at $|\beta| = K{+}1$ also involves
degree-$(K{+}1)$ unknowns (through the $|\alpha|=1$ terms),
while the right-hand side $\beta_i\,m_{\beta-e_i}$ involves
only degree-$K$ moments (still tracked).
Including these equations adds
$n\binom{K{+}n{-}1}{n{-}1}$ rows to the system.
Higher-degree moments (degree $\ge K{+}2$) that appear in
these equations through $|\alpha| \ge 2$ terms are truncated
(set to zero). For near-Gaussian densities this truncation
error is $O(\sigma^{K+2})$.

The extended row count is
\[
    R_{\mathrm{ext}}(n,r) = \sum_{j=0}^{r-1}
    n\binom{r{+}j{+}n{-}2}{n{-}1}.
\]
The additional $j = r{-}1$ term (from $|\beta| = K{+}1 = 2r{-}1$)
contributes the dominant number of equations at large $n$.
At $n{=}40$, $r{=}3$: $R_{\mathrm{ext}} = 5{,}428{,}400$
for $U = 1{,}086{,}008$ unknowns ($5.00\times$ overdetermined),
compared to $R = 492{,}000$ ($0.45\times$, underdetermined) without
the extension.
This restores overdetermination at large $n$ at the cost of
truncating higher-degree moments in the extended equations.

\textbf{Active closure for structured generators.}
The counts above are for closing \emph{all} degree-$(K{+}1)$ moments.
In a concrete moment ODE, only moments that actually appear in the
generator must be closed. For the coupled oscillator network at
$r{=}3$ ($K{=}4$), the only unclosed terms have the form
$m_{\alpha-e_{p_j}+2e_{q_j}}$ from the quadratic force
$-\beta q_j^2$ in $dp_j$.
At $n{=}20$, the full degree-five closure would have
$\binom{24}{5}=42{,}504$ unknowns and the standard augmented system is
underdetermined. Restricting to the dynamically active targets leaves
$14{,}500$ unknown moments and $32{,}800$ Stein equations, a
$2.26\times$ overdetermined system. This is the active closure used in
the high-dimensional coupled oscillator experiments. This also explains
why the $n{=}16,18,20$ runs use $r{=}3$ rather than $r{=}4$: the
full-basis heuristic counts moments that the coupled-oscillator generator
never requests, while $r{=}4$ would increase the score basis at $n{=}20$
from $\binom{23}{3}=1771$ to $\binom{24}{4}=10{,}626$ functions and
the propagated moment budget from degree $4$ to degree $6$.

\begin{remark}[Empirical well-posedness]
\label{rem:stein-extended}
In all full-basis experiments with $n \le 10$ and $r \le 8$
(using the standard augmented system, $|\beta| \in [r, K]$),
$\Lambda_1$ had full column rank at every time step
and we have not encountered a rank-deficient instance.
For the coupled oscillator scaling sweep through $n{=}20$, the
active closure systems were overdetermined after restricting to the
moments requested by the generator. For larger dense closures where the
standard system becomes underdetermined, the extended system
($|\beta|$ up to $K{+}1$) can restore overdetermination.
A rigorous truncation error bound for the extended equations
as a function of the departure from Gaussianity remains open.
\end{remark}

\section{The Fokker--Planck equation and the score evolution PDE}
\label{app:score-pde}

The SKF propagates finitely many moments rather than the full
density.
This section derives the underlying PDEs (Fokker--Planck equation and the score evolution equation) governing the density and its score, providing the continuous-time
foundation for the moment propagation via Dynkin's formula
used in Section~\ref{sec:stein}.

\subsection{The Fokker--Planck equation}\label{app:fp}

The probability density $p(t,x)$ of the state $X_t$ satisfying
the SDE~\eqref{eq:sde} evolves according to the
\emph{Fokker--Planck / forward Kolmogorov} equation~\cite{oksendal2003}:
\begin{equation}\label{eq:app-fp}
    \frac{\partial p}{\partial t}
    = -\sum_i \frac{\partial}{\partial x_i}\bigl(X_i\, p\bigr)
    + \sum_{i,j} \frac{\partial^2}{\partial x_i\,\partial x_j}\bigl(H_{ij}\, p\bigr),
\end{equation}
or in compact notation,
$\partial_t p = -\nabla\cdot(Xp) + \Tr(H\nabla^2 p)$
(for constant $H = \frac{1}{2}hh^\top$).
This is a \emph{linear} PDE in $p$, but it lives in an
infinite-dimensional function space, meaning $p(t, \cdot)$ must be tracked
as a full field over $\R^n$.
Direct numerical computation by finite-difference or finite-element
methods scales as $O(G^n)$ for grid resolution $G$, the same
exponential cost as the partition function~\eqref{eq:med}.

\subsection{Derivation of the moment ODE (Dynkin's formula)}

The moment ODE~\eqref{eq:moment-ode} arises~\cite[Ch.~7]{oksendal2003} by testing the
Fokker--Planck equation against monomials $\phi_\alpha(x) = x^\alpha$.
Multiplying~\eqref{eq:app-fp} by $x^\alpha$ and integrating over $\R^n$,
\begin{equation}
    \frac{d}{dt} m_\alpha
    = \frac{d}{dt} \int x^\alpha \, p(t,x)\, dx
    = \int x^\alpha \, \partial_t p(t,x)\, dx.
\end{equation}
Substituting the Fokker--Planck equation and integrating by parts
(assuming $p$ and its derivatives decay sufficiently fast at infinity),
\begin{align}
    \int x^\alpha \bigl[-\nabla\cdot(Xp)\bigr] dx
    &= \int \nabla(x^\alpha) \cdot X(x)\, p(x)\, dx
    = \E\bigl[\nabla\phi_\alpha \cdot X\bigr], \\
    \int x^\alpha \sum_{i,j}\partial_i\partial_j(H_{ij}p)\, dx
    &= \int \sum_{i,j} H_{ij}(x)\,\partial_i\partial_j(x^\alpha)\, p(x)\, dx
    = \E\bigl[\Tr(H(X_t)\nabla^2\phi_\alpha)\bigr].
\end{align}
The first equality in each line uses integration by parts
(once for the drift term, twice for the diffusion term),
transferring the derivatives from $p$ onto the test function $x^\alpha$.
Combining both terms gives
\begin{equation}
    \frac{d}{dt} m_\alpha
    = \E\bigl[\nabla\phi_\alpha \cdot X + \Tr(H\nabla^2\phi_\alpha)\bigr]
    = \E\bigl[\mathcal{A}\phi_\alpha(X_t)\bigr],
\end{equation}
which is the moment ODE~\eqref{eq:moment-ode}.
This reduces the infinite-dimensional Fokker--Planck PDE to a
system of ODEs for the moments,
at the cost of the closure problem discussed in
Section~\ref{sec:background}.

\subsection{The score evolution PDE}\label{app:score-evo}

An alternative to tracking $p$ is to track the
\emph{score} $s(t,x) = \nabla_x \log p(t,x)$.
We derive its evolution equation from~\eqref{eq:app-fp}
and show that it is equally or more intractable than the Fokker--Planck equation.

\textbf{Step 1: Time derivative of $\log p$.}
The Fokker--Planck equation $\partial_t p = -\nabla \cdot (Xp) + \Tr(H\nabla^2 p)$
gives, after dividing by $p$ and using $\nabla p / p = s$:
\begin{equation}\label{eq:app-dt-logp}
    \partial_t \log p = -(\nabla \cdot X) - X \cdot s + s^\top H s + \Tr(HS),
\end{equation}
where $S = \nabla s = \nabla^2 \log p$ is the \emph{score Hessian}.
To obtain~\eqref{eq:app-dt-logp}, we used
$\partial_i\partial_j p / p = s_i s_j + S_{ij}$
(from differentiating $\partial_i p = p\, s_i$).

\textbf{Step 2: Gradient to get $\partial_t s$.}
Since $s = \nabla \log p$, we have $\partial_t s = \nabla(\partial_t \log p)$.
Taking $\nabla$ of each term in~\eqref{eq:app-dt-logp}:
\begin{align}
    \nabla(-X \cdot s) &= -J_X^\top s - S X, \label{eq:app-term2} \\
    \nabla(s^\top H s) &= 2\, S H s, \label{eq:app-term3} \\
    [\nabla\Tr(HS)]_k &= \textstyle\sum_{ij} H_{ij}\, T_{ijk},
    \quad T_{ijk} := \partial_k S_{ij}, \label{eq:app-term4}
\end{align}
where $J_X$ is the Jacobian of $X$ ($[J_X]_{ij} = \partial_j X_i$)
and $T$ is the \emph{third-order score tensor}.

\textbf{Result: score evolution PDE.}
\begin{equation}\label{eq:app-score-pde}
    \boxed{
    \partial_t s = -\nabla(\nabla \cdot X) - J_X^\top s - SX + 2SHs
    + \nabla[\Tr(HS)].
    }
\end{equation}

\textbf{Why this is intractable.}
The PDE~\eqref{eq:app-score-pde} for the score $s$ (a vector field)
is not self-contained.
The term $2SHs$ involves the score Hessian $S_{ij} = \partial_i s_j$
(a matrix field), and the term $\nabla[\Tr(HS)]$ involves the third-order
tensor $T_{ijk} = \partial_k S_{ij}$ (the spatial derivative of the Hessian).
To evolve $s$, one must know $S$. To evolve $S$, one must know $T$.
To evolve $T$, one must know the fourth derivative of $\log p$,
and so on.
This is an \emph{infinite hierarchy of coupled PDEs} in which each level
requires the next, structurally analogous to the moment
hierarchy~\eqref{eq:moment-ode} where each moment order
couples to the next.
Truncating the score hierarchy at any finite order introduces
an uncontrolled approximation, just as truncating the moment
hierarchy does without a proper closure.

Directly propagating the score therefore does not avoid the
closure problem. It merely reformulates it in score space.
Our approach works with finite-dimensional moment equations rather than
solving either the Fokker--Planck PDE or the score PDE directly. It does
so by propagating \emph{moments} via the
ODE~\eqref{eq:moment-ode} and using score matching~\eqref{eq:lambda-star}
only for density reconstruction from those moments.

\section{Linear-Gaussian specialization and recovery of the information-form Kalman filter}
\label{app:gaussian}

We show that at $r = 2$, Algorithm~\ref{alg:skf} reduces exactly to the
information-form Kalman filter, thus establishing the SKF as a strict
generalization of the classical linear filter.

\subsection{The information form of the Kalman filter}\label{app:info-form}

The standard Kalman filter maintains the mean $\mu$ and covariance $P$
of a Gaussian belief.
The \emph{information form}  is reparameterized in terms of the
\emph{precision} (information) matrix $\Omega := P^{-1}$ and the
\emph{information vector} $\eta := \Omega\mu = P^{-1}\mu$.
The Gaussian density is:
\begin{equation}\label{eq:gauss-info}
    p(x) = \frac{|\Omega|^{1/2}}{(2\pi)^{n/2}}
    \exp\!\Bigl(-\frac{1}{2}(x-\mu)^\top \Omega\,(x-\mu)\Bigr)
    \propto \exp\!\Bigl(-\frac{1}{2}x^\top\Omega\, x + \eta^\top x\Bigr).
\end{equation}
The information form is preferred in multi-sensor fusion because
the update is additive (no matrix inversion), while the prediction
requires inverting $\Omega$, the opposite of the covariance form.

\subsection{Correspondence between SKF and information-form parameters}\label{app:skf-info-corr}

At $r = 2$, the polynomial MED is
$p(x;\lambda) \propto \exp(-\lambda_{\mathbf{0}} - \lambda_1^\top x - x^\top \Lambda_2 x)$,
where $\lambda_{\mathbf{0}} \in \R$ is the constant component, $\lambda_1 \in \R^n$ collects the degree-1 parameters, and
$\Lambda_2 \in \R^{n\times n}_{\mathrm{sym}}$ collects the degree-2 parameters.
Comparing with~\eqref{eq:gauss-info}:
\begin{equation}\label{eq:lambda-info-map}
    \Omega = 2\Lambda_2, \qquad \eta = -\lambda_1, \qquad
    \mu = \Omega^{-1}\eta = -\frac{1}{2}\Lambda_2^{-1}\lambda_1.
\end{equation}
The score is $s(x) = -\nabla(\lambda \cdot \phi) = -\lambda_1 - 2\Lambda_2 x
= \eta - \Omega x = -\Omega(x - \mu)$,
confirming that the score of a Gaussian is affine in $x$ with
slope $-\Omega$ (negative precision).
The score Hessian is $S = \nabla s = -\Omega$ (constant),
and all third-order derivatives vanish ($T_{ijk} = 0$).

\subsection{Score matching recovers the information parameters}\label{app:sm-recovers}

At $r = 2$ in 1D (for clarity), the score matching system
$A\lambda = b$ (Proposition~\ref{prop:sm}) with basis $(x, x^2)$
and moments $m_0 = 1$, $m_1 = \mu$, $m_2 = \mu^2 + P$ gives:
\[
    \begin{pmatrix} 1 & 2m_1 \\ 2m_1 & 4m_2 \end{pmatrix}
    \begin{pmatrix} \lambda_1 \\ \lambda_2 \end{pmatrix}
    = \begin{pmatrix} 0 \\ 2 \end{pmatrix}.
\]
Solving this yields
\begin{equation}
    \lambda_1 = -\frac{m_1}{m_2 - m_1^2} = -\frac{\mu}{P} = -\eta,
    \qquad
    \lambda_2 = \frac{1}{2(m_2 - m_1^2)} = \frac{1}{2P} = \frac{\Omega}{2}.
\end{equation}
Hence, score matching recovers the information vector $\eta$ and the precision matrix $\Omega$ exactly
from the first two moments.

\subsection{Prediction step reduces to Riccati}\label{app:riccati}

For affine dynamics $dx = (Ax+b)\,dt + h\,dW$ with $H = \frac{1}{2}hh^\top$ constant,
the score PDE~\eqref{eq:app-score-pde} with $S = -\Omega$, $T = 0$,
$J_X = A$, $\nabla(\nabla\cdot X) = 0$ reduces to
$\partial_t s = -A^\top s + \Omega(Ax+b) - 2\Omega Hs$.
Substituting $s = -\Omega(x - \mu)$ and matching the coefficient of
$(x - \mu)$ and the constant term:
\begin{align}
    \dot{\Omega} &= -A^\top \Omega - \Omega A - 2\Omega H\Omega,
    \label{eq:app-riccati-omega} \\
    \dot{\mu} &= A\mu + b. \label{eq:app-mean-ode}
\end{align}
Equation~\eqref{eq:app-riccati-omega} is the \textbf{continuous-time
information-form Riccati equation}.
Using $\Omega = P^{-1}$ and $\dot{\Omega} = -P^{-1}\dot{P}P^{-1}$,
one recovers the standard covariance Riccati
$\dot{P} = AP + PA^\top + 2H$.

In terms of the information vector $\eta = \Omega\mu$:
\begin{equation}\label{eq:app-riccati-eta}
    \dot{\eta} = \dot{\Omega}\mu + \Omega\dot{\mu}
    = (-A^\top\Omega - \Omega A - 2\Omega H\Omega)\mu + \Omega(A\mu + b)
    = -(A^\top + 2\Omega H)\eta + \Omega b.
\end{equation}

\subsection{Update step recovers information-form Kalman}\label{app:kalman-update}

At measurement time with $z = Cx + v$, $v \sim \mathcal{N}(0, R)$,
the likelihood score is
$\nabla_x \log p(z|x) = C^\top R^{-1}(z - Cx)
= C^\top R^{-1}z - C^\top R^{-1}Cx$.
The score addition $s^+ = s^- + \nabla_x \log p(z|x)$ gives:
\[
    \underbrace{-\Omega^+ x + \eta^+}_{s^+}
    = \underbrace{-\Omega^- x + \eta^-}_{s^-}
    + \underbrace{-C^\top R^{-1}C\, x + C^\top R^{-1}z}_{\nabla_x \log p(z|x)}.
\]
Matching coefficients, we obtain
\begin{equation}\label{eq:info-kalman-update}
    \boxed{
    \Omega^+ = \Omega^- + C^\top R^{-1}C, \qquad
    \eta^+ = \eta^- + C^\top R^{-1}z.
    }
\end{equation}
These are \textbf{exactly the information-form Kalman filter update equations}.
The posterior mean is $\mu^+ = (\Omega^+)^{-1}\eta^+$.

In the SKF framework, the same update is written as
$\lambda^+ = \lambda^- + \lambda_{\mathrm{lik}}$
(eq.~\eqref{eq:conjugate}), where
$\lambda_{\mathrm{lik}} = (-C^\top R^{-1}z,\; \frac{1}{2}C^\top R^{-1}C)$
in the $(\lambda_1, \Lambda_2)$ parameterization.
The correspondence is exact, and the SKF conjugate update \emph{is} the
information-form Kalman update, expressed in natural parameters.
We summarize the preceding calculations as follows.

\begin{proposition}[Kalman filter recovery]\label{thm:kf}
Consider an affine-Gaussian system with dynamics
$dx = (Ax+b)\,dt + h\,dW$, where $A \in \R^{n \times n}$, $b \in \R^n$,
and $h \in \R^{n \times n_w}$ are constant, together with an affine
measurement model $z = Cx + v$, $v \sim \mathcal{N}(0, R)$.
Then Algorithm~\ref{alg:skf} at $r = 2$ produces the same
precision matrix $\Omega(t)$ and information vector $\eta(t)$
as the continuous-discrete information-form Kalman filter.
\end{proposition}

\begin{proof}
Each step is verified in the preceding subsections. Score matching recovers
$(\eta, \Omega)$ from the first two moments (Section~\ref{app:sm-recovers}),
the score PDE reduces to the information-form Riccati equation
(Section~\ref{app:riccati}),
and the conjugate score update reproduces the information-form
Kalman measurement update (Section~\ref{app:kalman-update}).
\end{proof}

The comparison between the SKF and the information-form Kalman filter is summarized in the table below.

\begin{center}
\renewcommand{\arraystretch}{1.2}
\small
\begin{tabular}{ll}
\toprule
\textbf{SKF quantity} & \textbf{Information-form Kalman} \\
\midrule
Score $s(x) = -J_\phi\lambda$ & $-\Omega(x-\mu) = -\Omega x + \eta$ \\
Score Hessian $S = -\Omega$ & Negative precision \\
Score matching solve & $(\Omega, \eta)$ from $(m_1, m_2)$ \\
$\lambda^+ = \lambda^- + \lambda_{\mathrm{lik}}$ (conjugate) &
  $\Omega^+ = \Omega^- + C^\top R^{-1}C$, $\eta^+ = \eta^- + C^\top R^{-1}z$ \\
Score PDE~\eqref{eq:app-score-pde} & Riccati ODE for $\Omega$ \\
Stein identity (trivial at $r=2$) & Moment-covariance relation \\
\bottomrule
\end{tabular}
\end{center}

\section{Iterative consistency refinement via Stein's method of moments}
\label{app:smom}

The conjugate measurement update
$\lambda^+ = \lambda^- + \lambda_{\mathrm{lik}}$
(Section~\ref{sec:update})
produces posterior natural parameters $\lambda^+$, but the subsequent
Stein moment recovery (Section~\ref{sec:recovery}) yields only
approximate moments $\hat{m}^+$ due to truncation.
In general, $\hat{m}^+$ and $\lambda^+$ are not mutually consistent:
re-fitting score matching to $\hat{m}^+$ produces
$\hat\lambda \ne \lambda^+$.
This section formalizes the iterative refinement that enforces
consistency, drawing on recent work connecting score matching to
Stein's method of moments (SMoM).

\subsection{Background: Stein's method of moments}

Stein's method~\cite{stein1972} characterizes a distribution
$p(x;\theta)$ through an operator $\mathcal{A}_\theta$ --
called a \emph{Stein operator} -- that satisfies
$\E_{p(\cdot;\theta)}[\mathcal{A}_\theta g(X)] = 0$
for a sufficiently rich class of test functions $g$.
Ebner et al.~\cite{ebner2025} developed Stein's method of moments
(SMoM), which derives parameter estimators by solving the empirical
Stein equations for selected test functions.
For the polynomial MED
$p(x;\lambda) \propto \exp(-\lambda \cdot \phi(x))$,
the divergence-based Stein operator is
$\mathcal{A}_\lambda g(x)
= \nabla \cdot g(x) - (\nabla E(x;\lambda)) \cdot g(x)$,
where $E(x;\lambda) = \lambda \cdot \phi(x)$ is the energy.
The condition $\E[\mathcal{A}_\lambda g] = 0$ reduces to the
Stein identity (Proposition~\ref{prop:stein}) via integration by parts.

An SMoM estimator solves the empirical Stein equations
$\frac{1}{n}\sum_{i=1}^n \mathcal{A}_\theta g_k(X_i) = 0$
for selected test functions $\{g_k\}$.
Different choices of $g_k$ yield different estimators with different
asymptotic variances.

\subsection{Score matching as the center of SMoM}

Kume and Walker~\cite{kume2026} proved that the score matching
estimator $\hat\lambda_{\mathrm{SM}}$ is a special case of SMoM,
corresponding to the choice $g_k(x) = \nabla \phi_k(x)$.
Nagai and Yano~\cite{nagai2026} showed that SMoM estimators admit a
canonical decomposition centered at score matching. In their setting,
\begin{equation}\label{eq:smom-decomp}
  \hat\lambda_{\mathrm{SMoM}}
  = \hat\lambda_{\mathrm{SM}} + \Delta(\hat\lambda_{\mathrm{SM}}),
\end{equation}
where $\Delta$ is an augmentation term that depends on additional
Stein equations beyond the score matching ones.
They use this decomposition to construct SMoM augmentations that can
improve the asymptotic variance of score matching.

\subsection{Application to the SKF measurement update}

In the SKF, we do not estimate $\lambda$ from samples.
Instead, after the conjugate update $\lambda^+ = \lambda^- + \lambda_{\mathrm{lik}}$,
we need to find moments $m^+$ consistent with $\lambda^+$.
The Stein identity with $\lambda = \lambda^+$ provides a linear system
$\Lambda_1(\lambda^+) m^+ = c_1(\lambda^+)$
(equation~\eqref{eq:stein-recovery}), but the truncated system is
approximate.

The iterative refinement enforces a fixed-point condition on the
$(\lambda, m)$ pair.
Define the map $\Phi\colon \lambda \mapsto m \mapsto \lambda'$:
\begin{enumerate}[label=(\roman*), nosep]
  \item $m = \mathrm{SteinRecover}(\lambda)$:
        solve the truncated Stein system~\eqref{eq:stein-recovery}
        for the moments.
  \item $\lambda_{\mathrm{refit}} = \mathrm{SM}(m)$:
        re-fit the score matching parameters to the recovered moments.
  \item $\lambda' = \lambda_{\mathrm{refit}} + \lambda_{\mathrm{lik}}$:
        re-apply the measurement update.
\end{enumerate}
A fixed point $\lambda^*$ of the iteration
$\lambda^{(k+1)} = \Phi(\lambda^{(k)})$ satisfies
$\mathrm{SM}(\mathrm{SteinRecover}(\lambda^*)) + \lambda_{\mathrm{lik}} = \lambda^*$,
i.e., score matching applied to the moments recovered from $\lambda^*$
returns $\lambda^* - \lambda_{\mathrm{lik}}$.

In the language of~\cite{nagai2026}, step (ii) projects back to the
score matching solution (the ``center'' of the SMoM family),
and step (iii) re-applies the measurement.
The net effect is that each iteration reduces the inconsistency
introduced by truncation in step (i).

\textbf{Empirical convergence.}
We do not yet have a convergence proof for this iteration.
Empirically, in the full-basis experiments ($n \le 10$, $r \le 8$),
2--3 iterations reduce the Stein recovery residual by 2--3 orders of
magnitude (from ${\sim}10^{-5}$ to ${\sim}10^{-8}$ on the Duffing
oscillator benchmark). In the active-closure coupled oscillator sweep
through $n{=}20$, one refinement step was used for speed, and we have
not observed divergence.
The per-iteration cost is one reassembled score-matching solve
($O(M^3)$ dense, lower if a structured or regularized solve is used)
and one Stein recovery solve. In the reported runs this refinement cost
is smaller than the prediction step, but it is not included in the
stand-alone score-matching $O(M^3)$ density-fit cost.

\section{Degree accounting and moment budget}
\label{app:degree}

This appendix records the degree bookkeeping used throughout the
experiments. Suppose the drift has polynomial degree $d_X$ and the
diffusion coefficient has polynomial degree $d_h$. Applying the generator
to a monomial of total degree $k$ gives two possible sources of degree
growth. The drift term differentiates once and then multiplies by $X$,
so it reaches degree $k+d_X-1$. The diffusion term differentiates twice
and multiplies by $h h^\top$, so it reaches degree $k+2d_h-2$.
Thus the moment equation for degree $k$ can involve moments up to
degree $k+\bar d$, where
\[
    \bar{d} = \max(d_X{-}1,\,2d_h{-}2).
\]
This number is the amount by which the moment hierarchy reaches beyond
the moments currently being propagated. It therefore gives the number
of Stein closure layers needed. The table gives the common cases used
in the examples.

\begin{center}
\renewcommand{\arraystretch}{1.3}
\begin{tabular}{p{4cm}cccc}
\toprule
\textbf{System type} & $d_X$ & $d_h$ & $\bar{d}$ & \textbf{Closure} \\
\midrule
Linear + additive noise & 1 & 0 & 0 & Exact \\
Linear + linear mult.\ noise & 1 & 1 & 0 & Exact \\
Quadratic + additive & 2 & 0 & 1 & 1 layer \\
Bilinear & 2 & 0--1 & 1 & 1 layer \\
Cubic + additive & 3 & 0 & 2 & 2 layers \\
\bottomrule
\end{tabular}
\end{center}

\textbf{Moment budget.}
Score matching (Proposition~\ref{prop:sm}) requires moments up to degree
$K = 2r{-}2$.
This fixes the natural propagated budget for an order-$r$ polynomial
score model. During prediction, however, the ODE for the degree-$K$
moments may request moments beyond this budget, up to degree
$K+\bar d$. The Stein closure is applied one degree at a time: the first
layer estimates degree $K+1$ moments from moments of degree $\le K$ and
the current score parameters, the next layer estimates degree $K+2$ from
the enlarged collection, and so on. When $\bar d=0$, no higher moments
are requested and the moment hierarchy closes exactly at the tracked
degree. When $\bar d>0$, the prediction step needs exactly $\bar d$
sequential closure layers.

\section{Generator calculations for the experimental systems}
\label{app:generators}

\subsection{Duffing oscillator}

\textbf{SDE.}
State $x = (x_1, x_2)$.
\begin{align}
    dx_1 &= x_2\, dt, \nonumber\\
    dx_2 &= (-\delta x_2 - \alpha x_1 - \beta x_1^2)\, dt
           + \sigma\, dW,
\end{align}
with parameters $\delta$ (damping), $\alpha$ (linear stiffness),
$\beta$ (quadratic stiffness), and $\sigma$ (noise intensity).

\textbf{Drift and diffusion classification.}
\begin{align}
    X(x) &= (x_2,\; -\delta x_2 - \alpha x_1 - \beta x_1^2)^\top,
    \nonumber\\
    h &= (0,\; \sigma)^\top.
\end{align}
Since $X$ has degree $d_X = 2$ (from the $\beta x_1^2$ term)
and $h$ is constant ($d_h = 0$), we have
$\bar{d} = \max(1, -2) = 1$.

\textbf{Generator.}
Applied to $\phi_\alpha = x_1^a x_2^b$:
\begin{align}
    \mathcal{A}\phi_\alpha
    &= \nabla\phi_\alpha \cdot X + \Tr(H\nabla^2\phi_\alpha)
    \nonumber\\
    &= a\, x_1^{a-1} x_2^b \cdot x_2
    + b\, x_1^a x_2^{b-1} \cdot
    (-\delta x_2 - \alpha x_1 - \beta x_1^2)
    + \frac{\sigma^2}{2}\, b(b{-}1)\, x_1^a x_2^{b-2}
    \nonumber\\
    &= a\, x_1^{a-1} x_2^{b+1}
    - b\delta\, x_1^a x_2^b
    - b\alpha\, x_1^{a+1} x_2^{b-1}
    \nonumber\\
    &\quad
    - b\beta\, x_1^{a+2} x_2^{b-1}
    + \frac{\sigma^2}{2}\, b(b{-}1)\, x_1^a x_2^{b-2}.
\end{align}
All terms have degree $|\alpha| = a + b$ except
$x_1^{a+2} x_2^{b-1}$, which has degree $|\alpha| + 1$.
This is the unclosed term that the Stein closure
(Section~\ref{sec:stein}) resolves.

\subsection{Lotka-Volterra (stochastic predator-prey)}

\textbf{SDE.}
State $x = (x_1, x_2)$ (prey, predator).
\begin{align}
    dx_1 &= (\alpha x_1 - \beta x_1 x_2)\, dt + \sigma_1\, dW_1,
    \nonumber\\
    dx_2 &= (-\gamma x_2 + \delta x_1 x_2)\, dt + \sigma_2\, dW_2,
\end{align}

\textbf{Drift and diffusion classification.}
Bilinear drift $X(x) = (\alpha x_1 - \beta x_1 x_2,\;
-\gamma x_2 + \delta x_1 x_2)^\top$,
$d_X = 2$ (from $x_1 x_2$), $d_h = 0$, $\bar{d} = 1$.

\textbf{Generator.}
Applied to $\phi_\alpha = x_1^a x_2^b$:
\begin{align}
    \mathcal{A}\phi_\alpha
    &= a\alpha\, x_1^{a} x_2^{b}
    - a\beta\, x_1^{a} x_2^{b+1}
    - b\gamma\, x_1^{a} x_2^{b}
    + b\delta\, x_1^{a+1} x_2^{b}
    \nonumber\\
    &\quad
    + \frac{\sigma_1^2}{2}\, a(a{-}1)\, x_1^{a-2} x_2^{b}
    + \frac{\sigma_2^2}{2}\, b(b{-}1)\, x_1^{a} x_2^{b-2}.
\end{align}
Taking expectations: $\frac{d}{dt}m_{(a,b)} = \E[\mathcal{A}\phi_\alpha]$.
All terms yield moments of degree $|\alpha|$ except
$x_1^{a} x_2^{b+1}$ and $x_1^{a+1} x_2^{b}$, which have
degree $|\alpha| + 1$. These are the unclosed terms resolved
by Stein closure.

\textbf{Centered coordinates.}
In strictly mean-centered coordinates $z=x-\mu(t)$, the reference evolves
as $\dot{\mu}=\E[X(x)]$, not generally as $X(\mu)$ for nonlinear drift.
Thus the centered drift is $X(z+\mu)-\E[X(x)]$ and has zero expectation.
For Lotka--Volterra this gives
\begin{align}
    dz_1 &= \bigl[(\alpha - \beta\mu_2) z_1 - \beta\mu_1 z_2
           - \beta (z_1 z_2 - m^z_{(1,1)})\bigr]\, dt + \sigma_1\, dW_1,
    \nonumber\\
    dz_2 &= \bigl[\delta\mu_2 z_1 + (-\gamma + \delta\mu_1) z_2
           + \delta (z_1 z_2 - m^z_{(1,1)})\bigr]\, dt + \sigma_2\, dW_2.
\end{align}
The linear coefficients $(\alpha - \beta\mu_2)$ etc.\ depend on
the current mean $\mu(t)$ and change at each time step.
The purely quadratic terms $\beta z_1 z_2$ and $\delta z_1 z_2$
remain unclosed at $|\alpha| = K$ and are resolved by Stein closure.

\subsection{Coupled Duffing oscillators}

\textbf{SDE.}
$N$ oscillators with positions $q_i$ and momenta $p_i$,
$i = 1, \ldots, N$. State dimension $n = 2N$.
\begin{align}
    dq_i &= p_i\, dt,
    \nonumber\\
    dp_i &= \bigl[-\gamma p_i - \alpha q_i - \beta q_i^2
           + \kappa(q_{i+1} {-} 2q_i {+} q_{i-1})\bigr]\, dt
           + \sigma\, dW_i,
\end{align}
with free boundary conditions ($q_0 = q_1$, $q_{N+1} = q_N$).

\textbf{Drift and diffusion classification.}
\begin{align}
    X_i(x) &= p_i, \quad i = 1, \ldots, N \quad\text{(kinematics)},
    \nonumber\\
    X_{N+j}(x) &= -\gamma p_j - \alpha q_j - \beta q_j^2
    + \kappa(q_{j+1} {-} 2q_j {+} q_{j-1}),
    \nonumber\\
    h &= \begin{bmatrix}0_{N\times N}\\ \sigma I_N\end{bmatrix}
    \quad\text{(independent noise on momenta only)}.
\end{align}
Since $X$ has degree $d_X = 2$ (from $\beta q_j^2$) and $h$ is constant
($d_h = 0$), we have $\bar{d} = \max(1, -2) = 1$.

\textbf{Diffusion tensor.}
$H = \frac{\sigma^2}{2}\,
\mathrm{diag}(0, \ldots, 0, 1, \ldots, 1)$
(nonzero only on the momentum block).

\textbf{Generator.}
Applied to $\phi_\alpha = \prod_{i=1}^{N} q_i^{a_i} p_i^{b_i}$,
with contributions from each coordinate.
For the kinematic dimensions ($dq_i = p_i\, dt$):
\begin{equation}
    a_i \cdot q_i^{a_i - 1} p_i^{b_i + 1} \prod_{j \ne i} q_j^{a_j} p_j^{b_j}
    \quad\Longrightarrow\quad
    a_i\, m_{\alpha - e_{q_i} + e_{p_i}}.
\end{equation}
For the momentum dimensions ($dp_j$), the drift terms give:
\begin{align}
    &{-}\gamma\, b_j\, m_\alpha
    \;-\; \alpha_{\mathrm{stiff}}\, b_j\, m_{\alpha - e_{p_j} + e_{q_j}}
    \;-\; \beta\, b_j\, m_{\alpha - e_{p_j} + 2e_{q_j}}
    \nonumber\\
    &+\; \kappa\, b_j\bigl(
    m_{\alpha - e_{p_j} + e_{q_{j-1}}}
    + m_{\alpha - e_{p_j} + e_{q_{j+1}}}
    - 2\, m_{\alpha - e_{p_j} + e_{q_j}}\bigr),
\end{align}
and the diffusion gives
$\frac{\sigma^2}{2} b_j(b_j{-}1)\, m_{\alpha - 2e_{p_j}}$.
The term $m_{\alpha - e_{p_j} + 2e_{q_j}}$ has degree
$|\alpha| + 1$ and is unclosed at $|\alpha| = K$.

\textbf{Centered coordinates.}
In strictly mean-centered coordinates $z=x-\mu(t)$ with
$\dot{\mu}=\E[X(x)]$,
the momentum equation for oscillator $j$ becomes
\begin{align}
    dz_{N+j} &= \bigl[-\gamma z_{N+j}
                - (\alpha + 2\beta\mu_j)\, z_j
                - \beta (z_j^2 - m^z_{2e_{q_j}})
    \nonumber\\
    &\quad\quad
                + \kappa(z_{j-1} + z_{j+1} - 2z_j)\bigr]\, dt
                + \sigma\, dW_j.
\end{align}
The centered linear coefficient for $z_j$ is
$-(\alpha + 2\beta\mu_j)$, not $-\alpha$:
this linearization around $\mu$ is essential for numerical stability.

\subsection{Double-well Langevin}
\label{app:double-well-generator}

\textbf{SDE.}
State $x = (x_1, x_2)$.
The gradient Langevin dynamics is
$dX = -\nabla V(X)\,dt + \sigma\,dW$
with the separable quartic potential
\begin{equation}\label{eq:double-well-potential}
    V(x) = \frac{1}{4}x_1^4 + \frac{1}{4}x_2^4
         - \frac{1}{2}x_1^2 - \frac{1}{2}x_2^2.
\end{equation}

\textbf{Drift and diffusion classification.}
The drift is
$X(x) = -\nabla V = (x_1 - x_1^3,\; x_2 - x_2^3)^\top$,
$h = \sigma I_2$.
Since $X$ has degree $d_X = 3$ (from $x_i^3$) and $h$ is constant
($d_h = 0$), we have $\bar{d} = \max(d_X - 1,\; 2d_h - 2) = 2$.
This is the highest excess degree among all experimental systems
and requires \emph{two} Stein closure layers.

\textbf{Generator.}
Applied to $\phi_\alpha = x_1^a x_2^b$:
\begin{align}\label{eq:double-well-generator}
    \mathcal{A}\phi_\alpha
    &= a(x_1 - x_1^3)\, x_1^{a-1} x_2^b
     + b(x_2 - x_2^3)\, x_1^a x_2^{b-1}
    \nonumber\\
    &\quad + \frac{\sigma^2}{2}\bigl[
      a(a{-}1)\, x_1^{a-2} x_2^b
    + b(b{-}1)\, x_1^a x_2^{b-2}\bigr].
\end{align}
Taking expectations:
\begin{align}\label{eq:double-well-moment-ode}
    \dot{m}_\alpha
    &= a\, m_\alpha - a\, m_{\alpha + 2e_1}
     + b\, m_\alpha - b\, m_{\alpha + 2e_2}
    \nonumber\\
    &\quad + \frac{\sigma^2}{2}\bigl[
      a(a{-}1)\, m_{\alpha - 2e_1}
    + b(b{-}1)\, m_{\alpha - 2e_2}\bigr].
\end{align}
The terms $m_{\alpha + 2e_1}$ and $m_{\alpha + 2e_2}$ have degree
$|\alpha| + 2$, confirming $\bar{d} = 2$.
At $|\alpha| = K = 2r{-}2$, the first Stein closure layer resolves
degree-$(K{+}1)$ moments, and a second layer resolves degree-$(K{+}2)$.

Because the potential is separable ($V = V_1(x_1) + V_2(x_2)$),
the drift couples coordinates only through the shared moment vector.
The independent noise ($H = \frac{\sigma^2}{2}I_2$) preserves
this separability: if the initial density is a product
$p_0(x) = p_1(x_1)\,p_2(x_2)$, then $p(t,x)$ remains a product
for all $t$.
In the experiments (Section~\ref{sec:experiments}),
we use a near-product initial condition and propagate
to $T = 3$~s, by which time the density has split into 4 modes
at $(\pm 1, \pm 1)$.

\subsection{3D tracer advection}
\label{app:generator-fluid}

\textbf{SDE.}
State $X = (x_1, x_2, x_3) \in \R^3$, velocity field
\begin{equation}
    u(X) = \begin{pmatrix}
        \varepsilon x_1 + \omega x_2 \\
        -\omega x_1 + \varepsilon x_2 \\
        -2\varepsilon x_3 + \alpha(x_1^2 + x_2^2)
    \end{pmatrix}, \qquad
    dX = u(X)\,dt + \sqrt{2\kappa}\,dW.
\end{equation}

\textbf{Drift classification.}
$u_1$ and $u_2$ are linear in $X$, while $u_3$ is quadratic via
$\alpha(x_1^2 + x_2^2)$.
Drift degree $d_X = 2$, noise additive ($d_h = 0$),
$\bar{d} = d_X - 1 = 1$.
One Stein closure layer.
Incompressibility: $\nabla \cdot u = \varepsilon + \varepsilon - 2\varepsilon = 0$.

\textbf{Generator.}
For $\phi(x) = x_1^a x_2^b x_3^c$:
\begin{align}
    \mathcal{A}\phi &= a(\varepsilon x_1 + \omega x_2)x_1^{a-1}x_2^b x_3^c
    + b(-\omega x_1 + \varepsilon x_2)x_1^a x_2^{b-1}x_3^c
    \nonumber\\
    &\quad + c\bigl[-2\varepsilon x_3 + \alpha(x_1^2{+}x_2^2)\bigr]
    x_1^a x_2^b x_3^{c-1}
    \nonumber\\
    &\quad + \kappa\bigl[a(a{-}1)x_1^{a-2}x_2^b x_3^c
    + b(b{-}1)x_1^a x_2^{b-2}x_3^c
    + c(c{-}1)x_1^a x_2^b x_3^{c-2}\bigr].
\end{align}

\textbf{Centered moment ODE.}
In $z = x - \mu(t)$ with $\dot{\mu} = u(\mu)$:
\begin{align}
    f_1^{\text{cen}} &= \varepsilon z_1 + \omega z_2, \nonumber\\
    f_2^{\text{cen}} &= -\omega z_1 + \varepsilon z_2, \nonumber\\
    f_3^{\text{cen}} &= -2\varepsilon z_3 + \alpha(z_1^2 + z_2^2
    + 2\mu_1 z_1 + 2\mu_2 z_2).
\end{align}
The $z_1, z_2$ equations are linear and close without Stein closure.
Only the $z_3$ equation has quadratic terms ($\alpha z_1^2$, $\alpha z_2^2$)
that couple to the closure at degree $K$.

\subsection{SE(2) kinematics}

\textbf{SDE.}
State $x = (c, s, p_x, p_y)$ with $c = \cos\theta$, $s = \sin\theta$,
$c^2 + s^2 = 1$.
The SDE in embedded coordinates is
\begin{align}
    dc  &= \Bigl(-\omega s - \frac{\sigma^2}{2} c\Bigr) dt
         - \sigma s\, dW, \nonumber\\
    ds  &= \Bigl(\omega c - \frac{\sigma^2}{2} s\Bigr) dt
         + \sigma c\, dW, \nonumber\\
    dp_x &= v c\, dt, \quad dp_y = v s\, dt,
\end{align}
where $\omega$ is the angular velocity, $v$ is the forward speed,
and $\sigma$ is the heading noise intensity.

\textbf{Drift and diffusion classification.}
Both are linear in $x$:
$d_X = 1$, $d_h = 1$, $\bar{d} = 0$.

\textbf{Diffusion tensor.}
\begin{equation}
    H = \frac{\sigma^2}{2}
    \begin{pmatrix}
        s^2 & -cs & 0 & 0 \\
        -cs & c^2 & 0 & 0 \\
        0   & 0   & 0 & 0 \\
        0   & 0   & 0 & 0
    \end{pmatrix}.
\end{equation}

\textbf{Generator.}
For $\phi_\alpha = c^{a_1} s^{a_2} p_x^{a_3} p_y^{a_4}$,
the generator $\mathcal{L}\phi_\alpha = \sum_i X_i \partial_i \phi_\alpha
+ \sum_{ij} H_{ij} \partial_i \partial_j \phi_\alpha$ gives the moment ODE
$\dot{m}_\alpha = \E[\mathcal{L}\phi_\alpha]$:
\begin{align}
    \dot{m}_\alpha
    &= \omega\bigl(a_2\, m_{\alpha+e_c-e_s}
    - a_1\, m_{\alpha-e_c+e_s}\bigr)
    \nonumber\\
    &\quad - \frac{\sigma^2}{2}(a_1{+}a_2)\, m_\alpha
    + v\bigl(a_3\, m_{\alpha+e_c-e_{p_x}}
    + a_4\, m_{\alpha+e_s-e_{p_y}}\bigr)
    \nonumber\\
    &\quad + \frac{\sigma^2}{2}\bigl[
    a_1(a_1{-}1)\, m_{\alpha-2e_c+2e_s}
    + a_2(a_2{-}1)\, m_{\alpha+2e_c-2e_s}\bigr]
    \nonumber\\
    &\quad - \sigma^2 a_1 a_2\, m_\alpha.
\end{align}
The first line is the rotation coupling (drift of $c$ and $s$),
the second is the It\^o correction plus the translation coupling,
and the last two lines are the state-dependent diffusion.
Every moment on the right-hand side has degree $|\alpha|$,
confirming $\bar{d} = 0$.

On the constraint manifold $c^2 + s^2 = 1$, the diffusion terms
simplify via $m_{\alpha-2e_c+2e_s} = m_{\alpha-2e_c} - m_\alpha$
and $m_{\alpha+2e_c-2e_s} = m_{\alpha-2e_s} - m_\alpha$
(replacing $s^2 \to 1 - c^2$ and $c^2 \to 1 - s^2$ respectively),
reducing to lower-degree moments plus diagonal corrections to $m_\alpha$.
The resulting ODE is linear:
$\dot{m}(t) = L\,m(t)$,
so $m(T) = e^{LT}m(0)$.

\textbf{Fourier--Legendre basis.}
A natural choice of orthogonal basis for $S^1 \times \R^2$ replaces
the embedded monomials with:
Fourier modes $\cos(m\theta)$, $\sin(m\theta)$ for the heading
($m = 0, \ldots, K_\theta$),
and centered Legendre polynomials $\widetilde{P}_{j}(z_{p_x})\,\widetilde{P}_{l}(z_{p_y})$
for the position.
For the propagation formula below, write the raw Fourier--position moments
$\mathrm{mc}_{m,j,l} := \E[\cos(m\theta)\, p_x^j\, p_y^l]$
and $\mathrm{ms}_{m,j,l}$ similarly with $\sin$.
The SE(2) kinematics give:
\begin{align}
    \frac{d}{dt}\mathrm{mc}_{m,j,l}
    &= -m\omega\;\mathrm{ms}_{m,j,l}
    - \frac{m^2\sigma^2}{2}\;\mathrm{mc}_{m,j,l}
    \nonumber\\
    &\quad + \frac{vj}{2}\bigl(\mathrm{mc}_{m-1,j-1,l} + \mathrm{mc}_{m+1,j-1,l}\bigr)
    \nonumber\\
    &\quad + \frac{vl}{2}\bigl(\mathrm{ms}_{m+1,j,l-1} - \mathrm{ms}_{m-1,j,l-1}\bigr),
    \label{eq:fourier-ode}
\end{align}
using $\cos\theta\cos(m\theta) = \frac{1}{2}(\cos((m{-}1)\theta) + \cos((m{+}1)\theta))$.
The coupling matrix $L$ is \emph{tridiagonal in $m$}:
the position drift $v\cos\theta$ couples Fourier mode $m$ to $m \pm 1$.
With $K_\theta = 6$ and position degree 4, the matrix is
${<}1\%$ nonzero, and the solution is $m(T) = e^{LT} m(0)$.
In the Fourier--Legendre implementation, the same operator is conjugated
into the Legendre position basis using the change-of-basis formulas in
Appendix~\ref{app:basis-change}. Position derivatives are then evaluated
with the Legendre recurrence below.

\subsubsection{Stein identity on \texorpdfstring{$S^1 \times \R^2$}{S1 x R2}}
\label{app:se2-stein}

For the measurement update on SE(2), we need to recover posterior
moments from posterior score parameters in the Fourier--Legendre basis.
The Stein identity on $S^1 \times \R^2$ takes the form:
for a test function $f(\theta, p_x, p_y)$ and a density
$p(\theta, p_x, p_y) \propto \exp(-E(\theta, p_x, p_y))$,
\begin{equation}\label{eq:stein-s1r2}
    \E\Bigl[\frac{\partial^2 f}{\partial \theta^2}
    + \frac{\partial^2 f}{\partial p_x^2}
    + \frac{\partial^2 f}{\partial p_y^2}
    - \frac{\partial E}{\partial \theta}\frac{\partial f}{\partial \theta}
    - \frac{\partial E}{\partial p_x}\frac{\partial f}{\partial p_x}
    - \frac{\partial E}{\partial p_y}\frac{\partial f}{\partial p_y}\Bigr] = 0,
\end{equation}
where the angular part uses the standard Laplace--Beltrami operator
$\Delta_{S^1} = \partial^2/\partial\theta^2$ on the circle.

\textbf{Fourier--Legendre specialization.}
Let the energy be
$E(\theta, p_x, p_y) = \sum_\alpha \lambda_\alpha\, \psi_\alpha(\theta, p_x, p_y)$,
where $\psi_\alpha$ are Fourier--Legendre basis functions
$\cos(m\theta)\,P_j(z_{p_x})\,P_l(z_{p_y})$ or
$\sin(m\theta)\,P_j(z_{p_x})\,P_l(z_{p_y})$.
The derivatives of basis functions are:
\begin{alignat}{2}
    \frac{\partial}{\partial\theta}\cos(m\theta)
    &= -m\sin(m\theta), &\quad
    \frac{\partial}{\partial\theta}\sin(m\theta)
    &= m\cos(m\theta), \nonumber\\
    \frac{\partial^2}{\partial\theta^2}\cos(m\theta)
    &= -m^2\cos(m\theta), &\quad
    \frac{\partial^2}{\partial\theta^2}\sin(m\theta)
    &= -m^2\sin(m\theta), \nonumber\\
    \frac{\partial}{\partial p_x}P_j(z_{p_x})
    &= \frac{1}{\sigma_{p_x}}P'_j(z_{p_x}), &\quad
    \frac{\partial^2}{\partial p_x^2}P_j(z_{p_x})
    &= \frac{1}{\sigma_{p_x}^2}P''_j(z_{p_x}),
    \label{eq:fl-derivs}
\end{alignat}
where $z_{p_x} = (p_x - \mu_{p_x})/\sigma_{p_x}$ and
$P'_j$, $P''_j$ are Legendre derivative coefficients
expressible in the Legendre basis via the recurrence
$P'_n(x) = \sum_{k < n,\, n{-}k\text{ odd}} (2k{+}1)\,P_k(x)$.

Substituting a Fourier--Legendre test function $\psi_\beta$ into
\eqref{eq:stein-s1r2} and using the product linearization
$\cos(m\theta)\cos(m'\theta) = \frac{1}{2}[\cos((m{-}m')\theta) + \cos((m{+}m')\theta)]$
and $P_j P_{j'} = \sum_k g_{jj'k}\,P_k$ (Clebsch--Gordan),
each Stein equation becomes a linear constraint on the
Fourier--position moments $\mathrm{mc}_{m,j,l}$, $\mathrm{ms}_{m,j,l}$.
The full system has the structure $\Lambda_1(\lambda)\,\mu = c_1(\lambda)$
where $\mu$ is the vector of unknown posterior moments,
exactly as in the Euclidean case (Section~\ref{sec:recovery}).

\textbf{Key advantage over the embedded basis.}
the Fourier basis diagonalizes the $S^1$ Laplacian
($\Delta_{S^1}\cos(m\theta) = -m^2\cos(m\theta)$),
so no constraint $c^2 + s^2 = 1$ appears in the system.
The Stein matrix $\Lambda_1$ is well-conditioned in this basis.

\textbf{Observation likelihood.}
For a landmark $L \in \R^2$ observed as
$y = R(\theta)^\top(L + v - p)$ with non-Gaussian noise $v$
(e.g., the 4-modal distribution of~\cite{teng2025}),
the log-likelihood $\log p(y \mid \theta, p)$
is not polynomial in $(\theta, p)$, but it can be projected
onto the Fourier--Legendre basis by regression:
evaluate $\log p(y \mid \theta_i, p_i)$ at random samples
$(\theta_i, p_i)$ and solve for $\lambda_{\mathrm{lik}}$
via least squares.
The conjugate update
$\lambda^+ = \lambda^- + \lambda_{\mathrm{lik}}$
then proceeds as in Section~\ref{sec:update},
followed by Stein moment recovery and iterative
consistency refinement (Appendix~\ref{app:smom}).

\subsection{SO(3) quaternion kinematics}
\label{app:so3-generator}

State $q = (q_0, q_1, q_2, q_3) \in S^3 \subset \R^4$ with $|q|^2 = 1$.
The unit quaternion represents a rotation $R \in \mathrm{SO}(3)$
via the double cover $S^3 \to \mathrm{SO}(3)$
(both $q$ and $-q$ map to the same rotation $R$).

\textbf{SDE.}
The rotation kinematics SDE in It\^o form~\cite{choukroun2006} is
\begin{equation}\label{eq:quat-sde}
    dq = \underbrace{\Bigl(\frac{1}{2}Q(q)\omega_0
    - \frac{3\sigma^2}{8}q\Bigr)}_{\displaystyle X(q)}\,dt
    + \underbrace{\frac{\sigma}{2}Q(q)}_{\displaystyle h(q)}\,dW,
\end{equation}
where $\omega_0 \in \R^3$ is the nominal angular velocity,
$\sigma > 0$ is the noise intensity,
$W \in \R^3$ is a standard Wiener process, and
\begin{equation}\label{eq:Q-mat}
    Q(q) = \begin{pmatrix}
        -q_1 & -q_2 & -q_3 \\
         q_0 & -q_3 &  q_2 \\
         q_3 &  q_0 & -q_1 \\
        -q_2 &  q_1 &  q_0
    \end{pmatrix} \in \R^{4 \times 3}
\end{equation}
is the quaternion multiplication matrix satisfying
$Q(q)^\top Q(q) = |q|^2 I_3$ and $Q(q)Q(q)^\top = |q|^2 I_4 - qq^\top$.
The It\^o correction $-\frac{3\sigma^2}{8}q$ ensures
$\frac{d}{dt}\E[|q|^2] = 0$, preserving the unit norm constraint.

\textbf{Drift and diffusion classification.}
Both $X(q)$ and $h(q)$ are \emph{linear} in $q$, giving
$d_X = 1$ and $d_h = 1$.  The excess degree is
$\bar{d} = \max(d_X - 1,\; 2d_h - 2) = \max(0, 0) = 0$:
the moment hierarchy closes exactly.

\textbf{Diffusion tensor.}
\begin{equation}
    H(q) = \frac{1}{2}h(q)\,h(q)^\top
    = \frac{\sigma^2}{8}\bigl(|q|^2 I_4 - qq^\top\bigr).
\end{equation}
On $S^3$ ($|q|^2 = 1$), $H = \frac{\sigma^2}{8}(I_4 - qq^\top)$
is the projection onto the tangent space of $S^3$ at $q$,
scaled by $\sigma^2/8$.

\textbf{Wigner D-matrix basis.}
The monomial basis in $(q_0, \ldots, q_3)$
fails on $S^3$ ($\kappa(A) \sim 10^{14}$) because the constraint
$|q|^2 = 1$ introduces $\binom{4+r}{4} - (r+1)^2$ algebraic dependencies
among the monomials at each degree.
The Wigner D-matrices $\{D^l_{mm'}\}_{l=0}^{l_{\max}}$,
$m, m' \in \{-l, \ldots, l\}$,
form a complete orthonormal basis for $L^2(\mathrm{SO}(3))$
by the Peter--Weyl theorem~\cite{hall2015,chirikjian2012}.
At $l_{\max} = 4$: $M = \sum_{l=0}^4 (2l+1)^2 = 165$ real basis functions.
Each $D^l_{mm'}$ is a polynomial of degree $2l$ in the quaternion
components. Functions on SO(3) correspond to \emph{even} polynomials
in $q$ (invariant under $q \mapsto -q$).

\textbf{Generator in the Wigner basis.}
Define the Wigner moments
$\hat{P}^l_{mm'}(t) := \E[D^l_{mm'}(R_t)]$.
The infinitesimal generator of the SDE~\eqref{eq:quat-sde},
when applied to $D^l_{mm'}$, yields:
\begin{itemize}
    \item \textbf{Deterministic rotation.}
    $\mathcal{A}_{\mathrm{rot}} D^l_{mm'} = -i[\omega_0 \cdot \mathbf{J}^{(l)}]_{mn}\, D^l_{nm'}$,
    where $\mathbf{J}^{(l)} = (\mathbf{J}^{(l)}_x, \mathbf{J}^{(l)}_y, \mathbf{J}^{(l)}_z)$
    are the $(2l{+}1) \times (2l{+}1)$ angular momentum matrices in the spin-$l$ representation.
    \item \textbf{Isotropic diffusion.}
    $\mathcal{A}_{\mathrm{diff}} D^l_{mm'} = -\frac{\sigma^2}{2}l(l{+}1)\, D^l_{mm'}$,
    using the Laplace--Beltrami eigenvalue $\Delta_{\mathrm{SO}(3)} D^l_{mm'} = -l(l{+}1) D^l_{mm'}$.
\end{itemize}
The moment ODE is therefore \emph{block-diagonal in $l$}:
\begin{equation}\label{eq:wigner-ode}
    \frac{d}{dt}\hat{P}^l
    = -\Bigl(\frac{\sigma^2}{2}l(l{+}1)\, I_{(2l+1)^2}
    + i\, \omega_0 \cdot \mathbf{J}^{(l)}\Bigr)\hat{P}^l,
\end{equation}
where $\hat{P}^l \in \mathbb{C}^{(2l+1) \times (2l+1)}$
and the angular momentum acts by left multiplication
on the row index $m$.
The solution is $\hat{P}^l(T) = e^{L_l T}\hat{P}^l(0)$,
with each block $L_l$ of size $(2l{+}1)^2 \times (2l{+}1)^2$.
No coupling between different $l$-blocks occurs.

\textbf{Riemannian score matching.}
The polynomial MED on SO(3) takes the form $p(R;\lambda) \propto \exp\!\bigl(-\sum_{l,m,m'} \lambda^l_{mm'} D^l_{mm'}(R)\bigr)$.
Using the Hyv\"arinen objective with the Laplace--Beltrami operator
on SO(3)~\cite{chirikjian2012}, we assemble:
\begin{align}
    b_k &= -l_k(l_k{+}1)\, \hat{p}^{l_k}_{m_k m_k'}
    \quad\text{(diagonal in the Wigner basis)}, \\
    A_{kj} &= \E[\nabla_{\mathrm{SO}(3)} \phi_k
    \cdot \nabla_{\mathrm{SO}(3)} \phi_j].
\end{align}
The $b$-vector requires only the propagated Wigner moments because
$\Delta_{\mathrm{SO}(3)}D^l_{mm'}=-l(l+1)D^l_{mm'}$.
The $A$-matrix is computed from products of differentiated Wigner functions.
Equivalently, the left-invariant vector fields act through angular-momentum
matrices, and products such as
$D^{l_1}_{m_1 m_1'} D^{l_2}_{m_2 m_2'}$ decompose via
Clebsch--Gordan coefficients~\cite{hall2015} into
$D^L_{MM'}$ with $|l_1{-}l_2| \le L \le l_1{+}l_2$,
giving $A$ a banded structure in $l$ (bandwidth $2l_{\max}$).
We do not use the simplified integration-by-parts identity
$\E_p[\nabla\phi_k\cdot\nabla\phi_j]=-\E_p[\phi_j\Delta\phi_k]$,
which would only hold without the extra $\nabla\log p$ term under the
uniform Haar measure.

\subsection{SE(3) rigid body kinematics}
\label{app:se3-generator}

\textbf{SDE.}
State $x = (q, p) \in S^3 \times \R^3$ with $q = (q_0, q_1, q_2, q_3)$
a unit quaternion ($|q|^2 = 1$) and $p = (p_1, p_2, p_3)$ the world-frame
position.
The manifold dimension is $n = 6$ (embedded in $\R^7$ via the quaternion
constraint $|q|^2 = 1$).
The It\^o SDE~\cite{choukroun2006} for the rotation is the same as the
SO(3) system~\eqref{eq:quat-sde}:
\begin{equation}\label{eq:se3-quat}
    dq = \Bigl(\frac{1}{2}Q(q)\omega_0
    - \frac{3\sigma_\omega^2}{8}q\Bigr)\,dt
    + \frac{\sigma_\omega}{2}Q(q)\,dW_\omega,
\end{equation}
where $\omega_0 \in \R^3$ is the nominal angular velocity
and $\sigma_\omega$ is the rotation noise intensity.
The position evolves by the body-frame velocity $v_0 \in \R^3$
rotated into the world frame:
\begin{equation}\label{eq:se3-pos}
    dp = R(q)\, v_0\, dt + \sigma_v\, dW_v,
\end{equation}
where $R(q) \in \mathrm{SO}(3)$ is the rotation matrix corresponding to $q$
and $\sigma_v\, dW_v$ is additive translational noise.

\textbf{Drift and diffusion classification.}
The rotation drift is linear in $q$ ($d_X = 1$ for the $q$-block).
Each entry of $R(q)$ is \emph{quadratic} in $(q_0, q_1, q_2, q_3)$,
so the position drift $R(q)v_0$ has $d_X = 2$.
The rotation diffusion $Q(q)$ is linear ($d_h = 1$) and
the position diffusion is constant ($d_h = 0$).
Overall:
$\bar{d} = \max(d_X - 1,\; 2d_h - 2) = \max(1, 0) = 1$.
One Stein closure layer is needed.

\textbf{Generator.}
For $\phi_\alpha = q_0^{a_0} q_1^{a_1} q_2^{a_2} q_3^{a_3}
p_1^{b_1} p_2^{b_2} p_3^{b_3}$
with $|\alpha| = \sum a_i + \sum b_j$, the generator
$\mathcal{A}\phi_\alpha = \nabla\phi \cdot X + \mathrm{Tr}(H\nabla^2\phi)$
has three groups of terms:

\emph{(i) Rotation block} (same as SO(3)):
the drift $\frac{1}{2}Q(q)\omega_0 - \frac{3\sigma_\omega^2}{8}q$
and diffusion $H_q = \frac{\sigma_\omega^2}{8}(I_4 - qq^\top)$
contribute terms involving only $q$-moments.
These preserve the degree of $q^a$ (i.e., $\bar{d}_q = 0$)
by the same mechanism as in Appendix~\ref{app:so3-generator}.

\emph{(ii) Position drift} ($R(q)v_0$):
\begin{equation}
    \sum_{j=1}^3 b_j \Bigl(\sum_{k=1}^3 R_{jk}(q)\, v_{0,k}\Bigr)
    q^a\, p^{b - e_j}.
\end{equation}
Since $R_{jk}(q)$ is degree 2 in $q$, this produces moments of degree
$|a| + 2 + |b| - 1 = |\alpha| + 1$.
These are the \textbf{unclosed terms} resolved by Stein closure.

\emph{(iii) Position diffusion} ($\sigma_v^2/2 \cdot \nabla_p^2\phi$):
$\frac{\sigma_v^2}{2} \sum_j b_j(b_j{-}1)\, m_{\alpha - 2e_{p_j}}$,
which has degree $|\alpha| - 2$ (closed).

\textbf{Rotation matrix entries.}
On $|q|^2 = 1$:
\begin{align}
    R_{11} &= 2(q_0^2 + q_1^2) - 1, &
    R_{12} &= 2(q_1 q_2 - q_0 q_3), \nonumber\\
    R_{21} &= 2(q_1 q_2 + q_0 q_3), &
    R_{22} &= 2(q_0^2 + q_2^2) - 1, \nonumber\\
    R_{31} &= 2(q_1 q_3 - q_0 q_2), &
    R_{33} &= 2(q_0^2 + q_3^2) - 1,
\end{align}
with $R_{13} = 2(q_1 q_3 + q_0 q_2)$,
$R_{23} = 2(q_2 q_3 - q_0 q_1)$,
$R_{32} = 2(q_2 q_3 + q_0 q_1)$.
Each $R_{jk}$ contributes a degree-2 monomial in $q$ to the
position drift, confirming $\bar{d} = 1$.

\textbf{Wigner--Legendre product basis.}
A natural choice of orthogonal basis for $\mathrm{SO}(3) \times \R^3$ is
\begin{equation}\label{eq:se3-wl-basis}
    \phi^{l,\alpha}_{mm'}(R, p)
    = D^l_{mm'}(R)\;\widetilde{P}_\alpha(p - \bar{p}),
\end{equation}
where $D^l_{mm'}$ are Wigner D-matrices ($l \le l_{\max}$)
and $\widetilde{P}_\alpha$ are centered Legendre polynomials
on $\R^3$ ($|\alpha| \le r_p$).
At $l_{\max} = 2$, $r_p = 4$: $M = 35 \times 35 = 1225$.
Define the product moments
$\hat{m}^l_{mm',\alpha}(t) := \E[D^l_{mm'}(R_t)\,\widetilde{P}_\alpha(p_t {-} \bar{p})]$.

The generator $\mathcal{A}\phi^{l,\alpha}_{mm'}$ decomposes into three terms:

\emph{(i) Rotation block:}
$\mathcal{A}_{\mathrm{rot}} \phi^{l,\alpha}_{mm'}
= (-\frac{\sigma_\omega^2}{2}l(l{+}1)\,\delta_{mn}
- i[\omega_0 {\cdot} \mathbf{J}^{(l)}]_{mn})\,
D^l_{nm'}\,\widetilde{P}_\alpha$.
Maps $(l, \alpha)$ to $(l, \alpha)$.  Always closed.

\emph{(ii) Position diffusion:}
$\frac{\sigma_v^2}{2}\Delta_p \phi^{l,\alpha}_{mm'}
= \frac{\sigma_v^2}{2}D^l_{mm'}\sum_j \alpha_j(\alpha_j{-}1)\widetilde{P}_{\alpha-2e_j}$.
Maps $(l, |\alpha|)$ to $(l, |\alpha|{-}2)$.  Always closed.

\emph{(iii) Position drift} (rotation--translation coupling):
$(Rv_0) \cdot \nabla_p\,\phi^{l,\alpha}_{mm'}
= \sum_j \alpha_j (\sum_k R_{jk} v_{0,k})\,
D^l_{mm'}\,\widetilde{P}_{\alpha-e_j}$.
Each $R_{jk}$ is a real combination of $D^1$ functions, so
$R_{jk} D^l_{mm'}$ decomposes via the Clebsch--Gordan series:
\begin{equation}\label{eq:se3-cg}
    D^1_{m_1 m_1'} D^l_{mm'}
    = \sum_{L=|l-1|}^{l+1}
    \langle 1\,m_1; l\,m | L\,M\rangle\,
    \langle 1\,m_1'; l\,m' | L\,M'\rangle\,
    D^L_{MM'}.
\end{equation}
This maps $(l, |\alpha|) \to (l{-}1, |\alpha|{-}1) \oplus (l, |\alpha|{-}1) \oplus (l{+}1, |\alpha|{-}1)$.

\textbf{Closure.} Only term (iii) is unclosed, producing Wigner level
$l_{\max}{+}1$ when $l = l_{\max}$.
One Stein closure layer resolves these, with the CG coefficients
playing the role of the monomial product rule.
The moment ODE matrix $L$ is \emph{tridiagonal in $l$}
(each Wigner level couples only to $l \pm 1$),
the direct SO(3) analog of the tridiagonal-in-$m$ Fourier
coupling~\eqref{eq:fourier-ode} for SE(2).

\section{Concrete matrix forms}
\label{app:matrices}

\subsection{Score matching matrix \texorpdfstring{$A$}{A} and vector \texorpdfstring{$b$}{b}}

For the polynomial MED $p(x;\lambda) \propto \exp(-\lambda\cdot\phi(x))$
with monomial basis $\phi_\alpha(x) = x^\alpha$, $|\alpha| \le r$,
the score matching normal equations $A\lambda = b$ have entries:
\begin{align}
    A_{\alpha\beta} &=
    \sum_{i=1}^n \alpha_i \beta_i\; m_{\alpha+\beta-2e_i}, \label{eq:A-entry}\\
    b_\alpha &= \sum_{i=1}^n \alpha_i (\alpha_i - 1)\; m_{\alpha-2e_i},
    \label{eq:b-entry}
\end{align}
where $\alpha,\beta$ index non-constant basis functions with
$1 \le |\alpha|,|\beta| \le r$, and $m_\gamma = \E[x^\gamma]$
are the propagated moments.
After eliminating the constant term ($|\alpha| = 0$),
$A$ is $(M{-}1) \times (M{-}1)$ with $M = \binom{n+r}{n}$.
All entries are moments of degree $\le 2r - 2$, which is
exactly the degree propagated by Dynkin's formula.

\subsection{Stein closure matrix \texorpdfstring{$\Lambda_1$}{Lambda-1}}

To resolve degree-$(K{+}1)$ moments, the Stein identity
\eqref{eq:stein} at $|\beta| = r$ gives the system
$\Lambda_1(\lambda)\, m^{(K+1)} = c(\lambda, m^{(\le K)})$ where:
\begin{equation}
    [\Lambda_1]_{(\beta,i),\gamma} =
    \lambda_{\gamma - \beta + e_i}\; (\gamma - \beta + e_i)_i
    \quad\text{if } |\gamma - \beta + e_i| \le r,
\end{equation}
with rows indexed by $(\beta, i)$ with $|\beta| = r$, $\beta_i \ge 1$,
and columns indexed by degree-$(K{+}1)$ multi-indices $\gamma$.
For $n = 2$, the first-layer system is $2r \times 2r$ (square).
For general $n$ and $r \ge 3$, the augmented system using
$|\beta| \in [r, K]$ has a dimension-dependent crossover.
For $r{=}3$, it is overdetermined through $n{=}15$ and
underdetermined starting at $n{=}16$.
Appendix~\ref{app:stein-wellposed} gives the full row/column count,
and Proposition~\ref{prop:stein-closure} gives the $n{=}2$
well-posedness result. For structured generators, the same matrix
construction may be restricted to the active columns requested by the
moment ODE.

\subsection{Change of basis}
\label{app:basis-change}

\textbf{General framework.}
Let $\{\phi_\alpha\}$ (e.g., monomials) and $\{\psi_\alpha\}$
(e.g., product Legendre) be two polynomial bases of degree $\le r$,
related by a nonsingular matrix $C$:
\begin{equation}
    \psi_\alpha(x) = \sum_\beta C_{\alpha\beta}\, \phi_\beta(x),
    \qquad
    \phi_\alpha(x) = \sum_\beta C^{-1}_{\alpha\beta}\, \psi_\beta(x).
\end{equation}
For product bases, $C$ has Kronecker structure:
$C_{\alpha\beta} = \prod_{i=1}^n C^{(i)}_{\alpha_i, \beta_i}$,
where $C^{(i)}_{k,j}$ is the $j$-th monomial coefficient of
the $k$-th basis function in dimension $i$.

\textbf{Natural parameters.}
The energy function $E(x) = \sum_\alpha \lambda^\phi_\alpha \phi_\alpha(x)
= \sum_\alpha \lambda^\psi_\alpha \psi_\alpha(x)$ transforms as
\begin{equation}
    \lambda^\phi = C^\top \lambda^\psi,
    \qquad
    \lambda^\psi = C^{-\top} \lambda^\phi.
    \label{eq:lambda-change}
\end{equation}
\textbf{Moments.}
The generalized moments $\mu^\psi_\alpha = \E[\psi_\alpha(x)]$
and $\mu^\phi_\alpha = \E[\phi_\alpha(x)]$ transform as
\begin{equation}
    \mu^\psi = C\, \mu^\phi, \qquad \mu^\phi = C^{-1}\, \mu^\psi.
    \label{eq:moment-change}
\end{equation}
\textbf{Score matching matrices.}
The score $s(x) = -\sum_\alpha \lambda_\alpha \nabla\phi_\alpha(x)$
is basis-invariant. Since $\nabla\psi = C\nabla\phi$, equivalently
$\nabla\phi = C^{-1}\nabla\psi$, the normal equations
$A^\psi \lambda^\psi = b^\psi$ in the new basis have
\begin{equation}
    A^\psi = D^{-1} A^\phi D^{-\top}, \qquad b^\psi = D^{-1} b^\phi,
    \label{eq:Ab-change}
\end{equation}
where $D$ is the matrix such that
$\nabla \phi = D \nabla \psi$. For a linear change of polynomial basis,
$D=C^{-1}$. Derivative recurrences are still useful for assembling
Legendre-basis matrices directly, but they do not change the
parameter-space transformation above.

\textbf{Conditioning.}
The key property of orthogonal bases (e.g., Legendre on $[-1,1]$)
is that $A^\psi$ is better conditioned than $A^\phi$ when the
data distribution is supported near $[-1,1]^n$.
For Legendre, $\E[\psi_\alpha \psi_\beta] \approx \delta_{\alpha\beta}
\prod_i (2\alpha_i+1)^{-1}$ (orthogonality under a nearly uniform
product measure), so $A^\psi$
is approximately diagonal, giving $\kappa(A^\psi) \sim O(1)$--$O(10^2)$
compared to $\kappa(A^\phi) \sim O(10^4)$--$O(10^6)$ for monomials.
This conditioning improvement propagates to the Stein closure and
moment recovery systems, since their coefficients are
products of $\lambda$ entries and moment entries.

\textbf{Stein identity under basis change.}
The Stein identity $\E[s_i f + \partial_i f] = 0$ is
basis-independent. When expressed in basis $\psi$:
\begin{equation}
    \sum_\alpha \lambda^\psi_\alpha\,
    \E[\partial_i \psi_\alpha \cdot \psi_\beta]
    = \E[\partial_i \psi_\beta].
\end{equation}
The products $\partial_i \psi_\alpha \cdot \psi_\beta$ and
$\partial_i \psi_\beta$ are polynomials that can be re-expanded
in the $\psi$ basis using the linearization coefficients
$g^\psi_{\alpha\beta\gamma}$ satisfying
$\psi_\alpha \psi_\beta = \sum_\gamma g^\psi_{\alpha\beta\gamma}\,
\psi_\gamma$.
For Legendre, these are the Adams--Neumann coefficients,
computable via Gauss--Legendre quadrature.

\textbf{Practical workflow.}
In our implementation:
\begin{enumerate}[leftmargin=2em]
    \item Moment propagation via Dynkin operates on monomial moments
    (the generator acts naturally on monomials).
    \item Before score matching: convert monomial moments to the
    chosen basis using \eqref{eq:moment-change},
    or equivalently assemble $A$, $b$ directly via
    \eqref{eq:A-entry}--\eqref{eq:b-entry} in centered/scaled
    monomial coordinates (which approximate orthogonality).
    \item After score matching: $\lambda$ is in the working basis.
    For Stein closure during propagation, convert to monomial
    $\lambda$ via \eqref{eq:lambda-change} if needed.
    \item For measurement update: the conjugate addition
    $\lambda^+ = \lambda^- + \lambda_{\mathrm{lik}}$ is performed
    in whichever basis both are expressed in.
\end{enumerate}

\subsection{Centered coordinate transformation}
\label{app:centering}

Given raw moments $m_\alpha = \E[x^\alpha]$ and mean
$\mu_i = m_{e_i}$, the centered moments
$m^z_\alpha = \E[(x - \mu)^\alpha]$ are obtained via
the multinomial expansion:
\begin{equation}
    m^z_\alpha = \sum_{\beta \le \alpha}
    \binom{\alpha}{\beta} (-\mu)^{\alpha-\beta}\, m_\beta,
    \label{eq:centering}
\end{equation}
where $\binom{\alpha}{\beta} = \prod_i \binom{\alpha_i}{\beta_i}$
and $\mu^{\alpha-\beta} = \prod_i \mu_i^{\alpha_i - \beta_i}$.
The inverse (uncentering) replaces $(-\mu)$ with $(+\mu)$.

Working in centered coordinates is essential for conditioning:
after centering, $m^z_{e_i} = 0$ and $m^z_\alpha$ is a central
moment of order $|\alpha|$, which is $O(\sigma^{|\alpha|})$
for a distribution with standard deviation $\sigma$.
This prevents the score matching matrix $A$ from being dominated
by mean-related terms.

\textbf{Scaling.} Further conditioning improvement is obtained by
normalizing to $w = z / s$ where $s_i = c \cdot \sigma_i$
(typically $c = 3.5$ maps $\pm 3.5\sigma$ to $[-1, 1]$).
The scaled moments $m^w_\alpha = m^z_\alpha / \prod_i s_i^{\alpha_i}$
are $O(1)$ for all degrees, giving $\|\lambda\| = O(10)$--$O(100)$
instead of $O(10^4)$--$O(10^6)$ with raw monomials.

\section{Moment closure: continuous-time vs.\ discrete-time}
\label{app:closure-comparison}

\textbf{The MEM-KF system class.}
MEM-KF~\cite{teng2025} operates on discrete-time polynomial systems
$X_{k+1} = f(X_k, U_k, W_k)$
where $f$ is polynomial and $W_k$ is \emph{state-independent} noise.
Their ``extended dynamical system'' (Eq.~6 in~\cite{teng2025})
lifts the dynamics into the monomial basis
$\phi_r(X_{k+1}) = A(W_k, U_k)\,\phi_r(X_k)$,
and the moment propagation factors as
$\E[\phi_r(X_{k+1})] = \E[A(W_k, U_k)]\,\E[\phi_r(X_k)]$
\emph{only because $W_k$ is independent of $X_k$}.
This is a restricted class from the perspective of stochastic systems theory:
\begin{itemize}
    \item \textbf{No multiplicative/state-dependent noise.}
    Continuous-time SDEs generically have state-dependent diffusion
    $dX = g(X)\,dt + h(X)\,dW$ where $h(X) \ne \text{const}$.
    This arises in rigid body kinematics
    ($h(q) = \frac{\sigma}{2}Q(q)$ for SE(2), SO(3), SE(3)),
    geometric Brownian motion ($h = \sigma X$),
    and aerodynamic systems with quadratic drag under turbulence.
    State-dependent noise contributes to the excess degree
    via $\bar{d} = \max(d_X{-}1,\; 2d_h{-}2)$
    and produces an It\^o correction in the drift.
    These effects have no analog in their additive
    discrete-time framework $X_{k+1} = f(X_k) + W_k$.
    \item \textbf{No continuous-time structure.}
    The infinitesimal generator $\mathcal{A}$,
    Dynkin's formula~\cite{oksendal2003},
    and the Fokker--Planck PDE
    provide exact moment evolution equations for continuous-time systems.
    Discrete-time polynomial maps do not expose this structure,
    replacing it with a single-step algebraic relation.
    \item \textbf{Discretization error for non-group systems.}
    Systems with native Lie group structure (e.g., SE(2), SO(3))
    admit exact discrete composition maps
    $X_{k+1} = X_k \cdot U_k \cdot W_k$ that preserve the group
    and introduce no time-stepping error.
    For general continuous-time SDEs that lack this structure,
    however, converting to the discrete polynomial form
    via Euler--Maruyama introduces $O(\Delta t)$ weak error
    per step, and higher-order schemes (Milstein, stochastic
    Runge--Kutta) produce more complex polynomial structure
    in the noise coupling.
\end{itemize}

The SE(2) system in~\cite{teng2025} has the special structure
$X_{k+1} = X_k \cdot U_k \cdot W_k$ (Lie group composition),
which is bilinear in the embedded coordinates $(c, s, p_x, p_y)$.
Combined with the quotient ring constraint $c^2 + s^2 = 1$,
the moment propagation closes automatically within a fixed basis.
This closure is a consequence of the bilinear group structure,
not a general property of the method.

\textbf{Closure obstruction for nonlinear systems.}
For systems with polynomial drift of degree $d_X \ge 2$
(Duffing: $d_X = 2$, Lotka-Volterra: $d_X = 2$,
coupled oscillators: $d_X = 2$, double-well Langevin: $d_X = 3$),
the moment hierarchy does not close, under
\emph{either} the continuous generator or a discrete polynomial map.

Under Euler discretization with step $\Delta t$:
\begin{equation}
    X_{k+1} = X_k + \Delta t\, g(X_k) + \sigma\sqrt{\Delta t}\,\xi_k,
\end{equation}
the map $f(x) = x + \Delta t\, g(x)$ has degree $d_X$.
Propagating degree-$K$ moments requires:
$\E[\phi(f(X_k))]$ for $|\phi| \le K$,
which involves monomials of $X_k$ up to degree $d_X \cdot K$.
The hierarchy is unclosed: degree-$K$ moments depend on
degree-$d_X K$ moments.

Truncating at degree $K$ (setting higher moments to zero)
introduces $O(\Delta t)$ error per step that compounds.
This is the discrete analog of the continuous-time excess degree
$\bar{d} = d_X - 1 \ge 1$: both formulations face the
\emph{same} fundamental obstruction.

\textbf{Stein closure as the resolution.}
The SKF resolves this via the Stein identity
(Section~\ref{sec:stein}), which provides algebraic relations
between moments of different orders, parameterized by $\lambda$.
Unlike truncation, the Stein closure is \emph{exact} for
the polynomial exponential family: it exploits the parametric
structure of the density to express unclosed moments as
functions of closed ones.
The empirical evidence is in
Figures~\ref{fig:moment-comparison-duffing}--\ref{fig:moment-comparison-lv}
(Appendix~\ref{app:moment-accuracy}):
on both the Duffing oscillator ($\bar{d} = 1$, $d_X = 2$)
and Lotka-Volterra ($\bar{d} = 1$, $d_X = 2$),
all 14 centered moments up to degree~4 match
$5 \times 10^5$-particle MC to within $1$--$7\%$
over the full propagation interval.
This capability (handling continuous-time SDEs with
$\bar{d} \ge 1$ without moment truncation) is absent
in the discrete-time MEM-KF framework and is a core
contribution of this work.

\textbf{SE(2) landmark benchmark.}
To verify that the SKF can match MEM-KF on its
\emph{own} benchmark, we reproduce the SE(2) localization
experiment of~\cite{teng2025}: a robot traverses a unit circle
with heading noise $\sigma = 0.05$\,rad/$\sqrt{\text{s}}$,
observing four symmetric landmarks with unknown data association
(4-modal likelihood).
We use an intrinsic Fourier$\,\times\,$Legendre basis on
$S^1 \times \R^2$ with $K_\theta = 8$ angular modes and
position degree $r = 4$, avoiding the
$c^2 + s^2 = 1$ quotient ring constraint of the embedded
$(c,s,p_x,p_y)$ coordinates used in~\cite{teng2025}.
The result is shown in Figure~\ref{fig:se2-landmark}:
the SKF mean trajectory tracks the ground truth comparably
to a $5 \times 10^5$-particle bootstrap filter.

\begin{figure}[ht]
    \centering
    \includegraphics[width=0.55\linewidth]{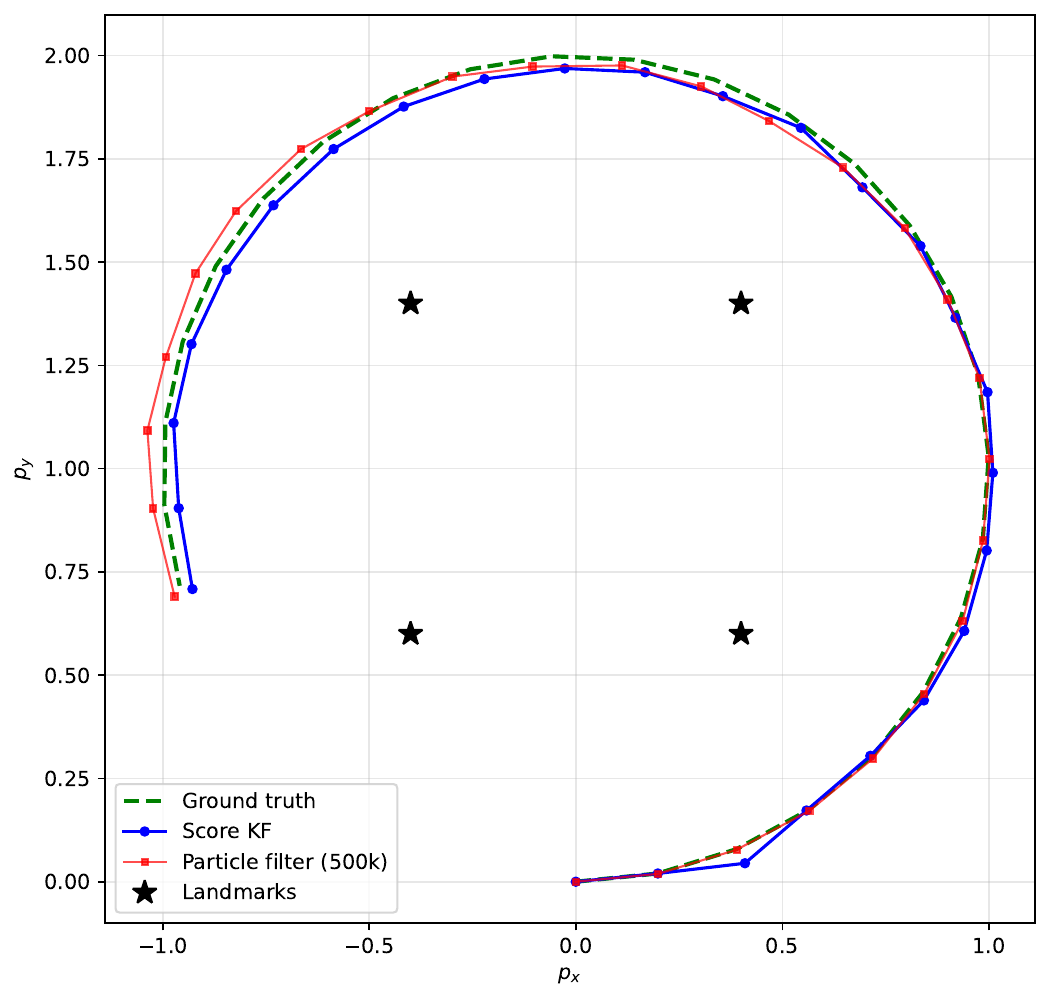}
    \caption{SE(2) landmark filter reproducing the benchmark
    of~\cite{teng2025}. The SKF (Fourier$\,\times\,$Legendre,
    $K_\theta{=}8$, $r{=}4$) tracks the circular ground truth
    comparably to the particle filter.
    Stars mark the four landmarks with unknown data association.}
    \label{fig:se2-landmark}
\end{figure}

Observe that a $5 \times 10^5$-particle bootstrap filter already tracks
this benchmark well:
the effective state dimension after the $c^2 + s^2 = 1$
constraint is $n = 3$, which is low enough for the particle
filter to operate without weight degeneracy.
We refer the reader to~\cite{teng2025} for a more detailed performance comparison on this system with
EKF, InEKF, UKF, UKF-M, and particle filters.
Their results confirm that MEM-KF tracks well here,
but so does the particle filter,
and neither the partition function reconstruction nor the
score matching machinery provides a clear advantage over
simple importance resampling at this dimension.
The scalability contrast with MEM-KF and the accuracy gap relative to the
reported bootstrap PF become apparent at higher dimensions:
as shown in Table~\ref{tab:comparison},
the SKF maintains practical 25-step filtering cost through $n = 20$
while the partition function integral in MEM-KF would require
$O(G^n)$ evaluations, which in practice hits the scalability limit at $n\le 8$ even with the GPU parallelization or compiled extensions.

\section{Additional experimental details}
\label{app:experiments}

\begin{table}[ht]
\caption{Comparison across methods.
The SKF score-matching fit costs $O(M^3)$ with
$M = \binom{n+r}{n}$. When Stein closure is needed, the closure solve
adds a factor depending on the active unclosed moment system.
MEM-KF~\cite{teng2025} cost is dominated by $O(G^n)$ quadrature for $Z(\lambda)$, exponential in $n$.
The coupled-oscillator rows report RMSE for the SKF and baselines.
The high-dimensional coupled oscillator rows ($n{=}12$--$20$) use
active Stein closure.}
\label{tab:comparison}
\centering
\small
\setlength{\tabcolsep}{4pt}
\begin{tabular}{lcccccc}
\toprule
\textbf{Experiment} & \textbf{SKF (ours)} & \textbf{EKF} & \textbf{UKF} & \textbf{EnKF} & \textbf{PF (500k)} & \textbf{MEM-KF} \\
\midrule
\multicolumn{7}{l}{\emph{Prediction (timing):}} \\
SE(2) ($n{=}4$, $\bar{d}{=}0$) & $0.6$~s & ${<}0.1$~s & -- & -- & $209$~s & est.\ ${\sim}35$~s$^\S$ \\
SO(3) ($l_{\max}{=}4$, $\bar{d}{=}0$) & ${<}0.1$~s & ${<}0.1$~s & -- & -- & $396$~s & est.\ ${\sim}100$~s$^\S$ \\
Lotka-Volterra ($n{=}2$, $\bar{d}{=}1$) & $0.7$~s & ${<}0.1$~s & -- & -- & $69$~s & n/a$^\ddagger$ \\
3D tracer ($n{=}3$, $r{=}6$, $\bar{d}{=}1$) & $8$~s & ${<}0.1$~s & -- & -- & $122$~s & n/a$^\ddagger$ \\
\midrule
\multicolumn{7}{l}{\emph{Filter (mean RMSE / 25 steps):}} \\
Duffing ($n{=}2$, $r{=}4$, $\bar{d}{=}1$) & tracks & tracks & tracks & tracks & tracks & n/a$^\ddagger$ \\
Coupled osc.\ ($n{=}4$, $r{=}3$) & $\mathbf{.0103\!\pm\!.0040}$ & $.0676$ & $.0725$ & $.0762$ & $.0727$ & n/a$^{\ddagger\dagger}$ \\
Coupled osc.\ ($n{=}6$, $r{=}3$) & $\mathbf{.0200\!\pm\!.0154}$ & $.0746$ & $.0746$ & $.0762$ & $.0754$ & n/a$^{\ddagger\dagger}$ \\
Coupled osc.\ ($n{=}8$, $r{=}3$) & $\mathbf{.0181\!\pm\!.0112}$ & $.0713$ & $.0775$ & $.0806$ & $.0747$ & n/a$^{\ddagger\dagger}$ \\
Coupled osc.\ ($n{=}10$, $r{=}3$) & $\mathbf{.0197\!\pm\!.0123}$ & $.0728$ & $.0743$ & $.0752$ & $.0744$ & n/a$^{\ddagger\dagger}$ \\
Coupled osc.\ ($n{=}12$, $r{=}3$) & $\mathbf{.0029\!\pm\!.0017}$ & $.0708$ & $.0731$ & $.0768$ & $.0741$ & n/a$^{\ddagger\dagger}$ \\
Coupled osc.\ ($n{=}14$, $r{=}3$) & $\mathbf{.0039\!\pm\!.0042}$ & $.0801$ & $.0810$ & $.0842$ & $.0818$ & n/a$^{\ddagger\dagger}$ \\
Coupled osc.\ ($n{=}16$, $r{=}3$) & $\mathbf{.0055\!\pm\!.0041}$ & $.0752$ & $.0773$ & $.0820$ & $.0781$ & n/a$^{\ddagger\dagger}$ \\
Coupled osc.\ ($n{=}18$, $r{=}3$) & $\mathbf{.0039\!\pm\!.0044}$ & $.0772$ & $.0784$ & $.0816$ & $.0791$ & n/a$^{\ddagger\dagger}$ \\
Coupled osc.\ ($n{=}20$, $r{=}3$) & $\mathbf{.0039\!\pm\!.0050}$ & $.0734$ & $.0752$ & $.0806$ & $.0761$ & n/a$^{\ddagger\dagger}$ \\
\midrule
\multicolumn{7}{l}{\emph{Filter runtime (25 steps, single CPU):}} \\
Duffing ($n{=}2$) & $0.3$~s & ${<}0.1$~s & ${<}0.1$~s & $1.8$~s & $46$~s & -- \\
Coupled osc.\ ($n{=}4$) & $1$~s & ${<}0.1$~s & ${<}0.1$~s & $1.2$~s & $110$~s & -- \\
Coupled osc.\ ($n{=}6$) & $4$~s & ${<}0.1$~s & $0.1$~s & $1.7$~s & $184$~s & -- \\
Coupled osc.\ ($n{=}8$) & $14$~s & ${<}0.1$~s & $0.1$~s & $2.9$~s & $291$~s & -- \\
Coupled osc.\ ($n{=}10$) & $41$~s & ${<}0.1$~s & $0.2$~s & $3.1$~s & $312$~s & -- \\
Coupled osc.\ ($n{=}12$) & $195$~s & $0.02$~s & $0.21$~s & $4.0$~s & $385$~s & -- \\
Coupled osc.\ ($n{=}14$) & $344$~s & $0.02$~s & $0.34$~s & $4.8$~s & $441$~s & -- \\
Coupled osc.\ ($n{=}16$) & $810$~s & $0.03$~s & $0.44$~s & $5.1$~s & $540$~s & -- \\
Coupled osc.\ ($n{=}18$) & $1182$~s & $0.03$~s & $0.47$~s & $5.8$~s & $603$~s & -- \\
Coupled osc.\ ($n{=}20$) & $2079$~s & $0.03$~s & $0.57$~s & $6.5$~s & $646$~s & -- \\
\midrule
$Z$ evaluations per step & \textbf{0} & 0 & 0 & 0 & 0 & $\ge 1$ \\
Iterative optimization & \textbf{no} & no & no & no & no & BFGS \\
Cost per step & $O(M^3)$ + closure & $O(n^3)$ & $O(n^3)$ & $O(N_e n)$ & $O(N_p n)$ & $O(G^n M)$ \\
\bottomrule
\end{tabular}

\vspace{4pt}
{\footnotesize\raggedright PF uses $5{\times}10^5$ bootstrap particles. EnKF uses 500 ensemble members.\\
$^\ddagger$MEM-KF~\cite{teng2025} propagates moments via discrete-time
polynomial maps $X_{k+1} = f(X_k, U_k, W_k)$ with state-independent noise,
which close automatically for bilinear maps.
Continuous-time SDEs with $\bar{d} \ge 1$
(e.g., Duffing, LV, coupled oscillators, etc.) require moment closure
(Appendix~\ref{app:closure-comparison}),
which MEM-KF does not address.\\
$^\dagger$\cite{teng2025} demonstrated MEM-KF on SE(2) (3D integration
after constraint reduction) using adaptive Gauss--Kronrod quadrature
(order 7, 15 Kronrod points/dim per subregion).
The partition function $Z(\lambda)$ requires $n$-dimensional
numerical integration at each BFGS iteration.
At $n{\ge}6$ the cost per $Z$ evaluation grows
as $O(G^n)$ and has not been demonstrated on a single laptop CPU.
For example, at $n{=}8$: $15^8 \approx 2.6 \times 10^9$,
and at $n{=}10$: $15^{10} \approx 5.8 \times 10^{11}$
evaluations per subregion.\\[2pt]
$^\S$\emph{MEM-KF timing estimates.}
SE(2): 3D quadrature ($c^2{+}s^2{=}1$ reduces $n{=}4$ to 3 effective dimensions)
with $M{=}45$ basis functions.
At $G{=}30$ points/dim: $G^3 {=} 27{,}000$ evaluations per $Z$ call,
${\sim}20$ BFGS iterations with function $+$ gradient
${\approx} 40 \times 27{,}000 \times 45$ ${\approx} 49$M~operations per
reconstruction. Estimated ${\sim}35$\,s in Python.
SO(3): same 3D integration but $M{=}165$ Wigner D-matrix basis
functions ($l_{\max}{=}4$: $\sum_{l=0}^{4}(2l{+}1)^2 {=} 165$).
Each $D^l_{mm'}(R)$ evaluation requires $O(l^2)$ operations
via the Wigner $d$-matrix recurrence,
${\sim}3\times$ costlier per basis function than monomial evaluation.
With more parameters and costlier evaluations:
estimated ${\sim}100$\,s.\par}
\end{table}

\subsection{Discussion of Table~\ref{tab:comparison}}

\textbf{Prediction timing.}
Across the prediction benchmarks, the SKF avoids sampling during moment
propagation and uses less wall-clock time than the 500k-particle Monte
Carlo reference in Table~\ref{tab:comparison}.
The largest timing gaps occur on manifold-valued systems.
For SO(3) with Wigner D-matrices, the propagation reduces to independent
small matrix exponentials on each $l$-block of size $(2l{+}1)^2$,
while MC must perform a quaternion multiplication for every particle at
every time step.
The $\bar{d} = 1$ systems (LV, 3D tracer) have a narrower gap because
Stein closure adds per-step linear algebra that MC does not incur.

\textbf{Filter accuracy.}
On the coupled Duffing oscillators, the SKF has lower mean RMSE than the
reported baselines (EKF, UKF, EnKF, and PF) over
10 independent seeds.
The EKF, UKF, and EnKF cluster at similar RMSE. In these runs, changing the
Gaussian propagation rule does not remove the error associated with a
non-Gaussian posterior.
The 500k-particle PF also performs comparably to the
Gaussian filters, which is consistent with particle degeneracy under partial
observation.
This PF baseline is a standard bootstrap filter using the transition prior
as the proposal. Proposal-adapted PF variants (for example EKF- or
UKF-informed proposals) are a useful stronger baseline for future study.
The SKF behavior is consistent with the polynomial exponential-family belief
retaining higher-order moment information about the nonlinear coupling.

\textbf{Filter timing and scalability.}
In the multi-seed PF comparison, the SKF filter runtime grows from
4\,s ($n{=}6$) to 41\,s ($n{=}10$), scaling polynomially.
At $n{=}10$, the SKF run uses less wall-clock time than the
500k-particle PF run on the
same CPU (41\,s vs.\ 312\,s).
The larger active-closure sweep reaches $n{=}20$ in 2079\,s for a
25-step trajectory on the same CPU.
The bootstrap PF is faster than the SKF in the largest rows
($n\ge 16$), although it remains much less accurate on this benchmark.
The factorization-caching structure (Remark~\ref{rem:factor-cache})
matters at this scale: the Stein closure matrix $\Lambda_1(\lambda)$ is
factorized once per prediction window and reused across all ODE substeps.

\textbf{Comparison with MEM-KF.}
The MEM-KF~\cite{teng2025} requires $O(G^n)$ quadrature for the
partition function $Z(\lambda)$ at each step.
At $n{=}6$ this is borderline feasible ($G^6 \sim 10^7$).
At $n{=}10$ it exceeds $10^{16}$ operations per step, and at
$n{=}20$ even a coarse $G=15$ grid exceeds $10^{23}$ evaluations.
The SKF avoids evaluating $Z$ via score matching
(Proposition~\ref{prop:sm}), replacing the exponential-cost
quadrature with a polynomial-cost linear solve.
Additionally, MEM-KF is restricted to discrete-time bilinear maps
whose moment propagation closes automatically.
Continuous-time SDEs with $\bar{d} \ge 1$ require the Stein closure
(Section~\ref{sec:stein}), which is a contribution of this work.
Once the posterior parameter $\lambda^+$ is available, another option
is to extract a MAP state by minimizing the polynomial energy
$\lambda^+\cdot\phi(x)$.
This is commonly done with a polynomial optimization problem and an SDP
relaxation in the moment-SOS
hierarchy~\cite{lasserre2001global,parrilo2003semidefinite}, as in
MEM-KF-style state extraction~\cite{teng2025}.
This POP/SDP state-extraction step is compatible with the posterior
parameter produced by the SKF, but the experiments here do not use it
because it would add a separate SDP layer.
We instead recover posterior moments, which the recursive filter already needs
for the next prediction step, and report the posterior mean rather than
a MAP point.

\textbf{Multi-seed variability for all filters.}
Table~\ref{tab:comparison} reports the SKF column as
mean$\,\pm\,$std and the baseline columns as mean only, to keep the
summary table compact. Table~\ref{tab:comparison-full} below repeats
the coupled-oscillator filter rows with mean$\,\pm\,$std for every
method. Two patterns are visible. First, the Gaussian baselines have
$\sigma/\mu \approx 10$--$30\%$ across all dimensions, i.e., they are
seed-robust at the cost of being non-Gaussian-blind. Second, the SKF
has comparatively wider $\sigma/\mu$: about $40$--$80\%$ for
$n{\le}10$, about $60\%$ at $n{=}12$, and above $100\%$ in several
larger rows because the mean RMSE is very small. Thus the absolute
standard deviation remains small, but the relative ratio does not
monotonically tighten. The SKF mean stays below every
baseline in every row and at every dimension.

\begin{table}[ht]
\caption{Coupled oscillator filter RMSE: full mean$\,\pm\,$std for
every method, averaged over 10 independent seeds at every dimension.}
\label{tab:comparison-full}
\centering
\footnotesize
\setlength{\tabcolsep}{4pt}
\begin{tabular}{lccccc}
\toprule
\textbf{System} & \textbf{SKF (ours)} & \textbf{EKF} & \textbf{UKF} & \textbf{EnKF} & \textbf{PF (500k)} \\
\midrule
Coupled osc.\ ($n{=}4$, $r{=}3$)  & $\mathbf{.0103\!\pm\!.0040}$ & $.0676\!\pm\!.0143$ & $.0725\!\pm\!.0136$ & $.0762\!\pm\!.0150$ & $.0727\!\pm\!.0136$ \\
Coupled osc.\ ($n{=}6$, $r{=}3$)  & $\mathbf{.0200\!\pm\!.0154}$ & $.0746\!\pm\!.0232$ & $.0746\!\pm\!.0255$ & $.0762\!\pm\!.0265$ & $.0754\!\pm\!.0232$ \\
Coupled osc.\ ($n{=}8$, $r{=}3$)  & $\mathbf{.0181\!\pm\!.0112}$ & $.0713\!\pm\!.0160$ & $.0775\!\pm\!.0129$ & $.0806\!\pm\!.0121$ & $.0747\!\pm\!.0155$ \\
Coupled osc.\ ($n{=}10$, $r{=}3$) & $\mathbf{.0197\!\pm\!.0123}$ & $.0728\!\pm\!.0134$ & $.0743\!\pm\!.0222$ & $.0752\!\pm\!.0221$ & $.0744\!\pm\!.0140$ \\
Coupled osc.\ ($n{=}12$, $r{=}3$) & $\mathbf{.0029\!\pm\!.0017}$ & $.0708\!\pm\!.0089$ & $.0731\!\pm\!.0086$ & $.0768\!\pm\!.0085$ & $.0741\!\pm\!.0087$ \\
Coupled osc.\ ($n{=}14$, $r{=}3$) & $\mathbf{.0039\!\pm\!.0042}$ & $.0801\!\pm\!.0084$ & $.0810\!\pm\!.0079$ & $.0842\!\pm\!.0087$ & $.0818\!\pm\!.0082$ \\
Coupled osc.\ ($n{=}16$, $r{=}3$) & $\mathbf{.0055\!\pm\!.0041}$ & $.0752\!\pm\!.0077$ & $.0773\!\pm\!.0072$ & $.0820\!\pm\!.0080$ & $.0781\!\pm\!.0072$ \\
Coupled osc.\ ($n{=}18$, $r{=}3$) & $\mathbf{.0039\!\pm\!.0044}$ & $.0772\!\pm\!.0074$ & $.0784\!\pm\!.0069$ & $.0816\!\pm\!.0073$ & $.0791\!\pm\!.0068$ \\
Coupled osc.\ ($n{=}20$, $r{=}3$) & $\mathbf{.0039\!\pm\!.0050}$ & $.0734\!\pm\!.0069$ & $.0752\!\pm\!.0063$ & $.0806\!\pm\!.0077$ & $.0761\!\pm\!.0062$ \\
\bottomrule
\end{tabular}
\end{table}

\subsection{Moment trajectory accuracy}
\label{app:moment-accuracy}

To verify that the Stein closure propagates moments correctly,
we compare all centered moments up to degree~4 against
large-sample Monte Carlo estimates over the full propagation interval.
The Duffing moment-trajectory test initializes from
$\mathcal{N}((0.5,0)^\top,\mathrm{diag}(0.04,0.04))$.
The compact Lotka--Volterra moment panel initializes at the equilibrium
$(\gamma/\delta,\alpha/\beta)$ with covariance $\mathrm{diag}(0.1,0.1)$.

\textbf{Duffing oscillator} ($n{=}2$, $r{=}4$, $\bar{d}{=}1$):
Figure~\ref{fig:moment-comparison-duffing} shows all 14 centered moments
up to degree~4 over $t \in [0, 2]$.
All moments (including variances, third-order skewness, and
fourth-order kurtosis) agree with MC to within $1\%$ relative error,
confirming that a single Stein closure layer resolves the unclosed
hierarchy with high accuracy.

\textbf{Lotka-Volterra} ($n{=}2$, $r{=}6$, $\bar{d}{=}1$):
Figure~\ref{fig:moment-comparison-lv} shows the same comparison
over $t \in [0, 1.5]$, starting from the predator-prey equilibrium.
Variances match MC within $1\%$. Third- and fourth-order moments
stay within $2$--$7\%$, reflecting the larger dynamic range
of the neutrally stable oscillatory dynamics.

\textbf{SE(2) kinematics} ($n{=}3$, $\bar{d}{=}0$):
Figure~\ref{fig:moment-comparison-se2} shows Fourier $\times$ position
moments ($\E[\cos(k\theta)\, p_x^a p_y^b]$ and
$\E[\sin(k\theta)\, p_x^a p_y^b]$) propagated via Dynkin's formula
($m(T) = e^{LT}m(0)$, no Stein closure needed).
All 16 moments match MC within $0.5\%$.

\textbf{SO(3) attitude} ($n{=}3$, $\bar{d}{=}0$):
Figure~\ref{fig:so3} shows 16 representative Wigner D-matrix moments
$\E[D^l_{mm'}(R_t)]$ for $l = 1, 2, 3, 4$,
propagated via block-diagonal matrix exponential ($m(T) = e^{L_l T}m(0)$
for each $l$-block).
All significant moments match $2 \times 10^5$-particle MC to within $0.3\%$.
The Wigner basis avoids the constraint-induced ill-conditioning
of the monomial basis ($\kappa(A)$ from $10^{14}$ to $O(10^2)$).

\begin{figure}[ht]
\centering
\includegraphics[width=\linewidth]{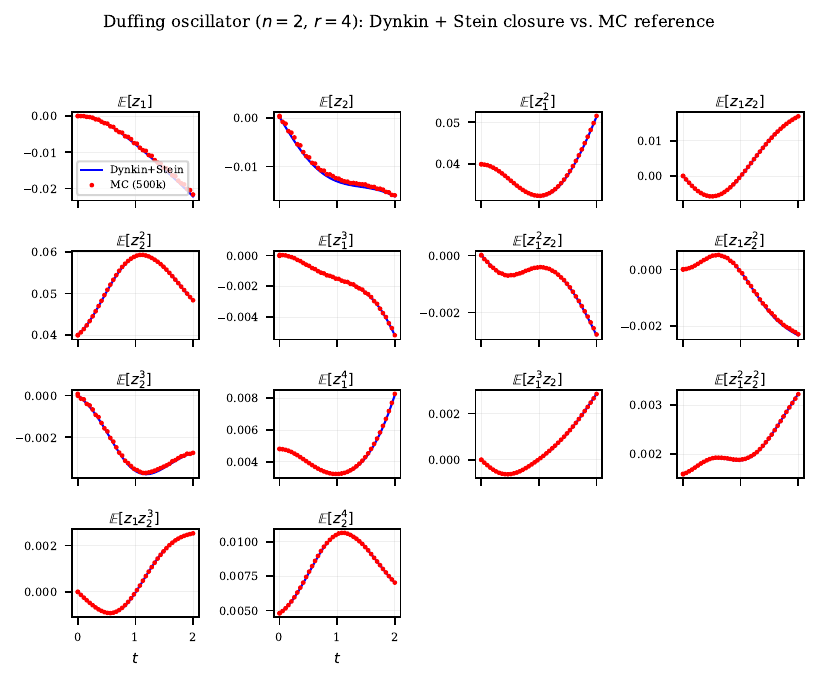}
\caption{Duffing ($n{=}2$, $r{=}4$): Dynkin + Stein (blue) vs.\ MC (red)
for all centered moments up to degree~4.
Relative errors are below $1\%$ across all displayed moments.}
\label{fig:moment-comparison-duffing}
\end{figure}
\begin{figure}[ht]
\centering
\includegraphics[width=\linewidth]{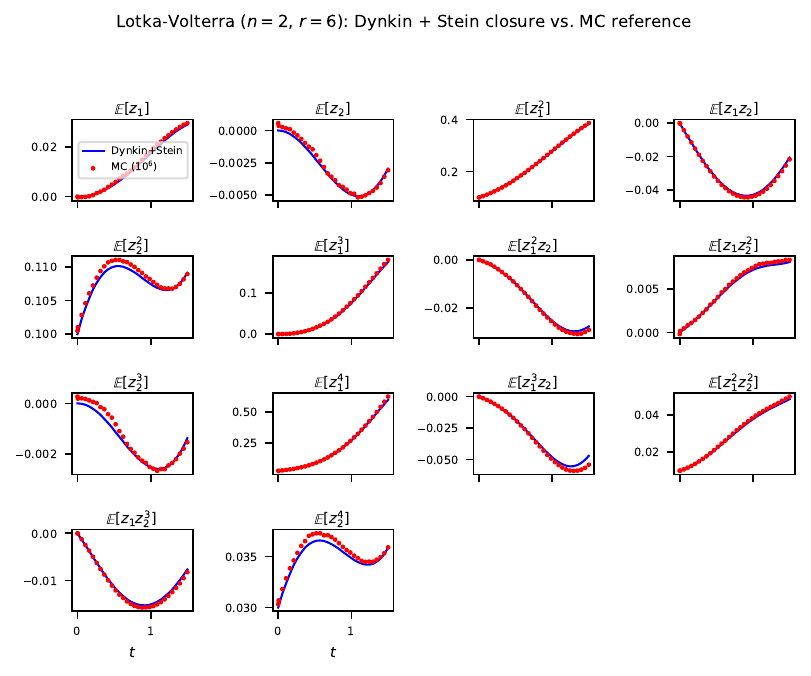}
\caption{Lotka-Volterra ($n{=}2$, $r{=}6$): Dynkin + Stein (blue) vs.\ MC (red).
Variance errors stay below $1\%$, while third- and fourth-order moments are
within $2$--$7\%$.}
\label{fig:moment-comparison-lv}
\end{figure}

\begin{figure}[t]
\centering
\includegraphics[width=\linewidth]{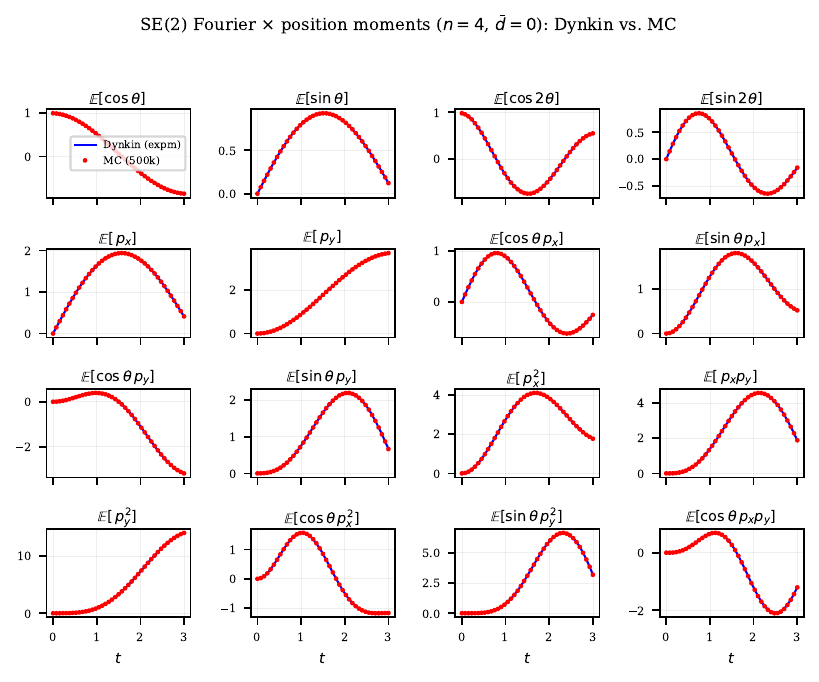}
\caption{SE(2) Fourier $\times$ position moments ($n{=}4$, $\bar{d}{=}0$):
Dynkin (blue) vs.\ MC (red) for 16 moments,
including the Fourier harmonics ($\cos\theta$, $\sin\theta$, $\cos 2\theta$, $\sin 2\theta$),
position means, second moments, and cross-moments
($\cos\theta \cdot p_x$, $\cos\theta \cdot p_x^2$, etc.).
Relative errors stay below $0.5\%$ on all of them.}
\label{fig:moment-comparison-se2}
\end{figure}

\begin{figure}[ht]
\centering
\includegraphics[width=\linewidth]{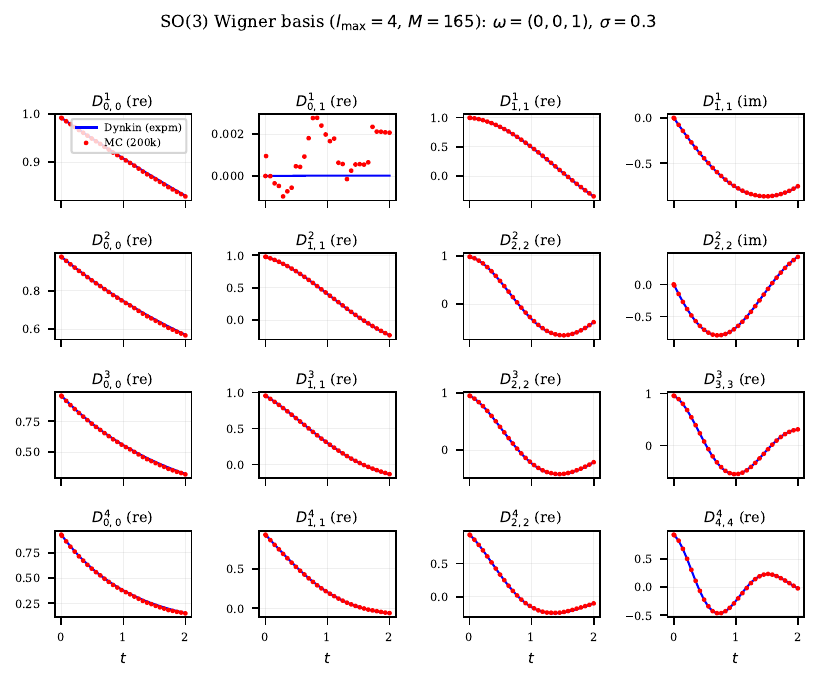}
\caption{SO(3) Wigner basis ($l_{\max}{=}4$, $M{=}165$):
Dynkin (blue) vs.\ MC (red) for representative Wigner moments.
Moments with appreciable magnitude ($D^1_{0,0}$, $D^1_{1,1}$) agree to
within ${\sim}0.3\%$, and moments that are near zero in the MC
reference remain near zero in the Dynkin trajectory.}
\label{fig:so3}
\end{figure}

\textbf{3D tracer advection} ($n{=}3$, $r{=}6$, $\bar{d}{=}1$):
\label{app:moment-comparison}
Figure~\ref{fig:fluid-moments} shows 12 selected centered moment
trajectories at degrees~2--5.
Second moments match MC within $1.3\%$, third moments within $3.2\%$,
and fourth moments within $5.7\%$ at $T{=}3$.
The largest errors are concentrated on pure $z_3$ powers, the
direction of the nonlinear quadratic lift.
The $z_1$ and $z_2$ moments, which are governed by linear equations,
match MC closely.
Figure~\ref{fig:fluid-error} summarizes the per-degree RMS scaled
error $|m_{\mathrm{SKF}}{-}m_{\mathrm{MC}}|/\max(|m_{\mathrm{MC}}|,\tau_d)$
where $\tau_d$ is the RMS moment magnitude at degree~$d$.
The hierarchy is clean: error grows monotonically with degree,
from ${<}1\%$ at degree~2 to ${\sim}8\%$ at degree~5.

\begin{figure}[ht]
\centering
\includegraphics[width=\linewidth]{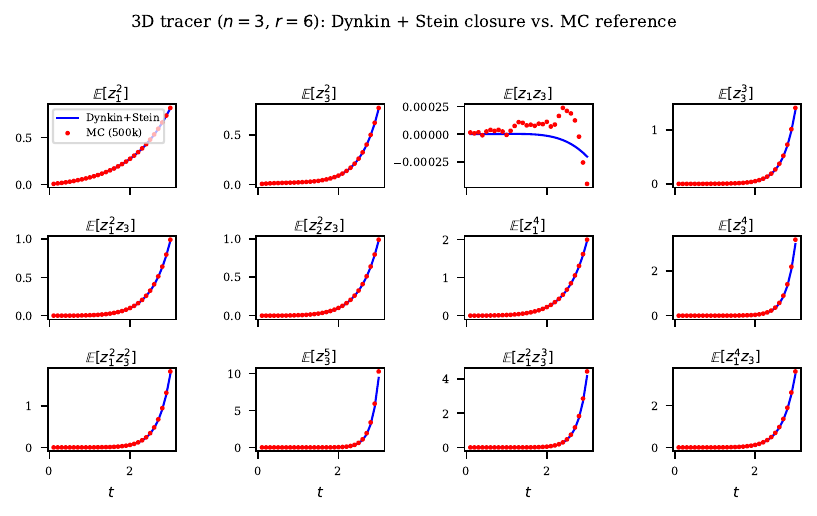}
\caption{Twelve selected centered moment trajectories for the
3D tracer ($n{=}3$, $r{=}6$): Dynkin + Stein (blue) vs.\ MC (red).
Degree 2--3 moments (rows 1--2) show excellent agreement through $T{=}3$.
Degree 4--5 moments (rows 3--4) track well through $T{=}2.5$,
with a small gap appearing at $T{=}3$.}
\label{fig:fluid-moments}
\end{figure}

\begin{figure}[ht]
\centering
\includegraphics[width=0.85\linewidth]{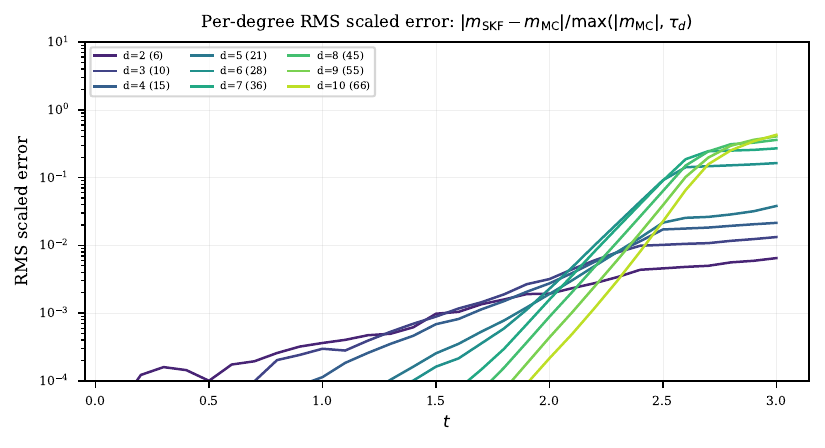}
\caption{Per-degree RMS scaled error for the 3D tracer.
Each curve aggregates all $\binom{d+2}{2}$ moments of degree~$d$.
The propagated state includes degrees up to $K{=}2r{-}2{=}10$
(not just the basis order $r{=}6$).
Degrees 2--5 remain below $8\%$ through $T{=}3$.
Higher degrees degrade faster, consistent with the
exponential amplification of the nonlinear $z_3$ coupling.}
\label{fig:fluid-error}
\end{figure}

\FloatBarrier

\subsection{Density reconstruction from propagated moments}
\label{app:additional-experiments}

The polynomial MED $p(x;\lambda) \propto \exp(-\lambda\cdot\phi(x))$
provides a closed-form density from the fitted parameters $\lambda$. We visualize this density on several systems, ordered as in Appendix~\ref{app:generators}.

\textbf{Lotka-Volterra} ($n{=}2$, $r{=}6$, $\bar{d}{=}1$):
Figure~\ref{fig:lv-appendix} shows the $(x_1, x_2)$ density of the
stochastic predator-prey system at $T = 1.5$~s. The bilinear interaction
$x_1 x_2$ produces an unclosed moment hierarchy that Stein closure
resolves (Section~\ref{sec:stein}). The polynomial MED at $r = 6$
reconstructs the anisotropic, non-Gaussian shape, matching the
$10^6$-particle MC reference. This density reconstruction uses initial
law $\mathcal{N}((\gamma/\delta,\alpha/\beta)^\top,
\mathrm{diag}(0.3^2,0.2^2))$.

\begin{figure}[t]
\centering
\includegraphics[width=0.85\linewidth]{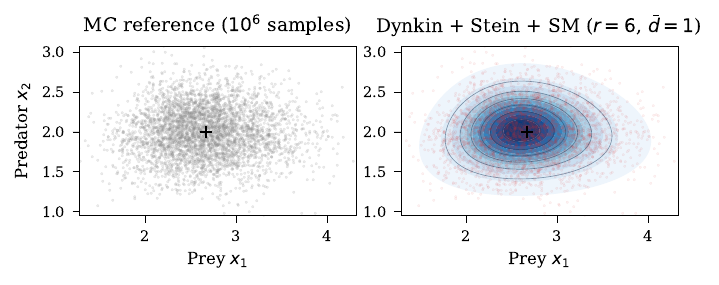}
\caption{Stochastic Lotka-Volterra ($\bar{d} = 1$):
the bilinear interaction $x_1 x_2$ produces an unclosed moment hierarchy
that is resolved algebraically by the Stein closure.
Left: $10^6$-particle MC reference.
Right: reconstructed density from Dynkin + Stein closure + score matching
($r = 6$), with MC samples overlaid in red.
The reconstruction tracks the anisotropic, non-Gaussian shape of the
reference.}
\label{fig:lv-appendix}
\end{figure}

\textbf{Double-well Langevin} ($n{=}2$, $r{=}6$, $\bar{d}{=}2$):
Figure~\ref{fig:double-well} shows a 2D double-well Langevin system
$dX = (-\nabla V(X))\,dt + \sigma\,dW$ with
$V(x) = \frac{1}{4}x_1^4 + \frac{1}{4}x_2^4 - \frac{1}{2}x_1^2 - \frac{1}{2}x_2^2$
($\sigma = 0.5$, $T = 3$~s).
Starting near the hilltop, the distribution splits into 4 modes at $(\pm 1, \pm 1)$.
Score matching ($r{=}6$, $M{=}28$, $\kappa(A) = 42$) reconstructs all 4 modes,
while the EKF Gaussian does not capture the multimodal structure
(kurtosis $1.4$ vs.\ $3$ for Gaussian).

\begin{figure}[ht]
\centering
\includegraphics[width=\linewidth]{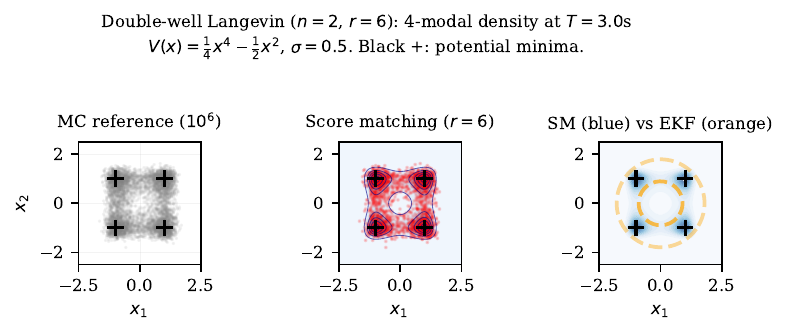}
\caption{Separable double-well Langevin ($n{=}2$, $r{=}6$):
$V(x)=\frac{1}{4}x_1^4+\frac{1}{4}x_2^4-\frac{1}{2}x_1^2-\frac{1}{2}x_2^2$, 4-modal density at $T = 3$~s.
Left: $10^6$-particle MC reference.
Center: score matching density reconstruction ($r{=}6$, $\kappa(A){=}42$) with MC samples (red).
Right: score matching (blue) vs.\ EKF Gaussian (orange, $1\sigma$/$2\sigma$).
Black crosses: potential minima at $(\pm 1, \pm 1)$.
The score-matching reconstruction reproduces the four-mode structure
of the reference. The Gaussian fit, restricted to a unimodal density,
does not.}
\label{fig:double-well}
\end{figure}

\textbf{3D tracer advection} ($n{=}3$, $r{=}6$, $\bar{d}{=}1$):
Starting from $X_0 \sim \mathcal{N}([0,0,0.5]^\top, 0.05^2 I)$, the
nonlinear coupling creates an asymmetric plume:
Figure~\ref{fig:fluid-appendix} shows the $(x_1, x_3)$ marginal of
the tracer density at four times. The quadratic lift
$\alpha(x_1^2{+}x_2^2)$ produces a parabolic plume that the EKF
(Gaussian, symmetric ellipse) cannot represent. SKF propagates 286
centered moments ($K{=}10$) and reconstructs the non-Gaussian marginal,
matching the 500k-particle MC reference. Moment accuracy: $\E[z_3^2]$
within $1.3\%$ and $\E[z_3^3]$ within $3.2\%$ of MC at $T{=}3$
(Appendix~\ref{app:moment-comparison}).
Timing: SKF $8$\,s, MC (500k) $122$\,s, EKF ${<}0.1$\,s.

\begin{figure}[t]
\centering
\includegraphics[width=\linewidth]{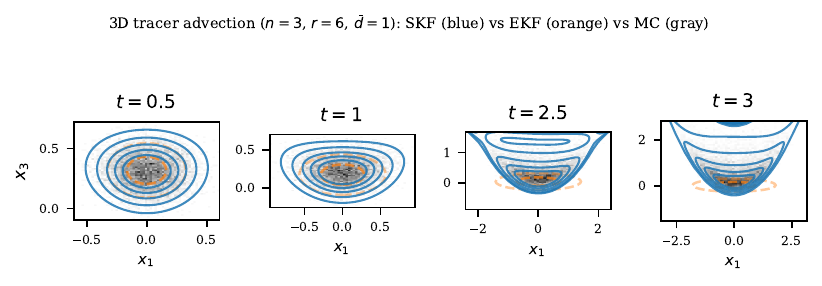}
\caption{3D tracer advection ($n{=}3$, $r{=}6$, $\bar{d}{=}1$):
$(x_1, x_3)$ marginal at four times.
The SKF reconstruction (blue) tracks the parabolic plume induced by
the quadratic lift, while the EKF reconstruction (orange) remains
symmetric and does not represent this curvature.
Reference: 500k MC particles (gray).}
\label{fig:fluid-appendix}
\end{figure}

\textbf{SE(2) localization} ($n{=}4$, $\bar{d}{=}0$):
The reconstructed $(p_x, p_y)$ marginal is Figure~\ref{fig:se2-density} in the main text. The non-Gaussian distribution arises from heading uncertainty coupled to translational motion, and the polynomial MED captures the curved, asymmetric shape starting at $r = 4$ ($M = 15$ parameters, one linear solve in $0.4$~ms). The MC reference uses $10^6$ particles.

\textbf{SO(3) attitude} ($l_{\max}{=}4$, $\bar{d}{=}0$):
Figure~\ref{fig:so3-sphere} shows the marginal density of each body-axis
direction on $S^2$ after rotation kinematics with angular velocity noise.
The MED fitted from Wigner D-matrix moments captures the concentration
near the north pole, matching the MC reference.
The Gaussian (EKF) approximation produces a similar shape for this
axially symmetric system. The two reconstructions are expected to
differ visibly under asymmetric dynamics or multimodal posteriors
from measurement ambiguity, which we defer to a companion paper on
Lie group filtering.

\begin{figure}[t]
\centering
\includegraphics[width=\linewidth]{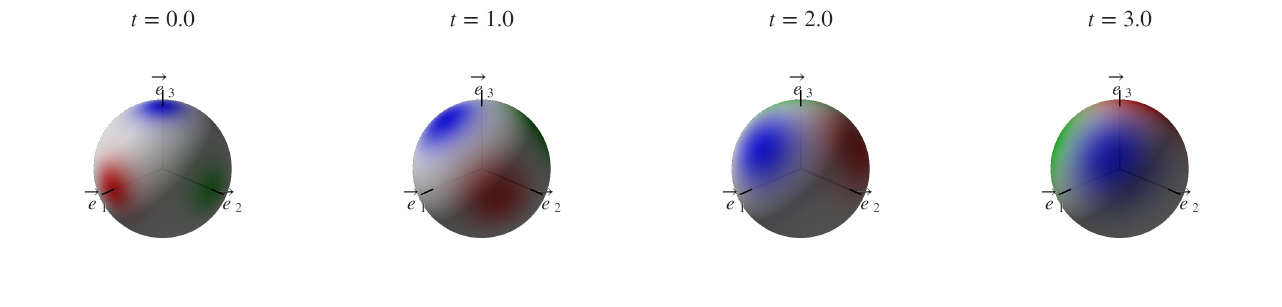}
\caption{SO(3) attitude density evolution: marginal distribution of each
body axis on $S^2$ (Red=$\vec{b}_1$, Green=$\vec{b}_2$, Blue=$\vec{b}_3$).
$\omega_0 = (0.5, 0.3, 1.0)$~rad/s, $\sigma = 0.15$.
At $t{=}0$ the distribution is concentrated near the identity rotation.
As time evolves, the tilted rotation creates asymmetric spreading.}
\label{fig:so3-sphere}
\end{figure}

\textbf{SE(3) position uncertainty} ($n{=}6$, $\bar{d}{=}1$):
Figure~\ref{fig:se3-banana} shows the $(x,y)$ marginal
of the position density for an SE(3) rigid body following a curved
trajectory with rotation noise
($T_{k+1} = T_k \exp(U_k + w_k)$, $w_k \sim \mathcal{N}(0,Q)$
in $\mathfrak{se}(3)$).
Moments are computed from $5 \times 10^5$ MC samples. Dynkin
propagation on SE(3) requires a Wigner $\times$ Legendre basis and is
left for the Lie-group filter extension.
The rotation-translation coupling produces a curved, asymmetric shape
that the Gaussian ($r{=}2$) does not capture.
At $r{=}5$ ($M{=}56$, $\kappa(A) = 86$ after Jacobi preconditioning),
the reconstructed density matches the MC reference.

\begin{figure}[t]
\centering
\includegraphics[width=\linewidth]{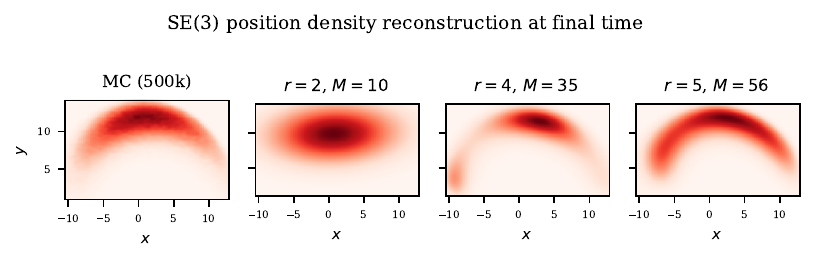}
\caption{SE(3) position density at the final time step,
shown as the $(x,y)$ marginal.
The leftmost panel is the 500k-particle MC reference.
The remaining panels show the score-matching reconstruction at
$r = 2$, $4$, and $5$.
The non-Gaussian curvature induced by rotation-translation coupling
is captured from $r = 4$ onward.}
\label{fig:se3-banana}
\end{figure}

\subsection{Duffing filter}
\label{app:duffing-filter}

The single Duffing oscillator ($n{=}2$, $r{=}4$, $\bar{d}{=}1$) with
$\sigma{=}0.15$ serves as a basic validation of the full predict-update
loop.
The measurement model is $y = x_1 + v$ with $v \sim \mathcal{N}(0, 0.04)$,
which is linear ($d_g{=}1$) and therefore conjugate for any $r \ge 2$.
We run 25 predict-update cycles at $\Delta t_{\mathrm{pred}} = 0.2$\,s,
with prediction via Dynkin + Stein closure (RK4) and update via the
conjugate score addition $\lambda^+ = \lambda^- + \lambda_{\mathrm{lik}}$
followed by Stein moment recovery (42 equations, 27 unknowns,
full column rank).
The initial law is $\mathcal{N}((0.5,0)^\top,0.1^2 I)$.
All computations are performed in centered coordinates $z = x - \mu$,
so each step reduces to a single linear system solve.

On this mildly nonlinear system, all methods track the ground truth
(Figure~\ref{fig:duffing-filter}).
The EKF, UKF, and EnKF perform comparably, indicating that for weak
nonlinearity the Gaussian approximation itself is adequate.
A clearer separation between the SKF and the Gaussian baselines emerges
at higher dimensions with stronger nonlinearity
(Section~\ref{sec:exp-highdim}).

\begin{figure}[t]
\centering
\includegraphics[width=\linewidth]{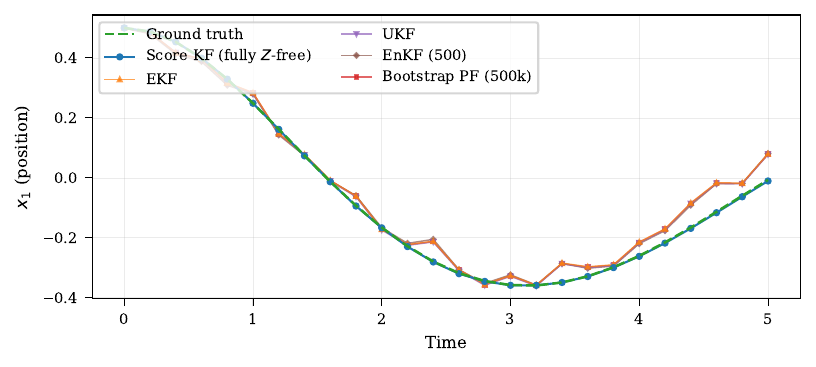}
\caption{Score Kalman Filter on the Duffing oscillator
($n{=}2$, $r{=}4$).
Ground truth (green dashed), SKF (blue), EKF (orange), UKF (purple),
EnKF (brown), and bootstrap PF ($5{\times}10^5$ particles, red).
On this mildly nonlinear benchmark all methods track the ground truth.}
\label{fig:duffing-filter}
\end{figure}

\subsection{Coupled oscillator filter time series}
\label{app:coupled-figures}

Figures~\ref{fig:coupled-n4}--\ref{fig:coupled-n20} show the full
predict-update filter trajectories for the coupled Duffing oscillator
network at $n{=}4,6,\ldots,20$.
At all dimensions, the SKF (blue) tracks the ground truth (green dashed)
more closely than the reported baselines: EKF (orange), UKF (purple),
EnKF (brown), and the 500k-particle bootstrap PF (red).
The EKF, UKF, and EnKF cluster together at similar RMSE,
suggesting that the limiting factor is the Gaussian approximation,
not the specific approximation strategy (linearization vs.\ sigma points vs.\ ensemble).
The PF RMSE is also close to these Gaussian baselines under partial
observation of the odd-indexed positions.
For $n{>}10$, the $r{=}3$ filter uses active Stein closure for the
degree-five moments requested by the moment ODE. This keeps the closure
overdetermined without introducing a chordal or banded sparsification.
Because this algorithmic change begins at $n{=}12$, the drop in RMSE from
$n{=}10$ to $n{=}12$ should be read as the effect of the active closure on
this structured generator, not as a pure dimension-scaling phenomenon.
The SKF infers the unobserved oscillators from coupling alone,
exploiting the full moment-based belief representation.
Quantitative RMSE comparisons are in Table~\ref{tab:comparison}.
For reproducibility, the coupled oscillator filter uses
$\gamma=0.3$, $\alpha=1.0$, $\beta=0.6$, $\kappa=0.3$, $\sigma=0.4$,
$\Delta t_{\mathrm{pred}}=0.15$, RK4 step $\Delta t_{\mathrm{ode}}=0.005$,
and $25$ predict-update cycles. The observed coordinates are
$q_1,q_3,\ldots$ with measurement covariance $R=0.3^2 I$.
The initial law is Gaussian with covariance $0.15^2 I$ and mean
$(q_1,\ldots,q_N,p_1,\ldots,p_N)$ whose position block cycles through
$(0.3,-0.2,0.1,-0.3,0.15,0.25,-0.1,0.2)$ and whose momentum block is zero.

\begin{figure}[t]
\centering
\includegraphics[width=\linewidth]{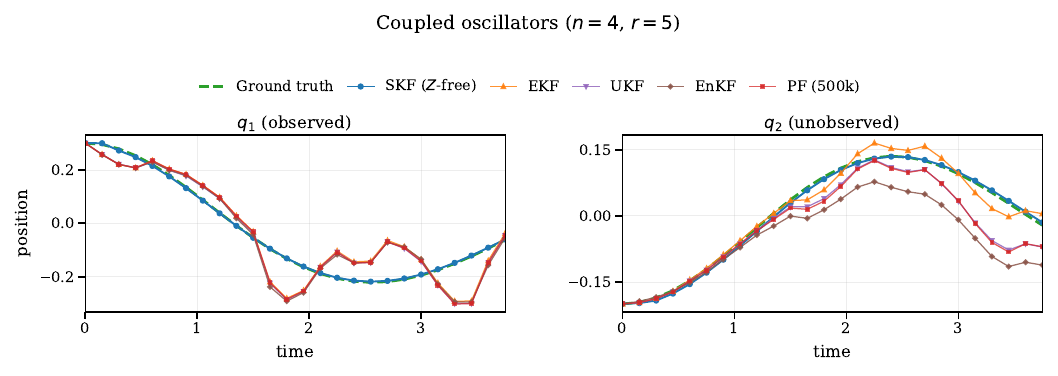}
\caption{Coupled oscillators ($n{=}4$, $r{=}5$).
Mean RMSE over 25 steps: SKF $0.005$, EKF $0.057$, UKF $0.061$,
EnKF $0.067$, PF $0.062$.}
\label{fig:coupled-n4}
\end{figure}

\begin{figure}[t]
\centering
\includegraphics[width=\linewidth]{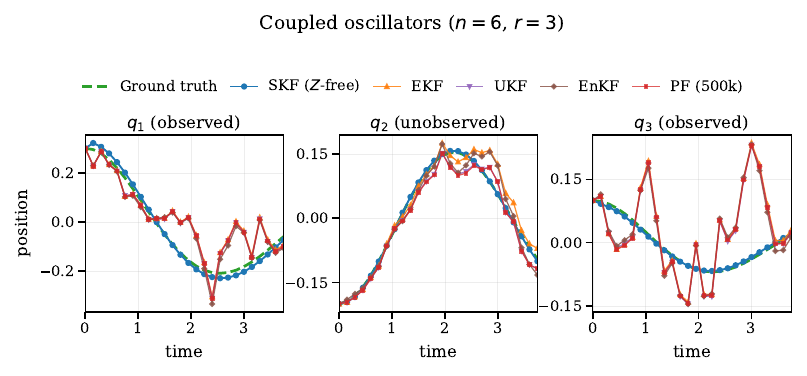}
\caption{Coupled oscillators ($n{=}6$, $r{=}3$).
Mean RMSE over 25 steps: SKF $0.011$, EKF $0.078$, UKF $0.077$,
EnKF $0.078$, PF $0.079$.}
\label{fig:coupled-n6}
\end{figure}

\begin{figure}[t]
\centering
\includegraphics[width=\linewidth]{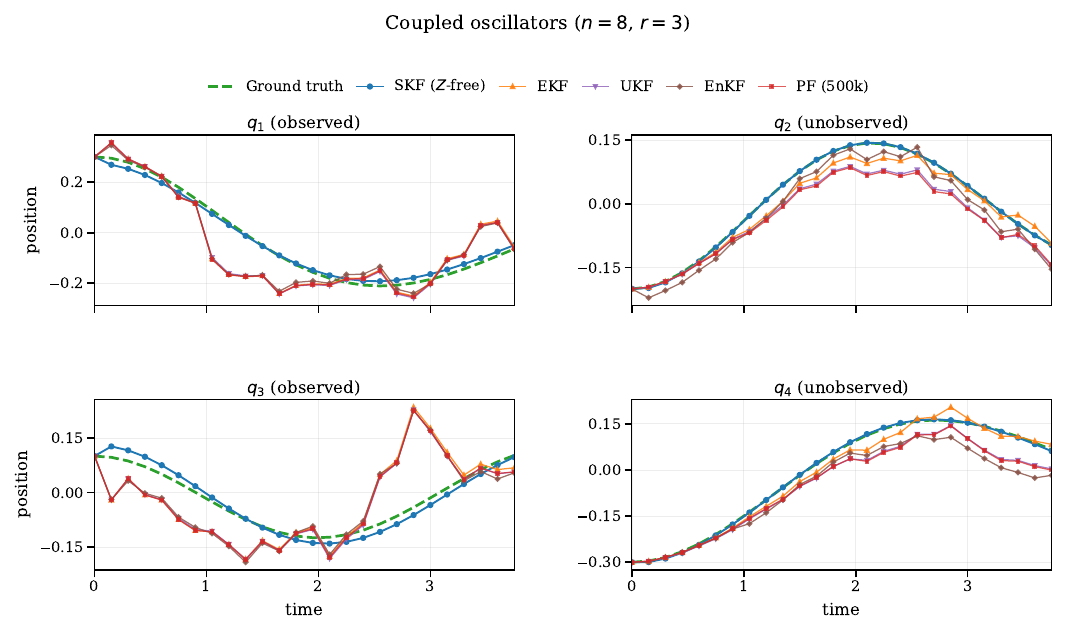}
\caption{Coupled oscillators ($n{=}8$, $r{=}3$).
Mean RMSE over 25 steps: SKF $0.013$, EKF $0.068$, UKF $0.070$,
EnKF $0.073$, PF $0.072$.}
\label{fig:coupled-n8}
\end{figure}

\begin{figure}[t]
\centering
\includegraphics[width=\linewidth]{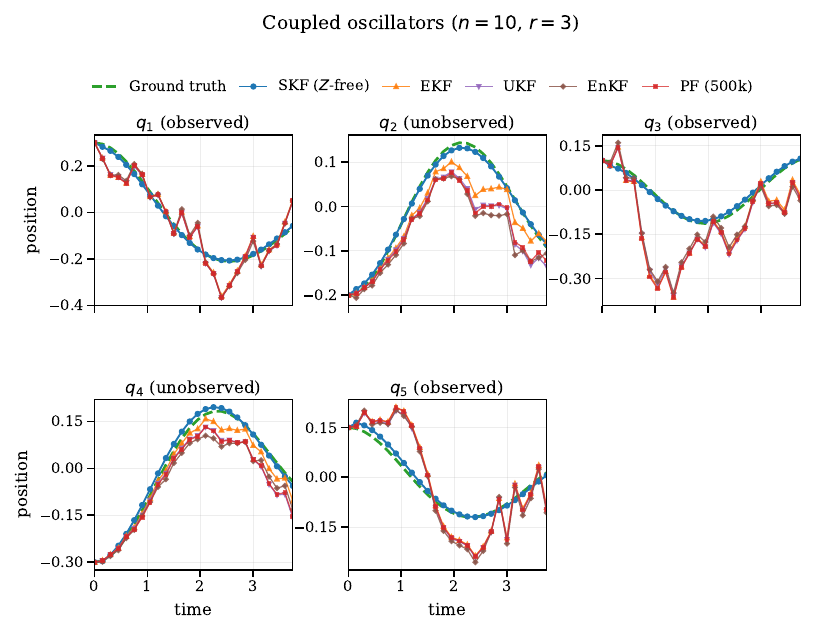}
\caption{Coupled oscillators ($n{=}10$, $r{=}3$).
Mean RMSE over 25 steps: SKF $0.012$, EKF $0.080$, UKF $0.081$,
EnKF $0.083$, PF $0.083$.}
\label{fig:coupled-n10}
\end{figure}

\begin{figure}[!b]
\centering
\includegraphics[width=\linewidth,height=.86\textheight,keepaspectratio]{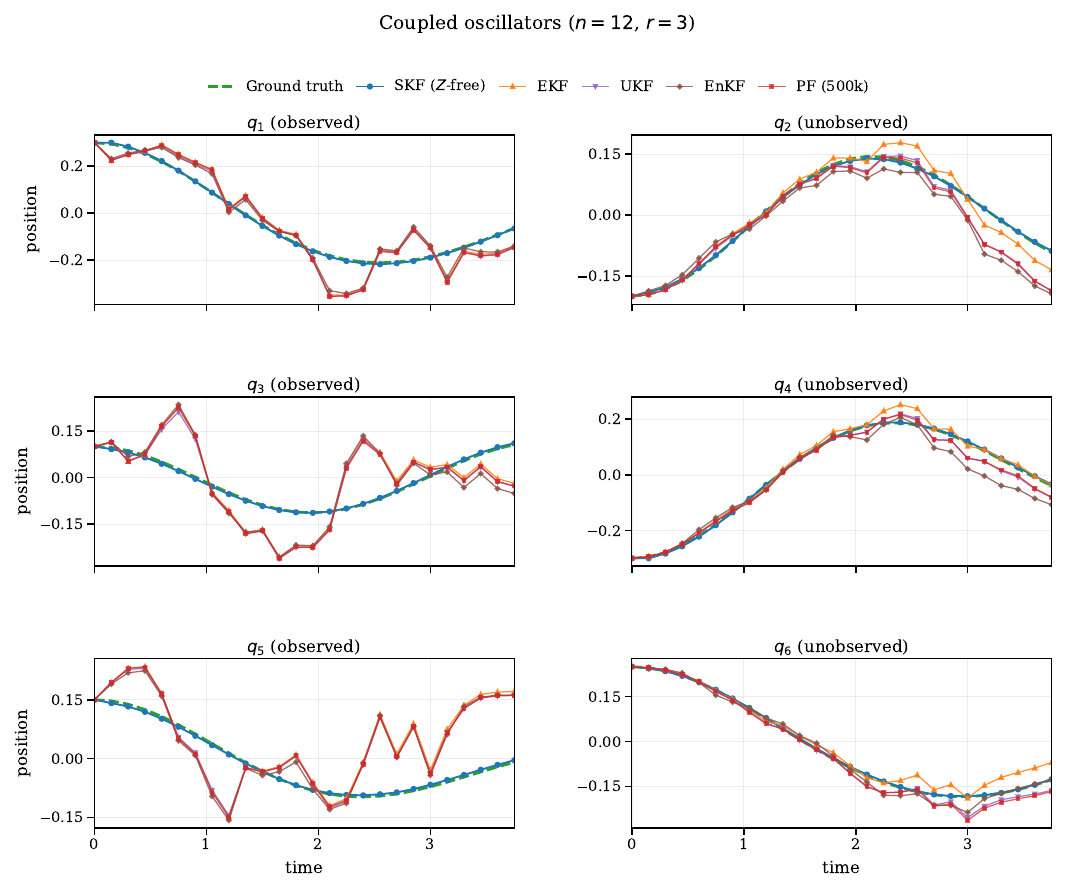}
\caption{Coupled oscillators ($n{=}12$, $r{=}3$, active Stein closure).
Mean RMSE over 25 steps: SKF $0.0040$, EKF $0.0696$, UKF $0.0718$,
EnKF $0.0756$, PF $0.0731$.}
\label{fig:coupled-n12}
\end{figure}

\begin{figure}[p]
\centering
\includegraphics[width=\linewidth,height=.86\textheight,keepaspectratio]{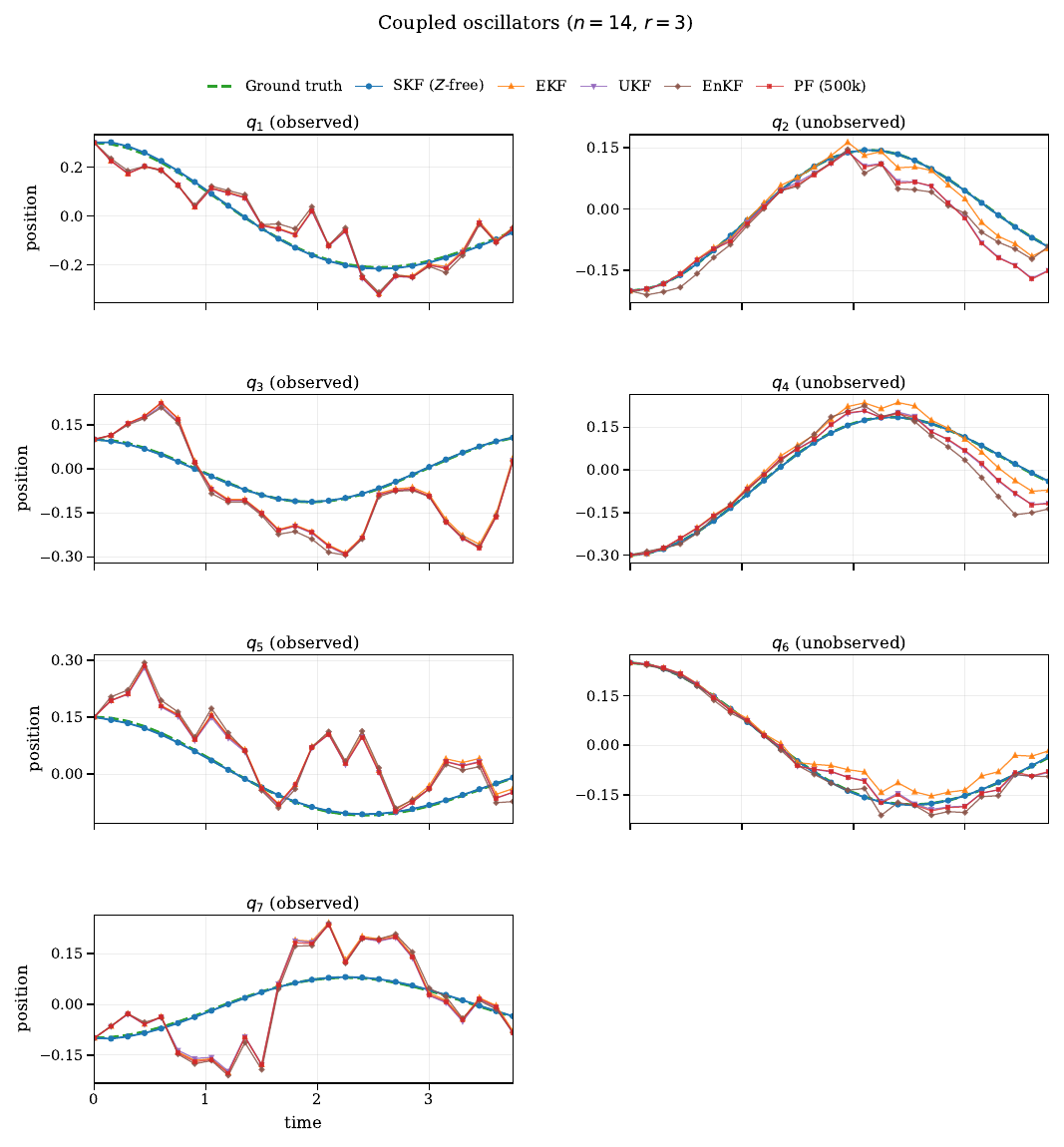}
\caption{Coupled oscillators ($n{=}14$, $r{=}3$, active Stein closure).
Mean RMSE over 25 steps: SKF $0.0024$, EKF $0.0814$, UKF $0.0836$,
EnKF $0.0880$, PF $0.0843$.}
\label{fig:coupled-n14}
\end{figure}

\begin{figure}[p]
\centering
\includegraphics[width=\linewidth,height=.86\textheight,keepaspectratio]{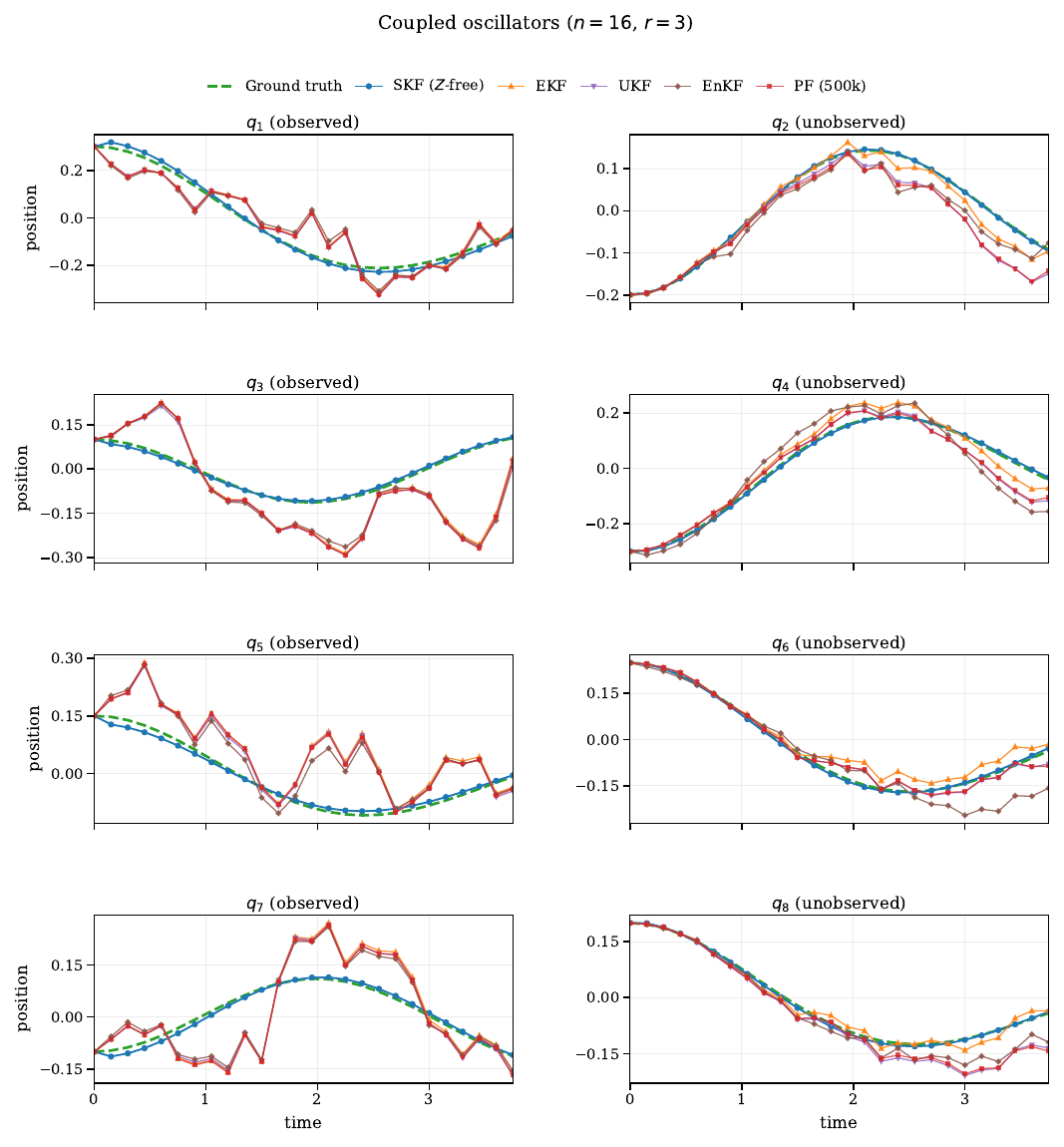}
\caption{Coupled oscillators ($n{=}16$, $r{=}3$, active Stein closure).
Mean RMSE over 25 steps: SKF $0.0085$, EKF $0.0763$, UKF $0.0795$,
EnKF $0.0845$, PF $0.0796$.}
\label{fig:coupled-n16}
\end{figure}

\begin{figure}[p]
\centering
\includegraphics[width=\linewidth,height=.86\textheight,keepaspectratio]{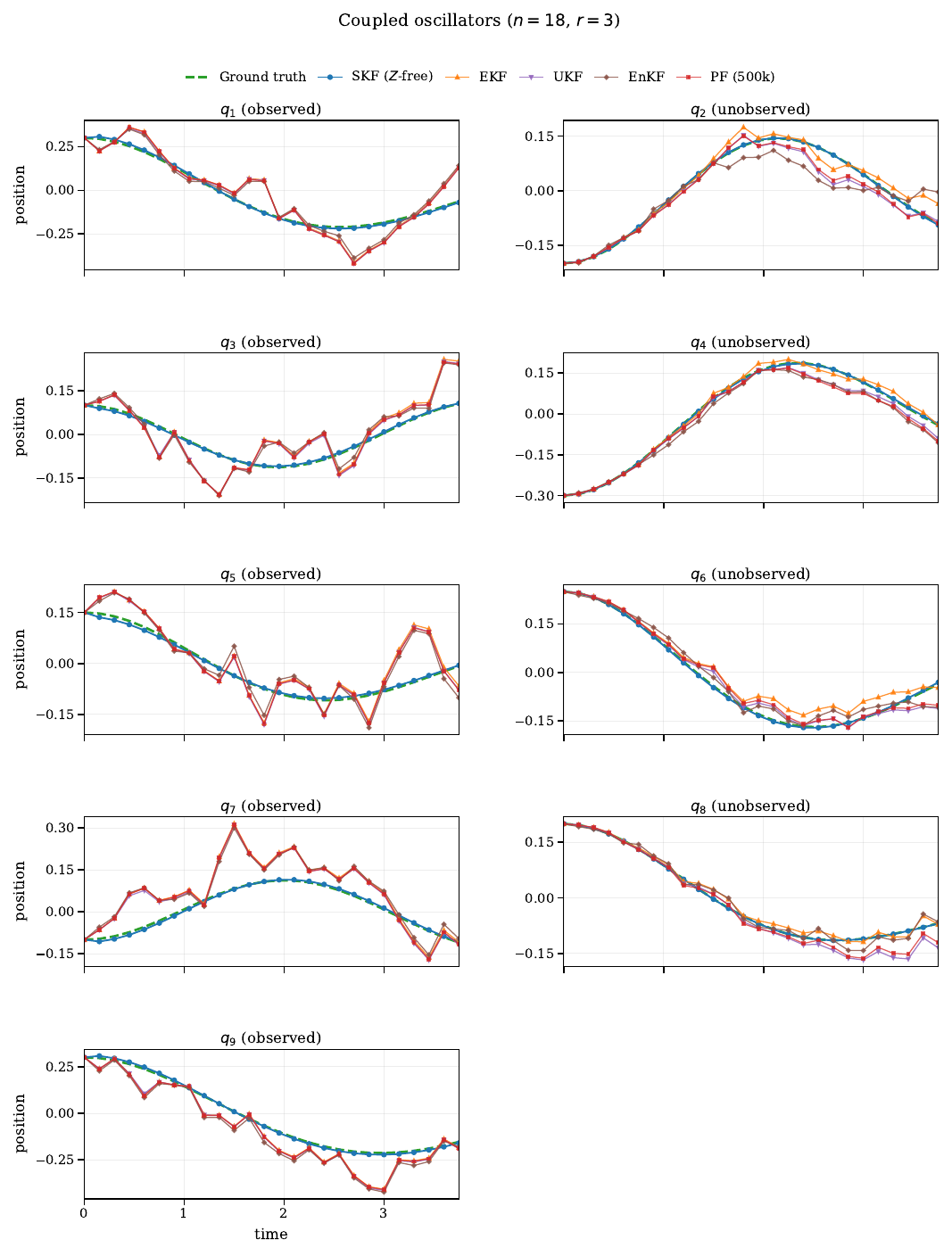}
\caption{Coupled oscillators ($n{=}18$, $r{=}3$, active Stein closure).
Mean RMSE over 25 steps: SKF $0.0052$, EKF $0.0646$, UKF $0.0660$,
EnKF $0.0673$, PF $0.0668$.}
\label{fig:coupled-n18}
\end{figure}

\begin{figure}[p]
\centering
\includegraphics[width=\linewidth,height=.86\textheight,keepaspectratio]{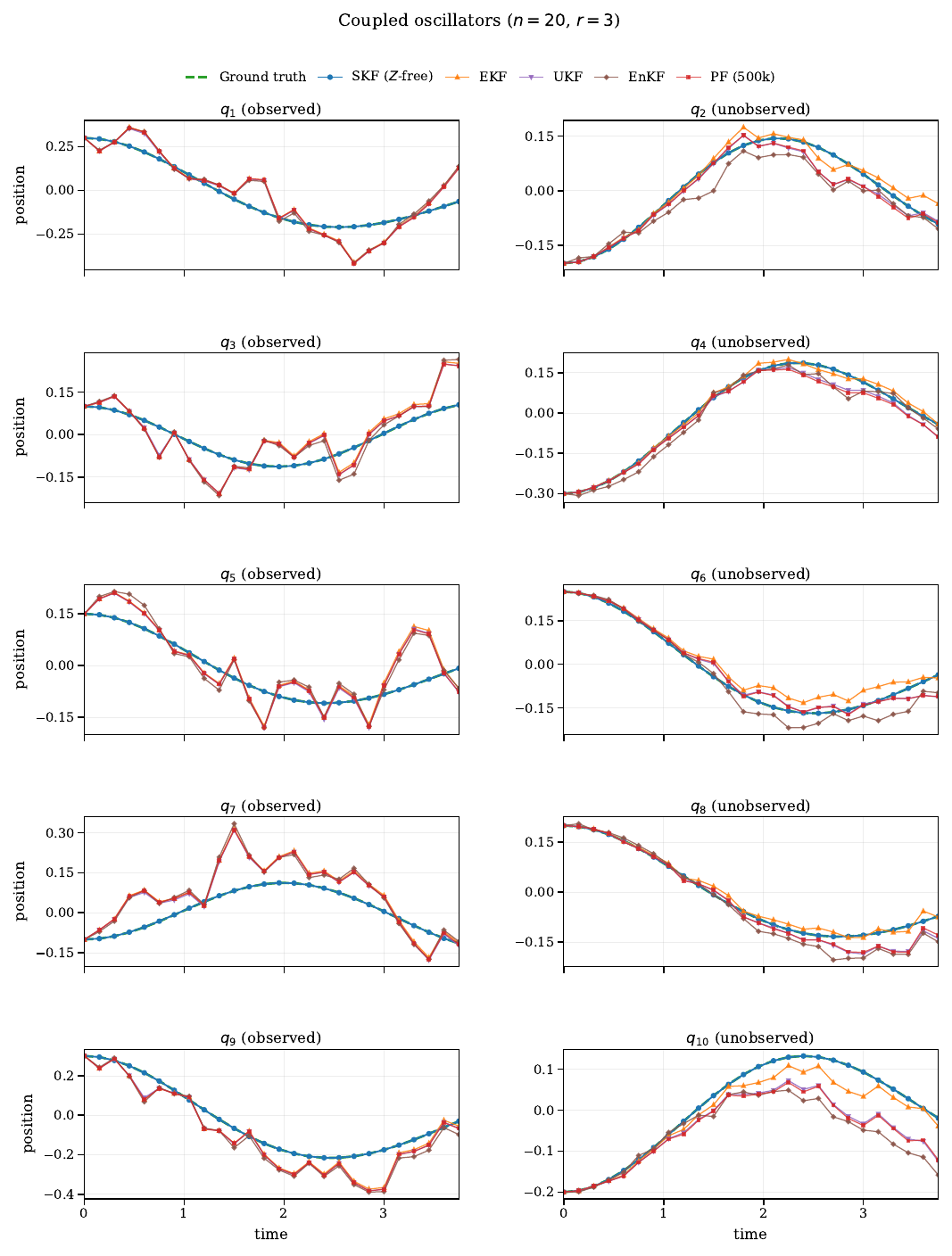}
\caption{Coupled oscillators ($n{=}20$, $r{=}3$, active Stein closure).
Mean RMSE over 25 steps: SKF $0.0002$, EKF $0.0617$, UKF $0.0644$,
EnKF $0.0680$, PF $0.0650$.}
\label{fig:coupled-n20}
\end{figure}

\section{Stein perspective on the duality of parameter and moment recovery}
\label{app:stein-duality}

This section provides a unified, high-level interpretation of both score parameter estimation (Proposition~\ref{prop:sm}) and moment recovery (Proposition~\ref{prop:stein}) through the Stein identity. Rather than viewing these procedures as separate constructions, we show that they arise from the same underlying weak formulation.

\paragraph{Stein identity as a weak constraint.}
For a smooth density $p$ with score $s(x) = \nabla \log p(x)$, the Stein identity~\cite{stein1972} states that
\begin{equation}\label{eq:stein-weak}
    \E[s(x) \cdot \nabla f(x) + \Delta f(x)] = 0,
\end{equation}
for all sufficiently regular test functions $f$. This is naturally interpreted as a \emph{weak formulation}: instead of imposing a pointwise condition on $p$, the constraint is enforced in an integrated sense against arbitrary test functions. This weak-form perspective underlies a range of modern computational-statistics methods built on Stein operators, including kernel Stein discrepancies~\cite{gorham2015} and minimum Stein discrepancy estimation~\cite{barp2019}. Anastasiou et al.~\cite{anastasiou2023} review these developments.

\paragraph{Score matching as a projected Stein system.}
The score matching formulation in Proposition~\ref{prop:sm} can be interpreted as a projection of the weak Stein identity onto a finite-dimensional test space with a parametric score model, paralleling the broader Stein-discrepancy framing of score-matching estimators in~\cite{barp2019,kume2026}. Restrict the test functions to
\[
    \mathcal{V}_r = \mathrm{span}\{\phi_\alpha(x) = x^\alpha \mid |\alpha| \le r\}.
\]
Substituting the score model $s_\lambda(x) = -J_\phi(x)\lambda$ into~\eqref{eq:stein-weak} and taking $f = \phi_\beta$ yields
\[
    \E_p\!\left[ -(J_\phi \lambda) \cdot \nabla \phi_\beta + \Delta \phi_\beta \right] = 0,
\]
for all $|\beta| \le r$. Collecting these equations for all basis functions gives
\[
    \E_p[J_\phi^\top J_\phi]\, \lambda = \E_p[\Delta \phi],
\]
which recovers~\eqref{eq:lambda-star}.

\paragraph{Moment recovery as the dual problem.}
Conversely, if the score parameters $\lambda$ are given, the same projected Stein identity becomes a system of linear relations among the moments. The unknowns are now the higher-order moments, and the Stein equations provide a mechanism for reconstructing them from lower-order statistics. This is precisely the role of the Stein closure in Proposition~\ref{prop:stein}.

In the ideal (infinite-dimensional) setting, the Stein identity characterizes the distribution through an infinite family of constraints. In practice, both the test space and the moment representation are truncated. For moment recovery we use the component-wise identities underlying Proposition~\ref{prop:stein}, not only the scalar collapsed form~\eqref{eq:stein-weak}. With $x^\beta$ restricted to $\mathcal{V}_K = \mathrm{span}\{x^\beta : |\beta| \le K\}$, given $\lambda$ and moments $m^{(\le K)}$, the moments of degree $K{+}1$ are obtained by minimizing the residual of the projected component-wise Stein operator:
\[
    m^{(K+1)} = \arg\min_{m} \sum_{|\beta| \le K}\sum_{i=1}^n
    \left( \E\!\left[ s_{\lambda,i}(x)x^\beta + \partial_i x^\beta \right] \right)^2.
\]

\paragraph{Summary.}
Score matching and Stein-based moment recovery arise from the same fundamental principle: enforcing the Stein identity over a finite set of test functions. The difference lies only in the choice of unknowns. Score matching solves for the score parameters given moments (Proposition~\ref{prop:sm}). Stein closure solves for moments given the score parameters (Proposition~\ref{prop:stein}). In both cases, an infinite collection of constraints is approximated by a finite one, typically solved in a least-squares sense due to truncation. This unifies parameter estimation and moment closure as dual instances of the same projection principle.

This connects to weak formulations of stochastic dynamics more broadly. Moment propagation follows from Dynkin's formula, which describes how expectations evolve over time, while the Stein identity enforces consistency between the score representation and the moments at each time. The proposed framework can therefore be viewed as combining two complementary projections: one governing temporal evolution (via Dynkin) and one enforcing distributional consistency (via Stein). Together, they provide a finite-dimensional approximation of the underlying density evolution problem.
\section{Future research directions and open questions}
\label{app:limitations}

(i) \emph{Off-model moment consistency.}
The reconstruction step minimizes Fisher divergence rather than KL
divergence. When the true density lies in the polynomial exponential
family, Theorem~\ref{thm:sm-maxent} shows that the distinction
disappears: score matching recovers the same parameters as maximum
entropy. Away from the model class, however, the fitted density need
not reproduce every propagated moment exactly. The Stein augmentation
in Section~\ref{sec:stein} reduces this mismatch, but the remaining
error should depend on both the basis order $r$ and the regularity of
the target density. Making that dependence explicit is still open.

(ii) \emph{Automatic bases and conditioning.}
The theory is basis-invariant, but the linear algebra is not.
Raw monomials give $\kappa(A) \sim 10^4$--$10^6$ in practice.
Centering coordinates and switching to orthogonal bases such as
Legendre polynomials bring this down to $O(10)$--$O(100)$ in the
examples here. For that reason, centered orthogonal coordinates are the
default implementation rather than a cosmetic change
(Appendix~\ref{app:centering} gives the centering formulas and
Appendix~\ref{app:basis-change} gives the change-of-basis theory).
A practical open problem is that this basis choice is still hand tuned.
Automatic basis selection or adaptive preconditioning along the filter
trajectory remains future work.

(iii) \emph{Lie group score matching beyond compact groups.}
On SO(3), the Wigner D-matrix basis is already intrinsic:
the block-diagonal moment propagation follows from the
representation-theoretic structure of the generator
(Casimir element $+$ Lie algebra drift), and the embedding-coordinate
derivation in Appendix~\ref{app:so3-generator} is merely a
convenient route to the same result.
Score matching on general Riemannian manifolds is
well-studied~\cite{debortoli2022riemannian},
and on compact groups (e.g., tori, SO($N$)),
harmonic analysis provides a finite spectral basis with
closed-form moment propagation~\cite{wang2022lie}.
The gap is in noncompact groups such as SE(3) $=$ SO(3) $\ltimes \R^3$:
the irreducible unitary representations are infinite-dimensional
(induced representations from the semidirect structure).
This is exactly where the compact-group story stops being automatic:
we do not know a finite spectral basis with comparable moment-closure
properties.
Our SE(2) experiments use a product basis
(Fourier for the compact SO(2) factor, Legendre for $\R^2$)
adapted to the semidirect structure, but this is not the full SE(2)
harmonic analysis.
For SE(3), the natural product basis is Wigner $\times$ Legendre
(Wigner for the rotation subgroup, Legendre for the translational
coordinates).
The SE(3) density reconstruction in
Appendix~\ref{app:additional-experiments}
uses MC-sampled moments as a proof of concept.
Dynkin + Stein closure on SE(3) with the Wigner $\times$ Legendre basis,
and the corresponding predict-update filter on Lie groups, are the
next step rather than part of the present paper.

(iv) \emph{Stein closure counts and conditioning.}
There are two questions hiding under ``well-posedness.''
The first is combinatorial: does the finite Stein system provide enough
equations for the degree-$(K{+}1)$ moments?
The second is numerical: once there are enough equations, is the
resulting matrix full column rank and well conditioned along the
trajectory?
Appendix~\ref{app:stein-wellposed} answers the counting part for the
full-basis system and gives a sufficient condition for when the Stein
closure is overdetermined.
For $r{=}3$, the standard augmented closure is overdetermined through
$n{=}15$ (including the even-dimensional oscillator cases
$n \le 14$) and becomes underdetermined at $n{=}16$.
The oscillator sweep at $n{=}12$--$20$ uses the generator itself:
the active closure solves only for the degree-five moments requested
by the moment ODE, a restriction that becomes essential at
$n{=}16,18,20$.
At $n{=}20$, this leaves $14{,}500$ targets and $32{,}800$ Stein
equations.
The remaining gap is not the equation count but a structural rank and
conditioning theory for these counted and active systems.
Even with full column rank, $\Lambda_1(\lambda)$ can become poorly
conditioned as the fitted parameters $\lambda(t)$ approach the
bifurcation surface $\mathcal{B}$ discussed below, making the closure
numerically fragile near modality transitions.

(v) \emph{Stein recovery is a projection after truncation.}
The posterior moment recovery (Section~\ref{sec:recovery})
inverts the parameter-to-moment map via a truncated Stein system.
The self-contained system (moments $\le K$) is underdetermined.
The system we use instead adds higher-$|\beta|$ equations and
truncates the moments above the tracked degree.
This is a projection onto the consistency manifold of
the $(\lambda, m)$ map, not an exact inversion.
After propagation and measurement update, the moments being
projected are no longer exactly on that manifold, and the
projection can move the recovered moments away from the
physically propagated ones---even when the fitted density
looks good.
The approximation is accurate in centered coordinates on $\R^n$
(Duffing: residual $\sim 10^{-5}$--$10^{-9}$), but on manifold
state spaces with truncated product algebras (e.g., the
trigonometric ring for SE(2)), the product truncation introduces
additional error.
The iterative SM-Stein consistency refinement
(Appendix~\ref{app:smom}) reduces this truncation error
empirically by 2--3 orders of magnitude, but a formal
convergence proof for the fixed-point iteration remains open.

(vi) \emph{Generator-aware active closure.}
The $n{=}12$--$20$ coupled oscillator experiments exploit the generator:
only the degree-five moments produced by the local quadratic force
need to be closed. This active closure is structured in a useful way,
because it is tied directly to the moment ODE, but it does not amount
to a general chordal or banded sparse moment method. A system whose
generator requests a dense set of high-degree moments may still require
larger $r$, extended Stein equations, or a genuinely sparse
basis/closure construction. This is where a true sparse moment method
would enter beyond the present $n{=}20$ demonstration. We view this as
a promising direction for future work building on the SKF framework
developed here.

(vii) \emph{Model-based filtering under misspecification.}
Like MEM-KF and other model-based nonlinear filters, the SKF assumes
that the dynamics, measurement model, and noise statistics are accurate
enough for moment propagation to be useful. The contribution here is to
remove the partition-function bottleneck in MaxEnt reconstruction. It
does not by itself solve model misspecification or partially known
state-space models.

(viii) \emph{Polynomial families favor moderate-order structure.}
The polynomial exponential family is most effective when the filtering
density is reasonably captured by finitely many moderate-order moments
in the chosen basis.
For broad heavy-tailed distributions, or beliefs requiring very high-order
moments to represent accurately, richer bases or different density
families may be more appropriate.
In the settings considered here (nonlinear oscillators, predator-prey
dynamics, fluid tracer advection, rigid body kinematics), the dominant
challenge is multimodality and skewness rather than extreme tail behavior,
which is why this family is a useful approximation. Broadening the
scope to heavy-tailed or multi-scale distributions would require either
different sufficient statistics or a different density family.

(ix) \emph{Boundaries and hybrid jumps reintroduce density values.}
The present work addresses SDEs on open domains without boundaries.
More generally, moment propagation via Dynkin's formula can involve
additional integral terms (e.g., boundary flux in guard-triggered
hybrid systems, codimension-1 surface integrals for reflecting
boundaries) where the probability current across the boundary
requires pointwise density evaluation (not just the score)
to ensure mass conservation. This reintroduces the partition function
$Z$, and with it the obstacle that SKF was designed to avoid. Extending
the method to stochastic hybrid systems or domains with nontrivial
boundary conditions therefore requires new boundary-aware closure
machinery.

(x) \emph{Modality transitions for single-component closures when
$\bar{d} \ge 2$.}
\label{item:dbar2}

A key open case is filtering through changes in modality. The
score-matching step can still be well behaved, but the propagated moments
must first survive the closure. For filtering, this issue is most relevant
when prediction must carry a multimodal belief across a long, noisy, or
partially observed interval. Frequent informative measurements can collapse
the posterior before the closure error dominates. If those moments are
already consistent with the polynomial exponential family,
Theorem~\ref{thm:sm-maxent} recovers the corresponding parameters. The
trouble begins earlier, inside the closure used to propagate those
moments. When $\bar{d}\ge 2$, Stein closure has to resolve more than
one layer of missing moments. A single component that moves from one
mode to two modes also has to pass through a degenerate geometry where
the parameter-to-moment map is badly conditioned. In that regime the
multi-layer Stein backsolve can amplify small errors until the closed
moment ODE becomes unusable.
The double-well example separates representation from propagation. The
polynomial MED can still \emph{represent} the multimodal density
(Figure~\ref{fig:double-well}: $r{=}6$, $\kappa(A){=}42$, all four modes
recovered), but the single-component closure does not propagate cleanly
through the modal split.

\textbf{Bifurcation surface.}
As an illustrating example, we use one-dimensional systems to show this obstruction. Consider the
polynomial exponential family
$p(x;\lambda) \propto \exp(-E(x;\lambda))$ with
$E = \sum_{k=1}^{2r}\lambda_k x^k$ and $\lambda_{2r} > 0$
(normalizability).
The derivative $E'(x) = \sum_{k=1}^{2r} k\lambda_k x^{k-1}$
is a polynomial of degree $2r{-}1$ in $x$ whose coefficients
depend on $\lambda$.
The density $p$ is unimodal when $E$ has a single local minimum,
i.e., $E'$ has exactly one real root,
and bimodal when $E$ has two local minima separated by a
local maximum, i.e., $E'$ has at least three real roots.
The number of distinct real roots of $E'$ is locally constant
on each connected component of the complement of the
\emph{bifurcation surface}
$\mathcal{B} := \{\lambda \in \R^{2r} : \lambda_{2r} > 0,\;
\mathrm{disc}(E') = 0\}$,
where $\mathrm{disc}(E')$ is the discriminant of $E'$.
It vanishes exactly when $E'$ has a repeated root.
Since $\mathrm{disc}(E')$ is a polynomial in $\lambda$,
$\mathcal{B}$ is a codimension-$1$ algebraic variety.
Thus any continuous path in $\lambda$-space from a unimodal to a
bimodal configuration has to cross $\mathcal{B}$.

\textbf{Two-layer Stein closure and its sensitivity.}
For the sensitivity calculation, apply the Stein identity
(Proposition~\ref{prop:stein}) to the test function $x^k$. For each
$k\ge 1$,
\begin{equation}\label{eq:stein-1d}
  \textstyle\sum_{j=1}^{2r} j\lambda_j\, m_{k+j-1}
  = k\, m_{k-1}.
\end{equation}
If moments are truncated at degree $K = 2r{-}2$, the unclosed moments
$m_{K+1}$ and $m_{K+2}$ are resolved in two layers.
Layer~1 sets $k = K{-}2r{+}2$ in~\eqref{eq:stein-1d}
and isolates the leading term $2r\lambda_{2r}\,m_{K+1}$:
\begin{equation}\label{eq:layer1}
  m_{K+1}
  = \frac{(K{-}2r{+}2)\, m_{K-2r+1}
         - \sum_{j=1}^{2r-1}j\lambda_j\,m_{K-2r+j+1}}
         {2r\,\lambda_{2r}}.
\end{equation}
Layer~2 sets $k = K{-}2r{+}3$ and isolates $m_{K+2}$.
The resulting equation contains the layer~1 quantity $m_{K+1}$ with
coefficient $(2r{-}1)\lambda_{2r-1}$.
Substituting~\eqref{eq:layer1} yields the composed map
$\widetilde{m}_{K+2}(\lambda, m_{\le K})$ with sensitivity
\begin{equation}\label{eq:two-layer-sens}
  \frac{\partial \widetilde{m}_{K+2}}{\partial m_l}
  = O\!\Bigl(\frac{1}{\lambda_{2r}}\Bigr)
  + O\!\Bigl(\frac{\lambda_{2r-1}}{\lambda_{2r}^2}\Bigr),
  \qquad l \le K.
\end{equation}
The second term is the fragile one: layer~2 inherits the
$1/\lambda_{2r}$ amplification from layer~1 and then divides by
$\lambda_{2r}$ again.
For a single-layer closure ($\bar{d} = 1$), only the first
term is present and remains bounded as long as $\lambda_{2r} > 0$.
For two layers ($\bar{d} = 2$), the cross-term
$\lambda_{2r-1}/\lambda_{2r}^2$ can grow without bound
as $\lambda(t)$ evolves toward $\mathcal{B}$.

\textbf{Stiffness and spurious finite-time blowup.}
Specializing the separable double-well moment equation
\eqref{eq:double-well-moment-ode} to one coordinate gives the cubic
double-well
$dX = (X - X^3)\,dt + \sigma\,dW$ ($\bar{d} = 2$),
for which Dynkin's formula gives
\begin{equation}\label{eq:dw-moment-ode}
  \dot{m}_K = Km_K - K\,\widetilde{m}_{K+2}(\lambda, m_{\le K})
            + \frac{\sigma^2}{2}K(K{-}1)\,m_{K-2}.
\end{equation}
The Jacobian $J = \partial\dot{m}/\partial m$
inherits the $O(\lambda_{2r-1}/\lambda_{2r}^2)$
entries from~\eqref{eq:two-layer-sens} through the
$-K\,\widetilde{m}_{K+2}$ coupling.
As the density evolves from unimodal toward bimodal, the score-matching
fit tracks a path $\lambda(t)$ moving toward $\mathcal{B}$, and the
spectral radius $\rho(J)$ grows rapidly.
Once $\rho(J)$ exceeds the stability boundary of the integrator,
the numerical solution diverges.
The observed divergence time $t^*$ is mainly determined by the path of
$\lambda(t)$ in parameter space, not by the step size $\Delta t$. It is important to note that this is
the signature of a singularity in the closed moment ODE rather than an
ordinary time-discretization artifact.
The true density does not have this singularity. The additive noise makes
the Fokker--Planck equation uniformly parabolic, and the cubic drift is
dissipative at infinity, so standard Lyapunov estimates give a
non-explosive solution with finite polynomial moments for all finite
times. The blowup is therefore an artifact of the closure, not a
physical singularity.

\textbf{Numerical evidence.}
On the 2D double-well
(Appendix~\ref{app:double-well-generator},
Figure~\ref{fig:double-well}),
the closed moment ODE diverges at $t^* \approx 1.1$~s.
Refining $\Delta t$ by $16\times$ shifts $t^*$ by less than $5\%$,
which rules out a plain integration error as the main explanation.
Increasing the truncation order accelerates the blowup
($t^*{=}2.9$ for $r{=}4$, $t^*{=}0.54$ for $r{=}12$),
because higher $K$ exposes more eigendirections of $J$ to
the $1/\lambda_{2r}^2$ amplification.
The score matching Gram matrix remains well-conditioned
($\kappa(A) \approx 1$) throughout,
confirming the instability is in the closure.

\textbf{Universality.}
This bifurcation picture is not special to the linear parameterization
$E=\lambda^\top\phi$. In a smooth finite-dimensional family
$p(x;\theta)$, unimodal and bimodal regions are generically separated by
a codimension-$1$ boundary in parameter space, which is the usual
singularity-theory picture~\cite{thom1975}.
The specific amplification mechanism
($1/\lambda_{2r}^2$ from two-layer Stein closure)
is particular to the polynomial exponential family. The broader
geometric obstruction is that a single component has to pass through a
root-merging event in order to change modality.

\textbf{Possible resolutions.}
The same example also points to a useful path forward. The issue is not
that moments stop carrying information in highly nonlinear systems.
Rather, the exact Stein backsolve can be too rigid as a closure rule.
This makes it natural to treat closure as a small optimal control problem in
moment space. We only borrow standard control ingredients here:
linear-quadratic
optimal control for local feedback design~\cite{kalman1960optimal},
constrained MPC for receding-horizon optimization~\cite{mayne2000mpc},
and CBF-QPs for enforcing forward-invariance constraints through input
inequalities~\cite{ames2017cbf}. Note that in our setting the input is not a
physical actuator but is the unresolved moment, chosen to keep the
retained moment trajectory realizable and stable.

The resulting closure architecture has three parts. First, choose a
retained moment vector $x$ and treat the first unresolved moment block as
an input $u$, so the truncated Dynkin hierarchy has the control-affine
form
\begin{equation}\label{eq:closure-control-template}
  \dot{x}=F(x)+G(x)u.
\end{equation}
Second, restrict $u$ by moment-realizability constraints, such as Hankel
positive semidefiniteness. Third, use barrier inequalities to keep the
retained trajectory inside the feasible moment cone. If these constraints
leave a small amount of freedom, that freedom must be fixed by
system-specific structure such as the invariant law. This is not a
free lunch: the control architecture must know something about the
system, just as the higher-order moments in the moment ODE carry
physical information about that system. The point is that such structure can be analytic or
physics-informed, and need not come from particles. In the example
below, the anchor is the known invariant density of the one-dimensional
double well. Specializing the separable generator
\eqref{eq:double-well-generator} to one coordinate, consider the factor
\[
  dX = (X-X^3)\,dt + \sigma\,dW,
  \qquad \sigma=0.5,
\]
let $s_p=\mathbb{E}[X^{2p}]$ denote the even moments, set the retained
state to $x=(s_1,s_2,s_3)$, and use the unclosed eighth moment
$u=s_4=m_8$ as the control input. The retained
hierarchy is affine in $u$:
\begin{equation}\label{eq:dw-control-moments}
  \dot{s}_p =
  2p\,s_p - 2p\,s_{p+1}
  + \sigma^2 p(2p{-}1)s_{p-1},
  \qquad p=1,2,3.
\end{equation}
Equivalently, with $s_0=1$ absorbed into the constant term, this is the
standard affine control system
\begin{equation}\label{eq:dw-affine-control}
  \begin{gathered}
  \dot{x}=Ax+Bu+c,\\
  A=
  \begin{bmatrix}
    2 & -2 & 0\\
    6\sigma^2 & 4 & -4\\
    0 & 15\sigma^2 & 6
  \end{bmatrix},
  \quad
  B=
  \begin{bmatrix}
    0\\ 0\\ -6
  \end{bmatrix},
  \quad
  c=
  \begin{bmatrix}
    \sigma^2\\ 0\\ 0
  \end{bmatrix}.
  \end{gathered}
\end{equation}
The pair $(A,B)$ is controllable, since
$\operatorname{rank}[B,AB,A^2B]=3$. Thus, before imposing realizability
constraints, the unresolved moment can move the retained moment
dynamics. The only unusual part is what the input means: $u=m_8$ is
supplied by the closure, not by a physical actuator.
Instead of backsolving a second Stein layer, we require the next moment
to keep the truncated Stieltjes sequence feasible. With $Y=X^2$,
the even moments of $X$ become moments of a nonnegative variable:
$(1,s_1,s_2,s_3,u)=(\mathbb{E}[Y^0],\ldots,\mathbb{E}[Y^4])$.
We therefore impose the Hankel condition
\begin{equation}\label{eq:dw-hankel-condition}
  \begin{bmatrix}
    1 & s_1 & s_2\\
    s_1 & s_2 & s_3\\
    s_2 & s_3 & u
  \end{bmatrix}\succeq 0
\end{equation}
which gives the Schur-complement Hankel lower bound $u\ge L_H(x)$.
For the lower-order minor, we use the moment-cone barrier
$g(x)=s_1s_3-s_2^2$ and impose the control-barrier inequality
\begin{equation}\label{eq:dw-cbf-condition}
  \dot{g}(x,u)+\kappa_{\rm cbf} g(x)\ge 0,
  \qquad \kappa_{\rm cbf}>0,
\end{equation}
where $\kappa_{\rm cbf}$ is the barrier rate and
$\dot g(x,u)=\nabla g(x)\cdot(Ax+Bu+c)$. This gives an upper bound
$u\le U_{\rm CBF}(x)$, the CBF upper bound. We then use the one-step
constrained closure
\begin{equation}\label{eq:cone-cbf-closure}
  u(x)=L_H(x)+\theta_\infty(\kappa_{\rm cbf})
  \bigl(U_{\rm CBF}(x)-L_H(x)\bigr).
\end{equation}
This last scalar is the system-specific part of the construction:
$\theta_\infty(\kappa_{\rm cbf})\in[0,1]$, the stationary calibration
fraction, is calibrated from the known invariant density, not from
rollout samples. Thus the closure is physics-informed rather than
sample-informed. The one-dimensional double well has
$p_\infty(x)\propto\exp((x^2-x^4/2)/\sigma^2)$, so we compute its
stationary moments by deterministic quadrature and choose
$\theta_\infty$ so that \eqref{eq:cone-cbf-closure} is exact at that
stationary point. For $\kappa_{\rm cbf}=6$, this gives
$\theta_\infty=0.3496$.

Figure~\ref{fig:double-well-control-resolution} shows the closure
tracking the MC reference moments through $T=3$ without using MC
moments in the closure.
Figure~\ref{fig:double-well-control-evolution} reconstructs the 2D
density at $t=0,1,2,3$ from the same particle-free moments and recovers
the transition to the four-modal structure.

This is not a general solution to the two-layer closure problem yet, but
it is concrete evidence that the modality-transition failure is not fatal
to a moment and score-matching filter. Score matching can still supply
the density representation. The closure step just needs a stabilizing
layer that respects the deterministic moment cone. The lesson is to treat
Stein consistency as a soft preference or residual, while realizability
and barrier inequalities define the admissible closure.

\begin{figure}[H]
\centering
\includegraphics[width=\linewidth]{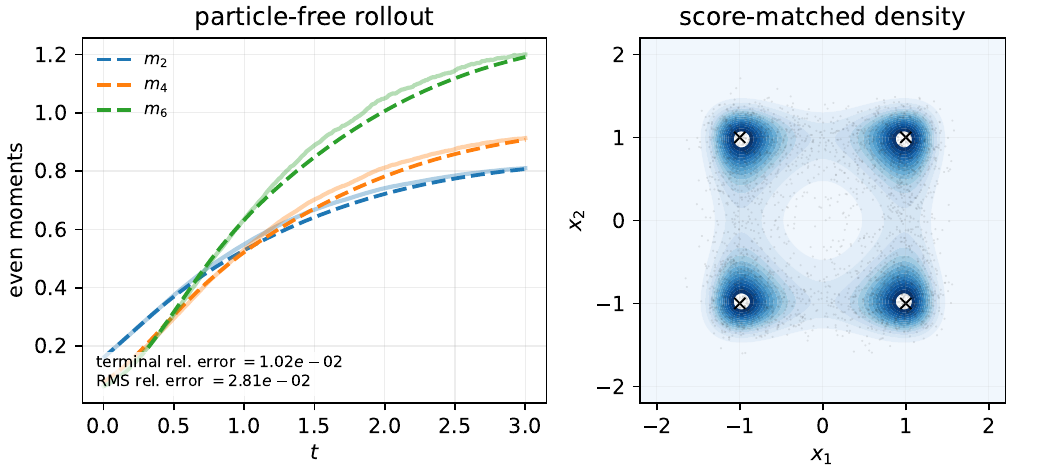}
\caption{Particle-free constrained closure on the double-well test case.
Left: even moments of the one-dimensional factor through $T=3$.
Solid faint curves are MC evaluation moments. Dashed curves are the
cone-CBF closure from~\eqref{eq:cone-cbf-closure}, which does not use MC
moments. Right: score-matched density for the separable 2D double well,
obtained from the product of the two particle-free marginal
reconstructions. MC samples are shown only as a reference.
The rollout has terminal relative error $1.0\times 10^{-2}$ and RMS
relative trajectory error $2.8\times 10^{-2}$ for $(m_2,m_4,m_6)$.}
\label{fig:double-well-control-resolution}
\end{figure}

\begin{figure}[H]
\centering
\includegraphics[width=\linewidth]{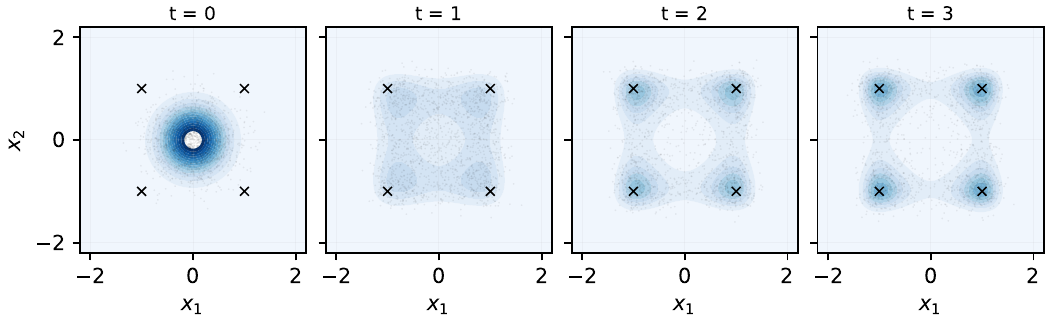}
\caption{Particle-free density evolution on the separable double-well
test case. Each panel shows the 2D score-matched density reconstructed
from the propagated moment sequence at the indicated time. MC samples
are shown only as a visual reference. The closure captures the transition
from the initial central mass to the four-well structure through $T=3$
without using particle moments.}
\label{fig:double-well-control-evolution}
\end{figure}


\section{Broader impacts}
\label{app:broader-impacts}

The Score Kalman Filter is a general-purpose nonlinear estimation
algorithm for stochastic dynamical systems. Its primary positive impact
is computational. By replacing the partition-function evaluation of
maximum-entropy density reconstruction with a linear solve, the SKF
makes high-order moment-based filtering tractable on commodity
hardware, which broadens access to uncertainty quantification
for applications that have until now relied on Gaussian-only or
particle-based filters. Likely beneficiaries are robotics and autonomous
systems, scientific data assimilation, and downstream tasks in
estimation-aware control where calibrated higher moments matter.

As with any improvement to nonlinear state estimation, the same
machinery could be deployed in surveillance or tracking pipelines whose
societal use is contested. We do not view these dual-use concerns as
unique to our method, since they apply broadly to the nonlinear
filtering literature, but we acknowledge them as a negative
impact. The paper releases no pretrained models, datasets, or other
artifacts with elevated risk of misuse, and all experiments use
synthetic dynamical systems.

\FloatBarrier

\end{document}